\documentclass[acmsmall,screen,nonacm]{acmart}

\AtBeginDocument{%
  }

\setcopyright{none}
\usepackage{amsthm, amsmath, stmaryrd, mathtools, url}
\usepackage{proof}
\usepackage{hyperref}
\usepackage{cleveref}
\usepackage[appendix=append]{apxproof}
\usepackage{tikz}
\usepackage{subfiles}
\usetikzlibrary{cd,decorations.pathreplacing}
\tikzset{purecomp/.style={rectangle, fill=green!20, rounded corners}}
\tikzset{comp/.style={rectangle, fill=magenta!10}}
\tikzset{del/.style={fill=blue}}

\allowdisplaybreaks[4]

\ifdefined\ACM@origsection
    \AtBeginDocument{\let\ACM@origsection\section}
\fi
\newtheorem{theorem}{Theorem}[section]
\newtheorem{lemma}[theorem]{Lemma}
\newtheorem{proposition}[theorem]{Proposition}
\newtheorem{corollary}[theorem]{Corollary}
\theoremstyle{acmdefinition}
\newtheorem{definition}[theorem]{Definition}
\newtheorem{notation}[theorem]{Notation}
\theoremstyle{acmplain}
\newtheoremrep{theoremapxproof}[theorem]{Theorem}
\newtheoremrep{lemmaapxproof}[theorem]{Lemma}
\newtheoremrep{propositionapxproof}[theorem]{Proposition}
\newcommand{\itemref}[1]{(\ref{#1})}

\definecolor{typeCbgColor}{rgb}{0.94,0.945,0.955}
\definecolor{typeDbgColor}{rgb}{0.955,0.945,0.94}
\newcommand{\typeAbgMath}[1]{\colorbox{orange!12}{$#1$}}
\newcommand{\typeAbg}[1]{\colorbox{orange!12}{#1}}
\newcommand{\typeBbgMath}[1]{\colorbox{cyan!6}{$#1$}}
\newcommand{\typeBbg}[1]{\colorbox{cyan!6}{#1}}
\newcommand{\typeCbgMath}[1]{\colorbox{typeCbgColor}{$#1$}}
\newcommand{\typeCbg}[1]{\colorbox{typeCbgColor}{#1}}
\newcommand{\typeDbgMath}[1]{\colorbox{typeDbgColor}{$#1$}}
\newcommand{\typeDbg}[1]{\colorbox{typeDbgColor}{#1}}

\newcommand{\blank}{{-}}

\newcommand{\defeq}{\mathrel{\coloneq}}

\newcommand{\Real}{\mathbb{R}}
\newcommand{\Nat}{\mathbb{N}}

\newcommand{\diag}{\delta}
\newcommand{\sym}{\mathrm{sym}}
\newcommand{\terminal}{{!}}
\newcommand{\assoc}{\mathrm{assoc}}

\newcommand{\expto}{\Rightarrow}

\newcommand{\bnfeq}{::=}
\newcommand{\BType}{\mathrm{BType}}
\newcommand{\Type}{\mathrm{Type}}
\newcommand{\typprod}{{\textstyle\prod}}
\newcommand{\typreal}{\mathtt{Real}}
\newcommand{\fromvec}[1]{\underline{#1}}
\newcommand{\tuple}[1]{\langle #1 \rangle}
\newcommand{\tproj}{\mathtt{proj}}
\newcommand{\tlet}{\mathtt{let}}
\newcommand{\tbind}{\leftarrow}
\newcommand{\tin}{\mathtt{in}}
\newcommand{\trd}{\mathtt{rd}}

\newcommand{\pret}[1]{\lfloor #1 \rfloor}
\newcommand{\plet}{\mathtt{let}}
\newcommand{\pbind}{\Leftarrow}
\newcommand{\pin}{\mathtt{in}}

\newcommand{\prevhandle}{\mathtt{rev\ handle}}
\newcommand{\phandle}{\mathtt{handle}}
\newcommand{\pwith}{\mathtt{with}}

\newcommand{\oprcolor}[1]{\textcolor{magenta}{#1}}
\newcommand{\funccolor}[1]{\textcolor{blue}{#1}}
\newcommand{\func}{\funccolor{\mathsf{f}}}
\newcommand{\const}{\funccolor{\mathsf{c}}}
\newcommand{\rd}{\mathsf{rd}}
\newcommand{\signatfunc}{\Sigma_{\mathrm{fun}}}
\newcommand{\signatopr}{\Sigma_{\mathrm{op}}}
\newcommand{\opr}{\oprcolor{\mathsf{Op}}}
\newcommand{\fwd}{\mathrm{f}}
\newcommand{\bwd}{\mathrm{b}}

\newcommand{\emptyenv}{\diamond}
\newcommand{\opto}{\rightarrowtriangle}
\newcommand{\asemicolon}{\mathrel{\fatsemi}}
\newcommand{\revhandlerjudgment}{\vdash^{\mathrm{r}}}

\newcommand{\tytvar}{\textsc{T-Var}}
\newcommand{\tytconst}{\textsc{T-Const}}
\newcommand{\tytfunc}{\textsc{T-Func}}

\newcommand{\tytplus}{\textsc{T-Plus}}
\newcommand{\tyttuple}{\textsc{T-Tuple}}
\newcommand{\tytproj}{\textsc{T-Proj}}
\newcommand{\tytlet}{\textsc{T-Let}}
\newcommand{\tytrd}{\textsc{T-RD}}
\newcommand{\tycpure}{\textsc{T-Pure}}
\newcommand{\tyclet}{\textsc{T-CLet}}
\newcommand{\tycop}{\textsc{T-Op}}
\newcommand{\tychandle}{\textsc{T-RHandle}}
\newcommand{\tyhandler}{\textsc{T-RHandler}}

\newcommand{\dsize}[1]{\left| {#1} \right|}
\newcommand{\rhdepth}[1]{\mathrm{depth}(#1)}
\newcommand{\opnum}[1]{\#_{\signatopr}(#1)}

\newcommand{\Val}{\mathrm{Val}}
\newcommand{\Eval}{\mathrm{Ev}}
\newcommand{\redto}{\rightarrow}

\newcommand{\rewriterd}[2]{{#1}.\mathcal{R}_{#2}}
\newcommand{\rewriterevh}[2]{\mathcal{RH}^{#1}_{#2}}
\newcommand{\ctxe}{\mathcal{E}}
\newcommand{\ctxfc}{\mathcal{F}^{c}}

\newcommand{\ctxft}{\mathcal{F}^{t}}

\newcommand{\catc}{\mathbb{C}}
\newcommand{\catd}{\mathbb{D}}
\newcommand{\cate}{\mathbb{E}}
\newcommand{\Sets}{\mathbf{Set}}
\newcommand{\Prof}{\mathbf{Prof}}
\newcommand{\Smooth}{\mathbf{Smooth}}

\newcommand{\obj}{\mathop{\mathrm{Ob}}}
\newcommand{\opposite}{\mathrm{op}}
\newcommand{\idmor}{\mathrm{id}}
\newcommand{\idprof}{\mathrm{I}}
\def\profto{\mathrel{\mkern3mu\vcenter{\hbox{$\scriptscriptstyle+$}}\mkern-12mu{\to}}}
\newcommand{\semicolon}{\mathrel{;}}
\newcommand{\natto}{\Rightarrow}

\newcommand{\catx}{\mathbb{X}}
\newcommand{\RD}{\mathsf{R}}
\newcommand{\axiomRD}[1]{[\textbf{RD.{#1}}]}

\newcommand{\dif}{\mathrm{d}}

\newcommand{\monad}{\mathcal{T}}
\newcommand{\promonad}{\mathcal{A}}
\newcommand{\strength}{\mathrm{st}}
\newcommand{\musemicolon}{\mathrel{;}^{\mu}}
\newcommand{\ATerm}{\mathrm{Arr}}
\newcommand{\apure}{\mathrm{arr}}
\newcommand{\acomp}{\mathrel{>\!\!>\!\!>}}
\newcommand{\astr}{\mathrm{first}}

\newcommand{\ralg}[1]{\langle #1 \rangle}

\newcommand{\itp}[1]{\llbracket #1 \rrbracket}
\newcommand{\bigitp}[1]{\left\llbracket #1 \right\rrbracket}

\newcommand{\logrelt}{\triangleleft}
\newcommand{\logrelc}{\blacktriangleleft}

\newcommand{\transpose}{{\top}}
\newcommand{\Jacobian}{\mathsf{D}}

\newcommand{\learningrate}{\alpha}
\newcommand{\Loss}{\mathcal{L}}
\newcommand{\sigmoid}{\mathrm{sigm}}
\newcommand{\Locations}{\mathit{Loc}}
\newcommand{\loctyp}[2]{#1 :: #2}
\newcommand{\heap}{\mathit{hp}}

\newcommand{\inp}{\mathrm{in}}
\newcommand{\hid}{\mathrm{hid}}
\newcommand{\out}{\mathrm{out}}
\newcommand{\oprlayer}       {\oprcolor{\mathsf{Layer}}}
\newcommand{\oprlinear}      {\oprcolor{\mathsf{FC}}}
\newcommand{\oprencodedecode}{\oprcolor{\mathsf{EnDe}}}
\newcommand{\oprconv}        {\oprcolor{\mathsf{Conv1D}}}
\newcommand{\oprpool}        {\oprcolor{\mathsf{MaxPool1D}}}
\newcommand{\oprget}         {\oprcolor{\mathsf{Get}}}
\newcommand{\oprput}         {\oprcolor{\mathsf{Put}}}
\newcommand{\opru}           {\oprcolor{\mathsf{U}}}
\newcommand{\opFC}[1]{\oprlinear\oprcolor{[}#1\oprcolor{]}}
\newcommand{\opGet}[1]{\oprget\oprcolor{[}#1\oprcolor{]}}
\newcommand{\opPut}[1]{\oprput\oprcolor{[}#1\oprcolor{]}}
\newcommand{\opEnDe}[1]{\oprencodedecode\oprcolor{[}#1\oprcolor{]}}
\newcommand{\opConv}[1]{\oprconv\oprcolor{[}#1\oprcolor{]}}
\newcommand{\opPool}[1]{\oprpool\oprcolor{[}#1\oprcolor{]}}
\newcommand{\opU}[1]{\opru\oprcolor{[}#1\oprcolor{]}}
\newcommand{\fnSwish}    {\funccolor{\mathsf{swish}}}
\newcommand{\fnscalarmul}{\mathbin{\funccolor{\cdot}}}
\newcommand{\fnminus}    {\mathbin{\funccolor{-}}}
\newcommand{\fnmatmul}   {\mathbin{\funccolor{\ast}}}
\newcommand{\fntranspose}{{\funccolor{\top}}}
\newcommand{\fnconcat}   {\funccolor{\mathsf{concat}}}
\newcommand{\fnupscale}  {\funccolor{\mathsf{upscale}}}
\newcommand{\fnpadding}  {\funccolor{\mathsf{padding}}}
\newcommand{\fnconv}     {\funccolor{\mathsf{conv1d}}}
\newcommand{\fnpool}     {\funccolor{\mathsf{maxpool1d}}}
\newcommand{\fnround}    {\funccolor{\mathsf{round}}}

\newcommand{\MLP}{\mathrm{MLP}}
\newcommand{\ResNet}{\mathrm{ResNet}}
\newcommand{\autoencoder}{\mathrm{AE}}
\newcommand{\CNN}{\mathrm{CNN}}
\newcommand{\Unet}{\mathrm{Unet}}

\newcommand{\restr}[1]{\overline{#1}}
\newcommand{\PartSmooth}{\mathbf{PSmooth}}

\newcommand{\partto}{\rightharpoonup}
\begin{document}

\title[Programming Backpropagation with Reverse Handlers for Arrows]{Programming Backpropagation with \\ Reverse Handlers for Arrows}

\author{Takahiro Sanada}
\orcid{0000-0003-3409-6963}
\affiliation{%
  \institution{Fukui Prefectural University}
  \city{Eiheiji}
  \country{Japan}
}
\email{tsanada@fpu.ac.jp}

\author{Keisuke Hoshino}
\orcid{0009-0007-7287-7572}
\affiliation{%
  \institution{Kyoto University}
  \city{Kyoto}
  \country{Japan}
}
\email{hoshinok@kurims.kyoto-u.ac.jp}

\author{Kenshin Hirai}
\orcid{0009-0008-1272-7684}
\affiliation{%
  \institution{Kyoto University}
  \city{Kyoto}
  \country{Japan}
}
\email{iyflken@kurims.kyoto-u.ac.jp}

\author{Shin-ya Katsumata}
\orcid{0000-0001-7529-5489}
\affiliation{%
  \institution{Kyoto Sangyo University}
  \city{Kyoto}
  \country{Japan}
}
\email{s.katsumata@cc.kyoto-su.ac.jp}

%\author{Anonymous Author(s)}


\begin{abstract}
    %We introduce novel handlers for arrows, called \emph{reverse handlers}, which enable us to implement reverse-mode automatic differentiation.
    %The denotational semantics of reverse handlers is given by homomorphisms induced by algebras, constructed by a new method, of strong promonads on Cartesian reverse differential categories.
    %Our reverse handlers provide a good separation between abstract design and implementation of neural networks.
    %We present numerous examples of applying reverse handlers to various neural networks, including multilayer perceptrons, residual networks, autoencoders, convolutional neural networks, U-nets and straight-through estimator in quantization-aware training.
    %Furthermore, we observe that string diagrams of strong promonads on a Cartesian reverse differential category is suitable as a formal graphical language for neural networks.
    %Hence, our framework provides a graphical syntax and categorical semantics of neural networks.

    % new introduction
    We introduce a new programming language and its categorical semantics in order to design and implement neural networks within the framework of algebraic effects and handlers for arrows.
    Our language enables us to construct neural networks symbolically, in the same manner as algebraic effects, and to assign implementations---such as backpropagation computations---to them via handlers.
    The advantage of this language design is that network descriptions become abstract and high-level, while implementations can be flexibly assigned to networks.
    We establish a rigorous foundation for our language by developing a type system, an operational semantics, a categorical semantics, and soundness and adequacy theorems.
%    The technical core of our development is the introduction of \emph{reverse handlers}, which are novel handler mechanisms for arrows, and
%    a new construction method of algebras of strong promonads on reverse differential restriction categories (RDRC),
%    which provide the semantical foundation of reverse handlers and the formal graphical syntax and semantics of neural networks.
    The technical core is the introduction of \emph{reverse handlers}, a novel handler mechanism for arrows for implementing backpropagation,
    together with new algebras of strong promonads on reverse differential restriction categories (RDRCs),
    whose string diagrams provide a formal graphical syntax and semantics for neural networks.
\end{abstract}


  
  
  
\maketitle
\section{Introduction}

One of the central tasks in developing machine learning systems
is implementing neural networks.
%\inred{is the implementation of a neural network → is implementing neural networks}.
Such implementations construct a
large network of parameterized, differentiable functions, either
declaratively or imperatively. The resulting network is used for
inference, while its gradient is employed in various gradient-based
optimization methods to search for parameters that fit real data.
Programming libraries and environments that support this
workflow typically provide mechanisms for (1) constructing such
networks and (2) computing their gradients (i.e., reverse-mode
automatic differentiation of networks) for backpropagation \cite{Linnainmaa1976,RumelhartHintonWilliams1986}.

In this paper, we introduce a new programming language that provides
facilities (1) and (2) within the framework of algebraic effects
\cite{PlotkinPower2001} and handlers \cite{PlotkinPretnar2013}, which
we briefly recall here. Handlers for algebraic effects provide
mechanisms for (A) constructing abstract computational effects and (B)
executing abstract effects by handlers. This framework revolutionizes
the way effectful programs are written: within a single language, one
can abstractly describe computational effects, and flexibly program
their concrete execution behavior by handlers.

%The main observation of this paper is that there is a striking analogy
%between mechanisms (1) and (2) provided by neural network libraries
%and environments, and mechanisms (A) and (B) provided by handlers for
%algebraic effects.
The main observation of this paper is that there is a striking analogy
between the mechanisms (1) and (2) in neural network libraries and environments,
and the mechanisms (A) and (B) in handlers for algebraic effects.
Building on this observation, we design our
programming language so that neural networks can be constructed
symbolically, analogously to algebraic effects, and
backpropagation computations can be assigned to them via handlers.
The advantage of this design is that network descriptions become abstract
and high-level, while backpropagation computations can be flexibly
assigned to these networks, including symbolic ones.
This stands in
sharp contrast to existing neural network libraries and environments,
in which networks and their backpropagations have fixed semantics.
For
instance, we exploit this flexibility to implement the
straight-through estimator (STE) \cite{BengioLeonardCourville2013} in
quantization-aware training (QAT), where a non-standard
backpropagation computation is assigned to the forward computation.

Designing the programming language poses a few technical challenges.
\paragraph{Challenge 1} One
essential feature of neural network programming is the ability to
connect multiple layers of networks.
However, the standard calculus with handlers and its monadic semantics do not naturally support this style of
programming, because the computational effect constructed by a program
$P\colon X \to \monad Y$ is a term tree over $Y$, whose structure differs
from that of neural networks. To overcome this mismatch, we employ the
arrow calculus \cite{LindleyWadlerYallop2010} extended with algebraic
operations and handlers \cite{Sanada2024}. %, together with its semantic interpretation via promonads.
Arrows are explicitly parameterized by
input and output types, like neural networks, and a program
$P\in\promonad(X,Y)$ modeled in an arrow represents computational
effects as chains of primitive arrows, closely resembling the layered
structure of neural networks.  Furthermore, employing arrows creates
room for naturally representing bidirectional computations, which
helps resolve the problem described below.

\paragraph{Challenge 2} The main aim of our language is to assign backpropagation
  computations to networks via handlers. Unfortunately, the notion of
  handlers in \cite{Sanada2024} cannot serve this purpose, as it
  allows only unidirectional computations, whereas backpropagation is
  inherently bidirectional, similar to the computation of lenses
  \cite{FongJohnson2019}. To address this issue, we introduce a novel
  handler mechanism, called the \emph{reverse handler}, into the arrow
  calculus. The bidirectional computation of backpropagation can be
  naturally modeled thanks to the generality of arrows.
  
\paragraph{Challenge 3} We also require that the language be able to express standard
  backpropagation as reverse‑mode differentiation of programs. For
  this purpose, we adopt the reverse-mode derivative operator
  introduced by Abadi and Plotkin \cite{AbadiPlotkin2019}.

  The resulting language is therefore a combination of the arrow
  calculus equipped with reverse handlers and the reverse-mode
  derivative operator. We establish a rigorous foundation for this
  language by developing a type system, an operational semantics, a
  categorical semantics, and soundness and adequacy theorems. For interpreting the
  reverse-mode derivative operator, we employ Cockett et al.'s
  \emph{Cartesian reverse differential categories} (CRDC) and Cruttwell et al.'s \emph{reverse differential restriction categories} (RDRC), which
  provide a well-established semantic framework for reverse
  differentiation
  \cite{Cockett+2020,Cruttwell+2021,Cruttwell+2022}. We interpret
  arrows by \emph{strong promonads}
  \cite{HeunenJacobs06,JacobsHeunenHasuo2009,Asada2010,DBLP:journals/entcs/Atkey11,Sanada2024},
  and handlers for arrows by {\em promonad algebras}
  \cite{Sanada2024}. In this paper we interpret reverse handlers using
  a specialized promonad algebra that resides on RDRCs.

  We end this
  section by listing the highlights of our development.
\begin{enumerate}
\item We formally design the language for reverse handlers: 
  syntax, type system, operational semantics, and denotational semantics. Soundness and
  adequacy of the denotational semantics with respect to the operational semantics are proved.
\item We give a new method for constructing algebras of strong promonads on an RDRC. The
  construction provides the semantical foundation of reverse handlers.
\item We present numerous examples of applying reverse handlers to neural networks, including
  multilayer perceptrons, residual networks, autoencoders, convolutional neural networks
  and U-nets (Table~\ref{tab:examples-summary}). Using reverse handlers, we can separate abstract designs of neural networks
  from implementations of detailed behavior. Furthermore, using reverse handlers, we can
  implement quantization-aware training via the straight-through estimator \cite{BengioLeonardCourville2013}, which cannot be
  formalized within RDRCs alone.
\item Our examples of neural networks suggest that string diagrams of strong promonads on a
    RDRC provide a formal graphical language for neural networks.
  Hence, our framework provides a graphical syntax and categorical semantics of neural networks.
\end{enumerate}

% Our contributions are summarized as follows:
% \begin{enumerate}
%     \item We introduce a programming language with reverse handlers for arrows. We define its syntax, type system, operational semantics, and denotational semantics. Soundness and adequacy of the denotational semantics with respect to the operational semantics are proved.
%     \item We give a new construction method of algebras of strong promonads on a CRDC.
%     The construction provides the semantical foundation of reverse handlers.
%     \item We present numerous examples of applying reverse handlers to neural networks, including multilayer perceptrons, residual networks, autoencoders, convolutional neural networks and U-nets.
%     Using reverse handlers, we can separate abstract designs of neural networks from implementations of detailed behavior.
%     Furthermore, using reverse handlers, we can implement quantization-aware training via straight-through estimator \cite{BengioLeonardCourville2013}, which cannot be formalized within CRDCs alone.
%     \item Our examples of neural networks suggests that string diagrams of strong promonads on a CRDC is a formal graphical language for neural networks. Hence, our framework provides a graphical syntax and categorical semantics of neural networks.
% \end{enumerate}

\section{Overview}

%We briefly explain the ideas of our approach.
%Our starting point is a parallel between effect handlers and backpropagation.
%In an effect-handler setting, one specifies local programmable data (= handler) for basic operations, and the effect carried by a computation is realized as the continuation transformation induced compositionally from that local data.
%We view backpropagation in the same way:
%a neural network architecture is an abstract composition of layers, while its implementation is determined by specifying the backpropagation behavior of each layer.
%This suggests understanding layers as operations and layer-wise backpropagation data as handler data; in our setting, such data is specified by \emph{reverse handlers}.
%Categorically, reverse handlers provide the local data from which one constructs a global continuation transformation, captured by a promonad algebra structure and the associated homomorphism from the free algebra.

We briefly explain the ideas underlying our approach.
Our starting point is a parallel between effect handlers and backpropagation.

In the setting of effect handlers, a handler specifies concrete local behavior for each algebraic operation. The overall effect of a program is then realized as a continuation transformation compositionally induced from those local behaviors.
We view backpropagation in the same way. A neural network architecture is given as an abstract composition of layers, while its concrete implementation is determined by specifying the backpropagation behavior of each layer—for example, which memory locations parameters are retrieved from, or customized differentiation rules such as the straight-through estimator (STE).

From this perspective, we understand layers as algebraic operations, and regard layer-wise backpropagation data as corresponding to handler data. In our setting, such data is specified by \emph{reverse handlers}.

Categorically, reverse handlers provide local behaviors from which a global continuation transformation is constructed. This structure is captured by a promonad algebra structure together with the corresponding homomorphism induced from the free algebra.

\subsection{Abstract Description of an MLP by Algebraic Effects for Arrows}
Let us consider implementing a multilayer perceptron (MLP, Fig~\ref{fig:mlp}) with three layers, an input layer, a hidden layer and an output layer.
The size of the input layer, hidden layer and output layer are denoted by $n_{\inp}$, $n_{\hid}$ and $n_{\out}$, respectively.
\begin{figure}
\centering
\[
% an illustration of three-layer multilayer perceptron
\begin{tikzpicture}[scale=0.8, transform shape]
    \draw[color=magenta, thick, fill=magenta!10] (0+0.3, 1.7) rectangle (2-0.3, -0.2);
    \draw[color=magenta, thick, fill=magenta!10] (2+0.3, 1.7) rectangle (4-0.3, -0.2);
    \draw[color=gray!60,fill=gray!20] (-0.2, 1.7) rectangle (0.2, -0.2);
    \draw[color=gray!60,fill=gray!20] (2-0.2, 1.7) rectangle (2+0.2, -0.2);
    \draw[color=gray!60,fill=gray!20] (4-0.2, 1.7) rectangle (4+0.2, -0.2);
    \draw (0, 1.5) -- (2, 1.5);
    \draw (0, 1)   -- (2, 1.5);
    \draw (0, 0)   -- (2, 1.5);
    \draw (0, 1.5) -- (2, 1);
    \draw (0, 1)   -- (2, 1);
    \draw (0, 0)   -- (2, 1);
    \draw (0, 1.5) -- (2, 0);
    \draw (0, 1)   -- (2, 0);
    \draw (0, 0)   -- (2, 0);
    \draw (2, 1.5) -- (4, 1.5);
    \draw (2, 1)   -- (4, 1.5);
    \draw (2, 0)   -- (4, 1.5);
    \draw (2, 1.5) -- (4, 1);
    \draw (2, 1)   -- (4, 1);
    \draw (2, 0)   -- (4, 1);
    \draw (2, 1.5) -- (4, 0);
    \draw (2, 1)   -- (4, 0);
    \draw (2, 0)   -- (4, 0);
    \node at (0, 2) {$n_{\inp}$};
    \node at (0, -0.4) {$x_0$};
    \draw[fill=white] (0, 1.5) circle (0.1);
    \draw[fill=white] (0,   1) circle (0.1);
    \node at (0, 0.6) {$\vdots$};
    \draw[fill=white] (0,   0) circle (0.1);
    \node at (2, 2) {$n_{\hid}$};
    \node at (2, -0.4) {$x_1$};
    \draw[fill=white] (2, 1.5) circle (0.1);
    \draw[fill=white] (2,   1) circle (0.1);
    \node at (2, 0.6) {$\vdots$};
    \draw[fill=white] (2,   0) circle (0.1);
    \node at (4, 2) {$n_{\out}$};
    \node at (4, -0.4) {$x_2$};
    \draw[fill=white] (4, 1.5) circle (0.1);
    \draw[fill=white] (4,   1) circle (0.1);
    \node at (4, 0.6) {$\vdots$};
    \draw[fill=white] (4,   0) circle (0.1);
    % input and output vectors
    \path ( -0.2,-0.2) -- node[sloped,above] {an input vector}  ( -0.2,1.7);
    \path (4+0.2,-0.2) -- node[sloped,below] {an output vector} (4+0.2,1.7);
    % fully connected layers
    \path (0+0.3,1.7) -- node[sloped,above] {\tiny fully connected}  (2-0.3,1.7);
    \node at (1, -0.4) {$w_0$};
    \path (2+0.3,1.7) -- node[sloped,above] {\tiny fully connected}  (4-0.3,1.7);
    \node at (3, -0.4) {$w_1$};
\end{tikzpicture}
\]
\Description{An illustration of an MLP with three layers.}
\caption{An illustration of an MLP with three layers.}
\label{fig:mlp}
\end{figure}
An input vector $x_0 \in \Real^{n_{\inp}}$ is fed into the input
layer, and then the hidden layer produces a vector
$x_1 = w_0 x_0 \in \Real^{\hid}$ where $w_0$ is an
$n_{\hid} \times n_{\inp}$ matrix, and finally the output layer
produces an output vector $x_2 = w_1 x_1 \in \Real^{n_{\out}}$ where
$w_1$ is an $n_{\out} \times n_{\hid}$ matrix.  We omit bias vectors
and activation functions for simplicity.

% We represent a connection of layers in a network as an algebraic operation for arrows.
% We give a high-level description of
% in the arrow calculus with operations \cite{Sanada2024}.
Before presenting a program example, let us briefly explain the data storage
model of our programming language. Matrices are stored in the heap,
like many machine learning libraries. They are retrieved and updated
through memory locations, ranged over by
$\ell,\ell_0,\ell_1,\cdots$. We also assume that a pair of natural
numbers $(n, m)$ is assigned to each location $\ell$.  We write
$\loctyp{\ell}{(n, m)}$ to indicate such an assignment.  The statement
$\loctyp{\ell}{(n,m)}$ means that the memory location $\ell$ contains
an $m \times n$ matrix.

\newcommand{\sanadacalc}{arrow calculus with operations and handlers}

In our programming language, which is based on the \sanadacalc{}
\cite{Sanada2024}, the MLP in Fig.~\ref{fig:mlp} is described by the
following program:
\begin{equation}
  \label{eq:abstract-mlp}
  P_{\MLP} \defeq
  \plet \; x_1 \pbind
  \opFC{\ell_0}(x_0) \; \pin \;
  \opFC{\ell_1}(x_1)
  \quad
  \quad
  (\loctyp{\ell_0}{(n_{\inp}, n_{\hid})},~
  \loctyp{\ell_1}{(n_{\hid}, n_{\out})})
\end{equation}
The constructor
$\opFC{\ell} : \typeAbgMath{\typreal(n)} \opto
\typeBbgMath{\typreal(m)}$ for $\loctyp{\ell}{(n,m)}$ is an {\em
  algebraic operation}---hereafter we call it \emph{operation} and write
it in red. It symbolically represents a fully connected layer referring
to the matrix stored at $\ell$. We note that it has
no concrete behavior yet.  The variable $x_0 : \typreal(n_{\inp})$ keeps
an input vector, % of dimension $n_{\inp}$,
and the whole expression (of type $\typreal(n_{\out})$) represents
the output vector of the MLP.  The above program thus abstractly
describes the structure of the MLP with three layers.
At this stage, the program specifies only the architecture through operations.
What remains is to supply, for each operation, concrete local behavior from which the behavior of the whole computation is induced.

\subsection{Handlers vs Reverse Handlers}

Our programming language is designed so that users can flexibly define
backpropagation algorithms by our novel mechanism of reverse handlers.
Before introducing reverse handlers, we first revisit the handlers in
the \sanadacalc{} \cite{Sanada2024}, and
discuss that there is a problem when directly employing it for
programming backpropagation.

% To implement the behavior of inference and backpropagation of the MLP, we use handlers for arrows.
% We can implement inference of the MLP by a usual handler for arrows \cite{Sanada2024} that implements each operation $\oprlinear^{\ell}_{nm}$ as a computation that performs matrix multiplications using the $n \times m$ parameter matrix stored in the location $\ell$.
% However, usual handlers for arrows are not sufficient to implement backpropagation of the MLP.
% Hence, we introduce a novel handler for arrows, called \emph{reverse handler}, to implement backpropagation of the MLP.

\subsubsection{Implementation of Inference}

\emph{Handlers} are a mechanism to give concrete implementations to
operations.
They are first introduced to flexibly program computational effects caused in programs \cite{PlotkinPretnar2013}.
Later they are reformulated in the context of arrows by Sanada \cite{Sanada2024}. 
In his \sanadacalc{}, the following handler syntax is provided to give a concrete implementation to operations invoked during the execution of the program surrounded by $\phandle$ and $\pwith$:
\[
\phandle \; \left(
    \plet \; x_1 \pbind \opFC{\ell_0}(x_0) \; \pin \;
    \opFC{\ell_1}(x_1)
\right) \; \pwith \; H.
\]
We write an implementation of each operation in the clause $H$.
Thus, a handler specifies locally how each operation acts on its continuation, and the overall inference behavior is induced compositionally from these clauses.
For
example, the clause implementing the operation $\opFC{\ell}$ (for
each $\loctyp{\ell}{(m,n)}$) representing a fully-connected layer with
the matrix multiplication is
\begin{equation*}
  \begin{aligned}
    H_{\mathrm{infer}} \defeq
    \{ x \mapsto \pret{x} \} \cup
    \{
    \opFC{\ell}, k \asemicolon x \mapsto
    \plet \; w \pbind \opGet{\ell}() \; \pin \; k \bullet (w \fnmatmul x)
    \}_{m,n\in\Nat,\loctyp{\ell}{(n,m)}}.
  \end{aligned}
\end{equation*}
The clause $H_{\mathrm{infer}}$ consists of two parts: the first part specifies the input continuation by simply returning the inner result, and the
second part assigns implementations to operations symbols.  What is
assigned to $\opFC{\ell}$ above is the program that retrieves the matrix stored in
the location $\ell$, multiplies it with the input vector $x$ and passes it
to the continuation.  In this program, $k$ holds the
continuation when an operation $\opFC{\ell}$ is invoked, and the
variable $x$ holds the argument passed to the operation $\opFC{\ell}$.
The operation $\opGet{\ell}()$ retrieves the parameter matrix stored
at $\ell$.  The term $k \bullet (w \fnmatmul x)$ represents the
application of the continuation $k$ to the matrix multiplication of
$w$ and $x$.  Thus the program
\[
\phandle \; \left(
    \plet \; x_1 \pbind \opFC{\ell_0}(x_0) \; \pin \;
    \opFC{\ell_1}(x_1)
\right) \; \pwith \; H_{\mathrm{infer}}
\]
returns a
vector $w_1 \fnmatmul (w_0 \fnmatmul x_{\inp})$ where $w_0$ and $w_1$
are the parameter matrices stored in locations $\ell_0$ and $\ell_1$,
respectively.
In this sense, $H_{\mathrm{infer}}$ supplies the local behavior of each layer, from which the inference behavior of the whole network is obtained via the constructor $\phandle$.

In the \sanadacalc{} \cite{Sanada2024}, the typing rules are similar to those of ordinary calculus with handlers (e.g.\ \cite{PlotkinPretnar2013}),
but differ in that the typing environment of effectful computation has two parts: one for pure context and the other for effectful context.
A typing judgment $\Gamma \asemicolon \Delta \vdash P : A$ means that an effectful program $P$ has type $A$ under the pure typing environment $\Gamma$ and the effectful typing environment $\Delta$.

The typing rule for handlers requires that the result type of the continuation $k$ is the same as the type of the entire clause.
\begin{equation}\label{eq:typing-rule-direct-handler}
    \infer{
    \vdash \{ x \mapsto P \} \cup \{ \opr, k \asemicolon x \mapsto Q_{\opr} \}_{\opr} : C \Rightarrow \typeDbg{$D$}
}{
    \emptyenv \asemicolon x : C \vdash P : \typeDbg{$D$}
    &
    \bigl(
        k : \typeBbg{$B$} \rightsquigarrow \typeDbg{$D$} \asemicolon x : \typeAbg{$A$} \vdash Q_{\opr} : \typeDbg{$D$}
    \bigr)_{\opr : \typeAbg{$\scriptstyle A$} \opto \typeBbg{$\scriptstyle B$}}
}
\end{equation}
The above typing rule claims that, for each operation $\opr : \typeAbgMath{A} \opto \typeBbgMath{B}$, the continuation $k$ has type $\typeBbgMath{B} \rightsquigarrow \typeDbgMath{D}$, which means a type of computation with input of type $\typeBbgMath{B}$ and output of type $\typeDbgMath{D}$, and the entire clause $Q_{\opr}$ takes an input of type $\typeAbgMath{A}$ and returns a result of type $\typeDbgMath{D}$.
The handler $H_{\mathrm{infer}}$ is well-typed in the existing framework of effect handlers for arrows.

\subsubsection{Implementation of Backpropagation}
Backpropagation is a training algorithm for neural networks.
To implement backpropagation, each layer $\opFC{\ell}$ is associated with the following computations:
\begin{enumerate}
\item Forward computation. Given an input vector $x$ from the previous layer, the parameter matrix $w$ is retrieved from the location $\ell$, and the matrix multiplication $w \fnmatmul x$ is computed. The result is then passed to the next layer.
\item Backward computation. Given a gradient vector $y$ from the next layer, the gradient $y \fnmatmul (x^{\fntranspose})$ with respect to the parameter matrix $w$ is computed. The parameter matrix $w$ is updated as $w \fnminus y \fnmatmul (x^{\fntranspose})$.
Finally, the gradient $w^{\fntranspose} \fnmatmul y$ is passed to the previous layer.
\end{enumerate}

To implement backpropagation with squared loss function $\Loss(x) = \sum_{i} (x_i - t_i)^2 / 2$, we would like to write a handler for the operation
$\opFC{\loctyp{\ell}{(n,m)}} : \typeAbgMath{\typreal(n)} \opto \typeBbgMath{\typreal(m)}$ like:
\begin{align*}
    H_{\mathrm{train}} & = \{ x \mapsto \pret{\fromvec{\learningrate} \fnscalarmul (x \fnminus \fromvec{t})} \} \cup
    \left\{
        \opFC{\ell}, k \asemicolon x \mapsto
        Q_{\opFC{\ell}}
    \right\},
    \\
    Q_{\opFC{\ell}} & =
    \left(
    \begin{aligned}
        & \plet \; w \pbind \opGet{\ell}() \; \pin \; \plet \; y \pbind k \bullet (w \fnmatmul x) \; \pin \; \\
        & \plet \; \_ \pbind \opPut{\ell}(w \fnminus y \fnmatmul (x^{\fntranspose})) \; \pin \; \pret{w^{\fntranspose} \fnmatmul y}
    \end{aligned}
    \right)
    \quad
    \begin{aligned}
        & \text{\textcolor{gray}{forward computation}} \\
        & \text{\textcolor{gray}{backward computation}}
    \end{aligned}
\end{align*}
where $\fromvec{\learningrate}$ is a constant symbol of type $\typreal(1)$ representing the learning rate $\learningrate \in \Real$, and $\fromvec{t}$ is a constant symbol of type $\typreal(n_{\out})$ representing the target output vector $t \in \Real^{n_{\out}}$.

The return clause $P \defeq \pret{\fromvec{\learningrate} \fnscalarmul (x \fnminus \fromvec{t})}$ in $H_{\mathrm{train}}$ has type $\typeCbgMath{\typreal(n_{\out})}$ and computes the gradient $\learningrate \frac{\partial \Loss}{\partial x} \in \Real^{n_{\out}}$ of the loss function $\Loss$ with respect to the output $x$ of the final layer of the MLP.
Each result $y$ of the continuation $k \bullet (w \fnmatmul x)$ is expected to have type $\typeBbgMath{\typreal(m)}$ and be the gradient $\learningrate \frac{\partial \Loss}{\partial x'} \in \Real^{m}$ of the loss function $\Loss$ with respect to the output $x' = w \fnmatmul x$ of forward propagation of the layer $\opFC{\ell}$.
The operation $\opPut{\ell}(w \fnminus y \fnmatmul (x^{\fntranspose}))$ updates the parameter matrix stored in the location $\ell$ to $w - \learningrate \frac{\partial \Loss}{\partial w}$ by gradient descent.
The term $\pret{w^{\fntranspose} \fnmatmul y}$ has type $\typeAbgMath{\typreal(n)}$ and returns the gradient $\learningrate \frac{\partial \Loss}{\partial x} \in \Real^n$ of the loss function $\Loss$ with respect to the input $x$ of the layer $\opFC{\ell}$.

However, the handler $H_{\mathrm{train}}$ is not well-typed under the typing rule \eqref{eq:typing-rule-direct-handler} because the result type $\typeBbgMath{\typreal(m)}$ of the continuation $k$ is different from the type $\typeAbgMath{\typreal(n)}$ of the entire handle clause $Q_{\opFC{\ell}}$.
The desired typing of the handler $H_{\mathrm{train}}$ is as follows:
\[
\footnotesize
\infer{
    \revhandlerjudgment \{ x \mapsto P \} \cup \left\{ \opFC{\ell}, k \asemicolon x \mapsto Q_{\opFC{\ell}} \right\}_{\opFC{\loctyp{\ell}{(n,m)}}} : \typeCbg{$\typreal(n_{\out})$}
}{
    \emptyenv \asemicolon x : \typeCbg{$\typreal(n_{\out})$} \vdash P : \typeCbg{$\typreal(n_{\out})$}
    &
    \left(
        k : \typeBbg{$\typreal(m)$} \rightsquigarrow \typeBbg{$\typreal(m)$} \asemicolon x : \typeAbg{$\typreal(n)$} \vdash Q_{\opFC{\ell}} : \typeAbg{$\typreal(n)$}
    \right)_{\opFC{\ell}}
    %\begin{aligned}
    %    & \emptyenv \asemicolon x : \typeCbg{$\typreal(n_{\out})$} \vdash P : \typeCbg{$\typreal(n_{\out})$}
    %    \\
    %    & \left(
    %        k : \typeBbg{$\typreal(m)$} \rightsquigarrow \typeBbg{$\typreal(m)$} \asemicolon x : \typeAbg{$\typreal(n)$} \vdash Q_{\opFC{\ell}} : \typeAbg{$\typreal(n)$}
    %    \right)_{\opFC{\loctyp{\ell}{(n,m)}}}
    %\end{aligned}
}
\]
Hence, in general, the typing rule for desired handlers for backpropagation should be as follows:
\[
\footnotesize
\infer{
    \revhandlerjudgment \{ x \mapsto P \} \cup \left\{ \opr, k \asemicolon x \mapsto Q_{\opr} \right\}_{\opr} : \typeCbg{$C$}
}{
    \emptyenv \asemicolon x : \typeCbg{$C$} \vdash P : \typeCbg{$C$}
    &
    \bigl(
        k : \typeBbg{$B$} \rightsquigarrow \typeBbg{$B$} \asemicolon x : \typeAbg{$A$} \vdash Q_{\opr} : \typeAbg{$A$}
    \bigr)_{\opr : \typeAbg{$\scriptstyle A$} \opto \typeBbg{$\scriptstyle B$}}
}
\]
To justify the above typing rule, we introduce a novel variant of handlers, \emph{reverse handlers}, for arrows.
The reverse handlers are based on a new construction method of algebras of strong promonads $\promonad$ \cite{HeunenJacobs06,Asada2010,Sanada2024} on
a reverse differential restriction category (RDRC) \cite{Cruttwell+2021}, which is a generalization of
a Cartesian reverse differential category (CRDC) \cite{Cockett+2020}.

Typical examples of RDRCs, e.g.\ the categories $\PartSmooth,\Smooth$ (see \S\ref{subsec:RDRC}) are not closed.
Hence, the type of continuations $B \rightsquigarrow B$ in the above typing rule for reverse handlers cannot be represented as an exponential $\itp{B} \expto \itp{B}$ nor $\promonad(\itp{B}, \itp{B})$, which is interpretation of effectful computation, in a CRDC $\catx$.
To solve this problem, we divide $Q_{\opr}$ into two parts: a forward computation part $Q^{\fwd}_{\opr}$ and a backward computation part $Q^{\bwd}_{\opr}$.
The modified typing rule is as follows:
\[
%\scalebox{0.9}{$
\footnotesize
\infer{
    \revhandlerjudgment \{ x \mapsto P \} \cup \{ (x \mapsto Q_{\opr}^{\fwd} \mid y, z \mapsto Q_{\opr}^{\bwd}) \} : C
}{
    \emptyenv \asemicolon x : \typeCbg{$C$} \vdash P : \typeCbg{$C$}
    &
    \emptyenv \asemicolon x : \typeAbg{$A$} \vdash Q_{\opr}^{\fwd} : \typeBbg{$B$} \times \typeDbg{$D_{\opr}$}
    &
    \emptyenv \asemicolon y : \typeBbg{$B$}, z : \typeDbg{$D_{\opr}$} \vdash Q_{\opr}^{\bwd} : \typeAbg{$A$}
    &
    \forall \opr : \typeAbg{$A$} \opto \typeBbg{$B$} \in \signatopr
}
%$}
\]
where the type $\typeDbgMath{D_{\opr}}$ is chosen by the programmer and is a type of auxiliary data for the operation $\opr$.
When an operation $\opr : A \opto B$ is handled, the forward computation part $Q_{\opr}^{\fwd}$ takes an input of type $A$ and returns a pair of the output of type $B$, which will be passed to the continuation, and some auxiliary data of type $D_{\opr}$.
The backward computation part $Q_{\opr}^{\bwd}$ takes a value of type $B$ from the continuation and the auxiliary data of type $D_{\opr}$ from the forward computation and returns the result of type $A$.
The informal diagram of the behavior is as follows:
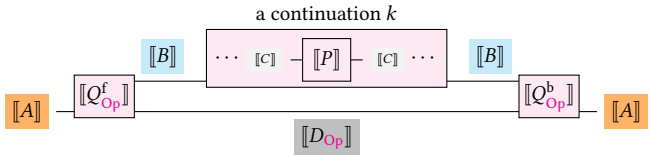
\[
    \begin{tikzpicture}[baseline=(current bounding box.center), scale=0.7, transform shape]
        \pgfmathsetmacro{\yw}{0.1}
        \pgfmathsetmacro{\yd}{0.5}
        \pgfmathsetmacro{\xd}{0.6}
        \pgfmathsetmacro{\xw}{1}
        \pgfmathsetmacro{\xwDiag}{0.2}
        \pgfmathsetmacro{\xStart}{0}
        \pgfmathsetmacro{\xQfAL}{\xStart + 0.3}
        \pgfmathsetmacro{\xQfAR}{\xQfAL + \xw}
        \pgfmathsetmacro{\xKL}{\xQfAR + 2 * \xd}
        \pgfmathsetmacro{\xKR}{\xKL + 4 * \xw}
        \pgfmathsetmacro{\xQbAL}{\xKR + 2 * \xd}
        \pgfmathsetmacro{\xQbAR}{\xQbAL + \xw}
        \pgfmathsetmacro{\xEnd}{\xQbAR + 0.3}
        \pgfmathsetmacro{\xPL}{0.5 * \xKL + 0.5 * \xKR - 0.4 * \xw}
        \pgfmathsetmacro{\xPR}{0.5 * \xKL + 0.5 * \xKR + 0.4 * \xw}
        % Q^fwd A
        \draw (\xStart, 0) node[left] {\typeAbg{$\itp{A}$}} -- (\xQfAR, 0);
        \draw[comp] (\xQfAL, \yd+\yw) rectangle node {$\itp{Q^{\fwd}_{\opr}}$} (\xQfAR, -\yw);
        % continuation k
        \draw (\xQfAR,     \yd) node[above right] {\typeBbg{$\itp{B}$}} -- (\xQbAL,     \yd) node[above left] {\typeBbg{$\itp{B}$}};
        \draw (\xQfAR,       0) to node[below] {\typeDbg{$\itp{D_{\opr}}$}} (\xQbAL,       0);
        \draw[comp] (\xKL, 2.5 * \yd + \yw) rectangle  node[above=0.5cm] {a continuation $k$} (\xKR, 1 * \yd - \yw);
        % Q^bwd A
        \draw (\xQbAR,       0) -- (\xEnd, 0) node[right] {\typeAbg{$\itp{A}$}};
        \draw[comp] (\xQbAL, \yd+\yw) rectangle node {$\itp{Q^{\bwd}_{\opr}}$} (\xQbAR, -\yw);
        % P
        \draw (\xPL - 0.2, 1.7 * \yd) node[left] {$\cdots$ \typeCbg{\tiny$\itp{C}$}} -- (\xPR + 0.2, 1.7 * \yd) node[right] {\typeCbg{\tiny$\itp{C}$} $\cdots$};
        \draw[comp] (\xPL, 2.2 * \yd + \yw) rectangle node {$\itp{P}$} (\xPR, 1.2 * \yd - \yw);
    \end{tikzpicture}.
\]
Dividing $Q_{\opr}$ into the forward part $Q^{\fwd}_{\opr}$ and the backward part $Q^{\bwd}_{\opr}$ introduces a restriction that a continuation $k$ must be used exactly once.
The restriction is natural because, in the implementation of backpropagation, continuations are used exactly once to propagate gradients backward.

Under the above typing rule for reverse handlers, a handler for $\opFC{\ell}$ can be written as follows (see Example~\ref{ex:mlp-backpropagation} for details):
\begin{align*}
    \vspace{-0.5em}
    H_{\mathrm{train}} &=
    \{ x \mapsto \fromvec{\alpha} \fnscalarmul (x \fnminus \fromvec{t}) \} \cup
    \bigl\{
        (x \mapsto Q_{\opFC{\ell}}^{\fwd} \mid y, z \mapsto Q_{\opFC{\ell}}^{\bwd})
    \bigr\},
    \\
    Q^{\fwd}_{\opFC{\ell}}
    &=
    \plet \; w \pbind \opGet{\ell}() \; \pin \;
    \pret{\tuple{w \fnmatmul x, \tuple{w, x}}},
    \\
    Q^{\bwd}_{\opFC{\ell}}
    &=
    \plet \; \tuple{w, x} \pbind z \; \pin \;
    \plet \; \_ \pbind \opPut{\ell}(w - y \fnmatmul (x^{\fntranspose})) \; \pin \;
    \pret{w^{\fntranspose} \fnmatmul y}.
\end{align*}

\subsection{Monads vs Arrows}
We now explain a basic intuition behind the notion of effects for monads and arrows categorically.
Our guiding idea is that an effect carried by a computation is realized as a continuation transformation.
We show that this continuation-transformational viewpoint is described more naturally in the setting of arrows.
Moreover, arrows admit a greater flexibility, and this is what makes reverse handlers possible.
%This viewpoint is available both for monads and for arrows, but arrows admit a more flexible notion of algebra carrier, which is makes reverse handlers possible.

\subsubsection{Monads}
For a monad $\monad$, a computation with effect is understood as a Kleisli morphism $X \to \monad(Y)$,
where $\monad(Y)$ is the type representing computations with effects that produce outputs of type $Y$ \cite{Moggi1991}.
In the semantics of effect handlers for monads \cite{PlotkinPretnar2013}, the effect carried by such a computation is understood in terms of continuation transformation,
and a $\monad$-algebra homomorphism can be regarded as a representation of such continuation-transforming behavior.
A $\monad$-algebra is an object $Z$ together with a morphism $\alpha \colon \monad(Z) \to Z$ satisfying certain axioms,
and there is a free $\monad$-algebra $\mu_Y \colon \monad(\monad(Y)) \to \monad(Y)$.
%Thus, in the monadic setting, the carrier of the algebra is an object.
Fixing an algebra $(Z,\alpha)$, we may regard an element of $\catc(X,Z)$ as a continuation at $X$.
Then a computation $R \in \catc(X,\monad(Y))$ induces a continuation transformation
$\catc(Y,Z) \to \catc(X,Z)$
as follows.
Given a continuation $\phi \in \catc(Y,Z)$, the algebra $(Z,\alpha)$ and $\phi$ uniquely induce a homomorphism $h \colon \monad(Y) \to Z$ from the free algebra.
%\[
%\begin{tikzcd}[row sep=small]
%    Y
%    \ar[r, "\eta_Y"]
%    \ar[rd, "\phi"']
%    &
%    \monad(Y)
%    \ar[dashed, d, "h"]
%    \\
%    &
%    Z
%\end{tikzcd}
%\qquad
%\begin{tikzcd}[row sep=small]
%    \monad(\monad(Y))
%    \ar[r, "\mu_Y"]
%    \ar[d, dashed, "\monad(h)"']
%    &
%    \monad(Y)
%    \ar[d, dashed, "h"]
%    \\
%    \monad(Z)
%    \ar[r, "\alpha"']
%    &
%    Z
%\end{tikzcd}
%\]
Then the above continuation transformation sends $\phi$ to the composite $R \semicolon h \in \catc(X,Z)$.
This gives the monadic picture of effectful computation as continuation transformation.

\subsubsection{Promonads}
\label{subsubsec:monads-vs-arrows:promonads}
We use promonads $\promonad$ on a category $\catc$ as a categorical foundation for arrows \cite{JacobsHeunenHasuo2009,Asada2010}.
Intuitively, a promonad is a family $\{\promonad(X, Y)\}_{X,Y}$,
where each $\promonad(X,Y)$ is a set of computations with effects that take an input of type $X$ and produce an output of type $Y$.
If we have a monad $\monad$ on $\catc$, we can construct a promonad $\promonad_\monad$ on $\catc$ by $\promonad_\monad(X, Y) = \catc(X, \monad(Y))$.

One point of passing from monads to promonads is that the continuation transformational viewpoint can be described more directly.
In contrast to the monad case, the carrier \emph{object} is replaced by a carrier \emph{presheaf}.
An $\promonad$-algebra is such a carrier presheaf $G$ equipped with a family of morphisms
$\{ \alpha_{X,Y} \colon \promonad(X, Y) \times GY \to GX \}_{X,Y}$ that satisfies certain axioms.
This is the categorical nature of the notion of continuation transformation: an element of $GY$ is regarded as a continuation at $Y$, and $\alpha$ specifies how a computation in $\promonad(X,Y)$ transforms a continuation at $Y$ into one at $X$.

This viewpoint subsumes the monadic one.
In fact, the continuation-transformational reading of $\monad$-algebras just explained can be recasted as a procedure to see a $\monad$-algebra
as an $\promonad_\monad$-algebra on the representable carrier presheaf $G=\catc(\blank,Z)$ via the Yoneda lemma \cite[pp.\ 7--8]{Wood1985}:
What we did was equivalent to defining
$\alpha_{X,Y} \colon \catc(X, \monad Y) \times \catc(Y, Z) \to \catc(X, Z)$
by the action of $\monad$ on hom-sets and post-composition with $\alpha$.
Thus, the usual Eilenberg--Moore viewpoint is precisely the representable case of the promonad-algebraic view of continuation transformation.
On the other hand, for arbitrary promonad-algebras, the notion of continuation no longer needs to be representable:
One may take an arbitrary carrier presheaf $G$.
An important instance is the one used by Sanada~\cite{Sanada2024} for usual arrow handlers, where the carrier presheaf is $\promonad(\blank,D)$.

This presheaf-valued carrier is the additional flexibility provided by arrows, and it is exactly what we use for reverse handlers.
For a promonad $\promonad$ on an RDRC $\catx$, we take as carrier presheaf the reverse algebra $\ralg{\promonad}$ (Definition~\ref{def:reverse-algebra}), defined by $\ralg{\promonad}(X) \defeq \promonad(X,X)$.
Thus the carrier consists of endomorphic continuations.
The reverse differential combinator of the RDRC makes this assignment into a presheaf.
%The crucial difference from the monadic setting is therefore not merely that arrows record both input and output types, but that promonad algebras allow such presheaf-valued carriers, and in particular the carrier presheaf of endomorphic continuations $\ralg{\promonad}$.
Given a reverse handler, we can construct an algebra structure $\alpha$ on the reverse algebra $\ralg{\promonad}$, so that
any computation $R\in\promonad(X,Y)$ induces a continuation transformation of the form $\promonad(Y,Y) \to \promonad(X,X)$.
%In particular, any operation symbol (i.e. a layer)
%$\opr :\typeAbg{$A$} \opto \typeBbg{$B$}$ induces a continuation transformation
%$\promonad(\typeBbg{$\itp{B}$}, \typeBbg{$\itp{B}$})  \to \promonad(\typeAbg{$\itp{A}$}, \typeAbg{$\itp{A}$})$

For each fixed $Z$, the carrier presheaf $\promonad(\blank, Z)$ is equipped with a free $\promonad$-algebra structure.
Thus, just as in the monadic case, an algebra structure on the chosen carrier presheaf induces homomorphisms from this free algebra.
%The interpretation is then presented by a homomorphism from the free algebra.
%\[
%\begin{tikzcd}[row sep=small]
%    \catc(X,Z)
%    \ar[r, "\eta_{X,Z}"]
%    \ar[rd, "\phi_{X}"']
%    &
%    \promonad(X,Z)
%    \ar[dashed, d, "h_{X}"]
%    \\
%    &
%    \promonad(X,W)
%\end{tikzcd}
%\qquad
%\begin{tikzcd}[row sep=small]
%    \promonad(X,Y) \times \promonad(Y,Z)
%    \ar[r, "\mu_{X,Y,Z}"]
%    \ar[d, dashed, "\idmor \times h_{Y}"']
%    &
%    \promonad(X,Z)
%    \ar[d, dashed, "h_{X}"]
%    \\
%    \promonad(X,Y) \times \promonad(Y,W)
%    \ar[r, "\alpha_{X,Y}"']
%    &
%    \promonad(X,W)
%\end{tikzcd}
%\]
\begin{equation}
\label{eq:reverse-handler-algebra}
\begin{tikzcd}[row sep=small]
    \catx(X,Z)
    \ar[r, "\eta_{X,Z}"]
    \ar[rd, "\phi_{X}"']
    &
    \promonad(X,Z)
    \ar[dashed, d, "h_{X}"]
    \\
    &
    \ralg{\promonad}(X)
\end{tikzcd}
\qquad
\begin{tikzcd}[row sep=small]
    \promonad(X,Y) \times \promonad(Y,Z)
    \ar[r, "\mu_{X,Y,Z}"]
    \ar[d, dashed, "\idmor \times h_{Y}"']
    &
    \promonad(X,Z)
    \ar[d, dashed, "h_{X}"]
    \\
    \promonad(X,Y) \times \ralg{\promonad}(Y)
    \ar[r, "\alpha_{X,Y}"']
    &
    \ralg{\promonad}(X)
\end{tikzcd}
\end{equation}
In our case, a reverse handler $H$ gives rise to an algebra structure on $\ralg{\promonad}$
and a morphism $\phi$, and the unique homomorphism $h$ is the interpretation of $H$.
The detailed construction is given in Section~\ref{sec:denotational-semantics}.

\section{Syntax}
\label{sec:syntax}
We introduce the syntax and typing rules of our language.
The language is an extension of the arrow calculus with operations and handlers \cite{Sanada2024}.
The pure part of the language is a simplified version of a simple differentiable programming language \cite{AbadiPlotkin2019}.
Let $\BType$ be a set of \emph{base types}.
A type $1 \in \BType$ is specified to represent the unit type.
\begin{definition}
    A \emph{type} $A$ and a \emph{typing environment} $\Gamma$ are defined by the following BNF:
    \begin{equation*}
        \Type \ni A, B, C, D \bnfeq \beta \mid \typprod_{i=1}^n A_i
        \ \text{where $\beta$ ranges over base types $\BType$,}
        \quad
        \Gamma, \Delta \bnfeq \emptyenv \mid x : A, \Gamma.
    \end{equation*}
\end{definition}
An example of $\BType$ is $\{ \typreal(n) \mid n \in \Nat \}$ with $\typreal(0)$ as the unit type.
The base type $\typreal(n)$ represents the type of $n$-dimensional real vectors.

\begin{definition}[signatures of function symbols]
    A \emph{signature of function symbols} $\signatfunc$ is a set of function symbols $\func$.
    For each function symbol $\func \in \signatfunc$, a type $A$ called \emph{domain} and a type $B$ called \emph{codomain} are assigned.
    We write $\func : A \to B$ to indicate that the domain and codomain of $\func$ are $A$ and $B$, respectively.
    For a base type $\beta$, we call a function symbol $\const : 1 \to \beta$ a \emph{constant symbol} of type $\beta$, and write $\const : \beta$ for it.
    We assume that $\signatfunc$ contains distinguished constants $0_{\beta} : \beta$ for each $\beta \in \BType$.
    A signature of function symbols $\signatfunc$ is \emph{reverse derivative closed} if for any function symbols $\func : A \to B \in \signatfunc$, there exists a function symbol $\rd[\func] : A \times B \to A \in \signatfunc$.
\end{definition}

Each function symbol $\func : A \to B$ represents a pure primitive function from type $A$ to type $B$.
For a function $f \colon \Real^n \to \Real^m$, its reverse derivative is the following function:
\begin{equation}
\label{eq:reverse-derivative}
    \Real^{n} \times \Real^{m} \ni (x, y)
    \mapsto
    \left(
        \sum_{j=1}^m y_j \frac{\partial f_j}{\partial x_1}(x), \dots, \sum_{j=1}^m y_j \frac{\partial f_j}{\partial x_n}(x)
    \right) \in \Real^n. 
\end{equation}
Intuitively, $\rd[\func]$ is the symbolic reverse derivative of the pure primitive function $\func$.

\begin{definition}[signatures of operation symbols]
    A \emph{signature of operation symbols} $\signatopr$ is a set of operation symbols $\opr$.
    For each operation symbol $\opr \in \signatopr$, a type $A$ called \emph{coarity} and a type $B$ called \emph{arity} are assigned.
    We write $\opr : A \opto B$ to indicate that the coarity and arity of $\opr$ are $A$ and $B$, respectively.
\end{definition}

An operation symbol $\opr : A \opto B$ represents an effectful computation whose input and output types are $A$ and $B$, respectively.
Programmers can implement its behavior by defining handlers.

\begin{definition}[terms and commands]
    \label{def:syntax:terms-commands}
    Terms $L, M, N$, commands $P, Q, R$ and handlers $H$ are defined by the following BNF:
    \begin{equation*}
        \begin{aligned}
            L, M, N
            & \bnfeq x \mid \const \mid \func(M) \mid M_1 + M_2 \mid \tuple{M_1, \dots, M_n}
            \mid \tproj_i(M)
            \mid \tlet \; x \tbind M \; \tin \; N
            \mid M.\trd(x.N)(L)
            \\
            U, V, W
            & \bnfeq x \mid \const \mid \tuple{V_1, \dots, V_n} \\
            & \text{where $x$ ranges over variables, $\const$ over constant symbols and $\func$ over function symbols.}
            \\
    %    \end{aligned}
    %\end{equation*}
    %\begin{equation*}
    %    \begin{aligned}
            P, Q, R
            & \bnfeq \pret{M} \mid \opr(M) \mid \plet \; x \pbind P \; \pin \; Q
            \mid \prevhandle(M) \tuple{x_1, \dots, x_n}.R \; \pwith \; H \\
            & \text{where $\opr$ ranges over operation symbols in $\signatopr$.}
            \\
            H
            & \bnfeq \{ x \mapsto P \} \cup \{ (x \mapsto Q_{\opr}^{\fwd} \mid y, z \mapsto Q_{\opr}^{\bwd}) \}_{\opr \in \signatopr}.
        \end{aligned}
    \end{equation*}
    We write $\tlet \; \tuple{x_i}_{i=1}^n \tbind M \; \tin \; N$ 
    and $\plet \; \tuple{x_i}_{i=1}^n \pbind P \; \pin \; Q$ for
    \[
    \begin{aligned}
        & \tlet \; x \tbind M \; \tin \; \tlet \; x_1 \tbind \tproj_1(x) \; \pin \; \dots \; \tlet \; x_n \tbind \tproj_n(x) \; \pin \; N,
        \\
        & \plet \; x \pbind P \; \pin \; \plet \; x_1 \pbind \pret{\tproj_1(x)} \; \pin \; \dots \; \plet \; x_n \pbind \pret{\tproj_n(x)} \; \pin \; Q,
    \end{aligned}
    \] respectively, where $x$ is a fresh variable.
\end{definition}

Terms represent pure computations, while commands represent effectful computations.
The syntax of terms is a simplified version of that of a simple differentiable programming language introduced by Abadi and Plotkin \cite{AbadiPlotkin2019}.
Compared to the language by Abadi and Plotkin, terms in our language do not have conditional expressions and recursion constructs for simplicity.
The syntax of commands is a variant of that of an arrow calculus with operations and handlers by Sanada \cite{Sanada2024}.
The difference is that our commands have a construct for reverse handling $\prevhandle(M) \tuple{x_i}_{i=1}^n.R \; \pwith \; H$ instead of the usual handling construct.

We give intuitive explanations for the constructs.
A term $\func(M)$ is application of a pure primitive function $\func$ to an argument $M$.
A term $M_1 + M_2$ represents addition of two terms $M_1$ and $M_2$ of base type.
A term $\tuple{M_i}_{i=1}^n$ is an $n$-tuple of terms $M_1, \dots, M_n$.
A term $\tproj_i(M)$ is the $i$-th projection of a term $M$ of product type.
A term $\tlet \; x \tbind M \; \tin \; N$ is a sequential composition of two terms $M$ and $N$, where the result of $M$ is bound to a variable $x$ in $N$.

A term $M.\trd(x.N)(L)$ represents the reverse derivative of a term $N$ by a variable $x$.
Intuitively, the terms $L$, $M$ and $N$ correspond to $x$, $y$ and $f$ in \eqref{eq:reverse-derivative}, respectively.
For example, let $\func : \typreal(1) \to \typreal(1)$, $L = \fromvec{2}$, $M = \fromvec{3}$ and $N = \func(x)$. Then, the term $\fromvec{3}.\trd(x.\func(x))(\fromvec{2})$ represents the value
\[3 \cdot \frac{\dif \itp{\func}}{\dif x}(2) \quad \in \quad \Real\]
where $\itp{\func}$ is the interpretation of $\func$ as a function $\Real \to \Real$.
Since our language does not include function types, the reverse derivative must be expressed in the application form $M.\trd(x.N)(L)$.

A command $\pret{M}$ represents a pure computation $M$ lifted to a command.
A command $\opr(M)$ represents an effectful operation $\opr$ applied to an argument $M$.
A command $\plet \; x \pbind P \; \pin \; Q$ is a sequential composition of  $P$ and $Q$, where the result of $P$ is bound to a variable $x$ in $Q$.

A command $P = \prevhandle(M) \tuple{x_i}_{i=1}^n.R \; \pwith \; H$ represents reverse handling of a command $R$ with a handler $H$, where the result of a pure computation $M$ is bound to variables $x_1, \dots, x_n$ in $R$.
The term $M$ is used to pass the input of the computation to be handled, and has a similar role to $L$ in the term $M.\trd(x.N)(L)$.
In the reduction of the command $P$, if $R$ invokes an operation $\opr$, then the command $Q^{\fwd}_{\opr}$ defined in the handler $H$ is executed, and after the continuation of the computation of $P$ returns a result, the command $Q^{\bwd}_{\opr}$ defined in the handler $H$ is executed.

\begin{definition}[typing]
    The derivations of typing judgments $\Gamma \vdash M : A$, $\Gamma \asemicolon \Delta \vdash P : A$ and $\revhandlerjudgment H : C$ for terms, commands and handlers are defined by the rules in Fig.~\ref{fig:typing-terms} and \ref{fig:typing-commands-handlers}.
    \begin{figure}
        \footnotesize
        \centering
        \begin{gather*}
            \infer[\tytvar]{
                \Gamma \vdash x : A
            }{x : A \in \Gamma}
            \quad
            \infer[\tytconst]{
                \Gamma \vdash \const : \beta
            }{\const : \beta \in \signatfunc}
            \quad
            \infer[\tytfunc]{
                \Gamma \vdash \func(M) : B
            }{
                \Gamma \vdash M : A
                &
                \func : A \to B \in \signatfunc
            }
            \quad
            \infer[\tytlet]{
                \Gamma \vdash \tlet \; x \tbind M \; \tin \; N : B
            }{
                \Gamma \vdash M : A
                &
                \Gamma, x : A \vdash N : B
            }
            \\
            \infer[\tytplus]{
                \Gamma \vdash M_1 + M_2 : \beta
            }{
                \Gamma \vdash M_1 : \beta
                &
                \Gamma \vdash M_2 : \beta
            }
            \quad
            \infer[\tyttuple]{
                \Gamma \vdash \tuple{M_1, \dots, M_n} : \typprod_{i=1}^n A_i
            }{
                \Gamma \vdash M_i : A_i \quad (1 \leq i \leq n)
            }
            \quad
            \infer[\tytproj]{
                \Gamma \vdash \tproj_i(M) : A_i
            }{
                \Gamma \vdash M : \typprod_{i=1}^n A_i
                &
                1 \leq i \leq n
            }
            \\
            \infer[\tytrd]{
                \Gamma \vdash M.\trd(x.N)(L) : A
            }{
                \Gamma \vdash M : B
                &
                \Gamma, x : A \vdash N : B
                &
                \Gamma \vdash L : A
            }
        \end{gather*}
        \Description{Typing rules for terms}
        \caption{Typing rules for terms}\label{fig:typing-terms}
    \end{figure}
    \begin{figure}
        \footnotesize
        \centering
        \begin{gather*}
            \infer[\tycpure]{
                \Gamma \asemicolon \Delta \vdash \pret{M} : A
            }{
                \Gamma, \Delta \vdash M : A
            }
            \quad
            \infer[\tycop]{
                \Gamma \asemicolon \Delta \vdash \opr(M) : B
            }{
                \Gamma, \Delta \vdash M : A
                &
                \opr : A \opto B \in \signatopr
            }
            \quad
            \infer[\tyclet]{
                \Gamma \asemicolon \Delta \vdash \plet \; x \pbind P \; \pin \; Q : C
            }{
                \Gamma \asemicolon \Delta \vdash P : A
                &
                \Gamma \asemicolon x : A, \Delta \vdash Q : C
            }
            \\
            \infer[\tychandle]{
                \Gamma \asemicolon \Delta \vdash \prevhandle(M) \tuple{x_1, \dots, x_n}.R \; \pwith \; H : \typprod_{i=1}^n A_i
            }{
                \Gamma, \Delta \vdash M : \typprod_{i=1}^n A_i
                &
                \Gamma \asemicolon x_1 : A_1, \dots, x_n : A_n \vdash R : C
                &
                \revhandlerjudgment H : C
            }
        \end{gather*}
        \[
            \scalebox{0.9}{$
            \infer[\tyhandler]{
                \revhandlerjudgment \{ x \mapsto P \} \cup \{ (x \mapsto Q_{\opr}^{\fwd} \mid y, z \mapsto Q_{\opr}^{\bwd}) \} : C
            }{
                \emptyenv \asemicolon x : C \vdash P : C
                &
                \emptyenv \asemicolon x : A \vdash Q_{\opr}^{\fwd} : B \times D_{\opr}
                &
                \emptyenv \asemicolon y : B, z : D_{\opr} \vdash Q_{\opr}^{\bwd} : A
                &
                \forall \opr : A \opto B \in \signatopr
            }
            $}
        \]
        \Description{Typing rules for commands and handlers}
        \caption{Typing rules for commands and handlers}\label{fig:typing-commands-handlers}
    \end{figure}
\end{definition}
The typing judgment $\Gamma \vdash M : A$ means that the pure term $M$ has type $A$ under the typing environment $\Gamma$.
The typing judgment $\Gamma \asemicolon \Delta \vdash P : A$ means that the command, that is, the effectful computation, $P$ has type $A$ under the typing environments $\Gamma$ and $\Delta$.
Intuitively, the environment $\Gamma$ is a pure typing environment and $\Delta$ represents inputs of the effectful computation $P$.
For handlers, the typing judgment $\revhandlerjudgment H : C$ means that the handler $H$ is well typed for handling computations of type $C$;
here, for simplicity of the semantics given later, we consider only handlers with an empty pure typing environment.
%Semantically, we will define the interpretation of the handler $H$ as a homomorphism $\{ \promonad(X, \itp{C}) \to \ralg{\promonad}(X)\}_X$ between $\promonad$-algebras.

\begin{definition}[depth and size]
    Let $M$ be a term, $P$ a command, $H$ a reverse handler.
    We define numbers $\rhdepth{R} \in \Nat$ and $\rhdepth{H} \in \Nat$ as the depth of reverse handler constructs, and numbers $\dsize{\Gamma \vdash M : A} \in \Nat$, $\dsize{\Gamma \asemicolon \Delta \vdash R : C} \in \Nat$ and $\dsize{\revhandlerjudgment H : C} \in \Nat$ as the size of derivation tree.
    The precise definitions are explained in Appendix~\ref{sec:definition-depth-size}.
\end{definition}
\begin{toappendix}
\section{Definition of depth and size}
\label{sec:definition-depth-size}
In order to define operational semantics and prove its properties, we introduce two measures: depth of reverse handlers and size of derivation trees.
We will use the lexicographic order on the set of pairs of these two measures in recursive definitions and inductive proofs of several lemmas in \S\ref{sec:operational-semantics} and \ref{sec:soundness-adequacy}.
\begin{definition}[depth of reverse handler]\label{def:depth-of-reverse-handler}
    For a well-typed command $\Gamma \asemicolon \Delta \vdash P : A$ and a well-typed handler $\revhandlerjudgment H : C$,
    the \emph{depth} $\rhdepth{P}$ and $\rhdepth{H}$ is defined in Fig.~\ref{fig:depth-of-reverse-handlers}.
    \begin{figure}
    \footnotesize
    \begin{align*}
        & \rhdepth{\pret{M}}
        = 0
        \qquad\qquad
        \rhdepth{\opr(M)}
        = 0
        \qquad\qquad
        \rhdepth{\plet \; x \pbind P \; \pin \; Q}
        = \max\{\rhdepth{P}, \rhdepth{Q}\}
        \\
        & \rhdepth{\prevhandle(M) \tuple{x_i}_{i=1}^n.R \; \pwith \; H}
        = 1 + \max\left\{\rhdepth{R}, \rhdepth{H}\right\}
        \\
        & \rhdepth{\{ x \mapsto P \} \cup \{ (x \mapsto Q^{\fwd}_{\opr} \mid y, z \mapsto Q^{\bwd}_{\opr})\}_{\opr \in \signatopr}}
        = \max\left\{\rhdepth{P}, \max_{\opr \in \signatopr}\{\rhdepth{Q^{\fwd}_{\opr}}, \rhdepth{Q^{\bwd}_{\opr}}\}\right\}
    \end{align*}
    \caption{Depth of reverse handlers}
    \label{fig:depth-of-reverse-handlers}
    \end{figure}
\end{definition}
\begin{definition}[size of derivation trees]\label{def:size-of-derivation-tree}
    For a well-typed term $\Gamma \vdash M : A$, a well-typed command $\Gamma \asemicolon \Delta \vdash P : A$ and a well-typed handler $\revhandlerjudgment H : C$, we define the \emph{size of the derivation tree} $\dsize{\Gamma \vdash M : A}$, $\dsize{\Gamma \asemicolon \Delta \vdash P : A}$ and $\dsize{\revhandlerjudgment H : C}$ are shown in Fig.~\ref{fig:size-of-derivation-trees}.
    \begin{figure}
        \footnotesize
        \begin{gather*}
            \dsize{\Gamma \vdash x : A}
            = 1
            \qquad
            \dsize{\Gamma \vdash \const : \beta}
            = 1
            \qquad
            \dsize{\Gamma \vdash \func(M) : B}
            = 1 + \dsize{\Gamma \vdash M : A}
            \\
            \dsize{\Gamma \vdash \tproj_i(M) : A_i}
            = 1 + \dsize{\Gamma \vdash M : \typprod_{i=1}^n A_i}
            \qquad
            \dsize{\Gamma \vdash \tuple{M_i}_{i=1}^n : \typprod_{i=1}^n A_i}
            = 1 + \textstyle\sum_{i=1}^{n} \dsize{\Gamma \vdash M_i : A_i}
            \\
            \dsize{\Gamma \vdash M_1 + M_2 : \beta}
            = 1 + \dsize{\Gamma \vdash M_1 : \beta} + \dsize{\Gamma \vdash M_2 : \beta}
            \\
            \dsize{\Gamma \vdash \tlet \; x \tbind M \; \tin \; N : B}
            = 1 + \dsize{\Gamma \vdash M : A} + \dsize{\Gamma, x : A \vdash N : B}
            \\
            \dsize{\Gamma \vdash M.\trd(x.N)(L) : A}
            = 1 + \dsize{\Gamma \vdash M : B} + \dsize{\Gamma, x : A \vdash N : B} + \dsize{\Gamma \vdash L : A}
        \end{gather*}
        \begin{gather*}
            \dsize{\Gamma \asemicolon \Delta \vdash \pret{M} : A}
            = \dsize{\Gamma, \Delta \vdash M : A}
            \quad
            \dsize{\Gamma \asemicolon \Delta \vdash \opr(M) : B}
            = 2 + \dsize{\Gamma, \Delta \vdash M : A}
            \\
            \dsize{\Gamma \asemicolon \Delta \vdash \plet \; x \pbind P \; \pin \; Q : C} = 1 + \dsize{\Gamma \asemicolon \Delta \vdash P : A} + \dsize{\Gamma \asemicolon x : A, \Delta \vdash Q : C}
            \\
            \dsize{\Gamma \asemicolon \Delta' \vdash \prevhandle(M) \tuple{x_i}_{i=1}^n.R \; \pwith \; H : \typprod_{i = 1}^n A_i}
            = 1 + \dsize{\Gamma, \Delta' \vdash M : \typprod_{i = 1}^n A_i} + \dsize{\Gamma \asemicolon \Delta \vdash R : C} + \dsize{\revhandlerjudgment H : C}
        \end{gather*}
        \begin{equation*}
            \begin{aligned}
                & \dsize{\revhandlerjudgment \{ x \mapsto P \} \cup \{ (x \mapsto Q_{\opr}^{\fwd} \mid y, z \mapsto Q_{\opr}^{\bwd}) \}_{\opr \in \signatopr} : C} \\
                & = 1 + \dsize{\emptyenv \asemicolon x : C \vdash P : C} + \sum_{\opr : A \opto B \in \signatopr} \left( \dsize{\emptyenv \asemicolon x : A \vdash Q_{\opr}^{\fwd} : B \times D_{\opr}} + \dsize{\emptyenv \asemicolon y : B, z : D_{\opr} \vdash Q_{\opr}^{\bwd} : A} \right)
            \end{aligned}
        \end{equation*}
        \caption{Size of derivation trees}\label{fig:size-of-derivation-trees}
        \Description{Size of derivation trees}
    \end{figure}
\end{definition}
\end{toappendix}

\section{Examples of Neural Networks}\label{sec:examples-neural-network}
\begin{table*}[t]
  \centering
  \caption{Examples of Neural Networks in our language. Each example shows different aspect of reverse handlers: ordinary backpropagation, structural composition, bidirectional behavior, and non-standard gradients.}
  \label{tab:examples-summary}
  \footnotesize
  \begin{tabular}{@{}p{0.14\textwidth}p{0.34\textwidth}
    %p{0.3\textwidth}
    p{0.47\textwidth}@{}}
    \toprule
    Example
      & Network signatures
      %& Reverse-handler clauses
      & What the example demonstrates \\
    \midrule
    MLP. \ref{ex:mlp-backpropagation}, \ref{ex:mlp-reduction}, \ref{ex:mlp-denotation}.
      & Fully connected layers $\opFC{\ell}$ with an activation function $\fnSwish$.
      %& $Q^{\fwd}_{\opFC{l}}$ reads the parameter matrix from the location $\ell$ and performs the matrix multiplication. $Q^{\bwd}_{\opFC{l}}$ updates the parameter and propagates the gradient backwards.
      & A symbolic network description with standard backpropagation by a reverse handler. \\
    \addlinespace
    ResNet. \ref{ex:residual-neural-network}, \ref{ex:resnet-denotation}.
      & An MLP that has skip connections.
      %& The same handler as the MLP example.
      & Skip connections are realized by the strength of promonads \\
    \addlinespace
    Autoencoder. \ref{ex:autoencoder}, \ref{ex:autoencoder-denotation}.
      & A network with an encoder part and a decoder part. A pair of an encoder and a decoder is represented by $\opEnDe{\ell,\ell'}$.
      %& The forward clause $Q^{\fwd}_{\opEnDe{\ell,\ell'}}$ implements the encoder; the backward $Q^{\bwd}_{\opEnDe{\ell,\ell'}}$ clause implements the decoder.
      & The backward clause $Q^{\bwd}_{\opEnDe{\ell, \ell'}}$ need not be the derivative of the forward
        clause $Q^{\fwd}_{\opEnDe{\ell, \ell'}}$; reverse handlers can express more general bidirectional
        computations. \\
    \addlinespace
    CNN. \ref{ex:convolutional-neural-network}.
      & Convolution layers $\opConv{n, \ell}$ and pooling layers $\opPool{n, m, c}$.
      %& Adds clauses for $\opConv{n, l}$ and $\opPool{n, m, c}$ to the fully connected layer clauses.
      & The language can express common neural-network layers, not just fully
        connected networks. \\
    \addlinespace
    U-net. \ref{ex:U-net}, \ref{ex:U-net-denotation}.
      & Nested encoder blocks with decoder-like structure and skip connections.
      %& The forward clauses build the encoder side; the backward clauses reconstruct the decoder side.
      & Our language can describe practical architectures. The string diagram obtained through the interpretation accurately represents the network structure. \\
    \addlinespace
    STE in QAT. \ref{ex:ste}.
      & A layer whose forward pass applies a quantization function such as
        $\fnround$.
      %& The forward clause uses the round function $\fnround$, while the backward clause propagates gradients as if the layer were the identity.
      & Reverse handlers can assign non-standard backward behavior, such as STE. \\
    \bottomrule
  \end{tabular}
\end{table*}
We show that numerous practical neural network architectures can be written in our language.
The examples are summarized in Table~\ref{tab:examples-summary}.

Let $\BType = \{ \typreal(n) \mid n \in \Nat \}$ be a set of base types with $1 = \typreal(0)$ as the unit type.
Let $\Locations = \{ \ell_0 , \ell_1, \dots \}$ be a set of \emph{locations}.
Assume that for each location $\ell \in \Locations$, a sequence of natural numbers $(n_1, n_2, \dots, n_t)$ is assigned.
We write $\loctyp{\ell}{(n_1, n_2, \dots, n_t)}$ to indicate that $(n_1, n_2, \dots, n_t)$ is assigned to the location $\ell$.
Intuitively, $\loctyp{\ell}{(n_1, n_2, \dots, n_t)}$ means that the memory location $\ell$ stores a tensor of shape $n_t \times \dots \times n_2 \times n_1$.

Let $\signatfunc$ be a reverse derivative closed signature of function symbols containing
\begin{gather*}
    \fnSwish_n : \typreal(n) \to \typreal(n),
    \quad
    (\fnscalarmul_n) : \typreal(1) \times \typreal(n) \to \typreal(n),
    \\
    (\fnminus_n) : \typreal(n) \times \typreal(n) \to \typreal(n),
    \quad
    (\fnmatmul_{n,m}) : \typreal(n m) \times \typreal(n) \to \typreal(m),
    \\
    (\blank)^\fntranspose_{n, m} : \typreal(m n) \to \typreal(m n),
    \quad
    \fromvec{v} : \typreal(0) \to \typreal(n)
\end{gather*}
for each natural numbers $n, m \in \Nat$ and each vector $v \in \Real^n$;
in particular, we prepare constant symbols $\fromvec{v}$ for all vectors $v$ of any dimension.
The symbol $\fnSwish_n$ is a symbol for the Swish activation function \cite{Ramachandran+2018} $\itp{\fnSwish_n} \colon \Real^n \to \Real^n$ defined by $\itp{\fnSwish_n}(x) = (x_i \cdot \sigmoid(x_i))_{i=1}^n$ where $\sigmoid$ is the sigmoid function defined by $\sigmoid(x) = 1 / (1 + \exp(-x))$.
The symbol $(\fnscalarmul_n)$ is a symbol for the scalar multiplication of a real vector.
The symbol $(\fnminus_n)$ is a symbol for the entry-wise subtraction of real vectors.
The symbol $(\fnmatmul_{n,m})$ is a symbol for the matrix-vector multiplication.
The symbol $(\blank)^\fntranspose_{n, m}$ is a symbol for the transpose of a matrix.

In the following examples, we may omit non-essential clauses in the construction of handlers.
Any omitted clause is assumed to be defined as
$(x \mapsto \plet \; y \pbind \opr(x) \; \pin \; \pret{\tuple{y, \fromvec{0}}} \mid y, z \mapsto \pret{\fromvec{0}})$
\begin{example}[multilayer perceptron]\label{ex:mlp-backpropagation}
    A multilayer perceptron (MLP) is the simplest type of neural networks.
    Let $\signatopr$ be a signature of operation symbols containing the following symbols for each location $\ell \in \Locations$:
    \begin{gather*}
        \begin{aligned}
            \opFC{\ell} & : \typreal(n) \opto \typreal(m)
            && \text{where $\loctyp{\ell}{(n,m)}$,}
            \\
            \opGet{\ell} & : \typreal(0) \opto \typreal(n_t \dots n_2 n_1)
            && \text{where $\loctyp{\ell}{(n_1, \dots, n_t)}$,}
            \\
            \opPut{\ell} & : \typreal(n_t \dots n_2 n_1) \opto \typreal(0)
            && \text{where $\loctyp{\ell}{(n_1, \dots, n_t)}$.}    
        \end{aligned}
    \end{gather*}
    Let $n_{\inp}, n_{\hid}, n_{\out} \in \Nat$ be natural numbers,
    $\loctyp{\ell_0}{(n_{\inp}, n_{\hid})}$ and $\loctyp{\ell_1}{(n_{\hid}, n_{\out})}$.
    We describe a simple MLP with one hidden layer of $n_{\hid}$ neurons, an input vector of size $n_{\inp}$ and an output vector of size $n_{\out}$, as follows:
    \[
        R_{\MLP} =
        \left(\begin{aligned}
            & \plet \; z_{\hid} \pbind \opFC{\ell_0}(x_{\inp}) \; \pin \;
            \plet \; x_{\hid} \pbind \pret{\fnSwish_{n_{\hid}}(z_{\hid})} \; \pin \\
            & \plet \; z_{\out} \pbind \opFC{\ell_1}(x_{\hid}) \; \pin \; \pret{z_{\out}}
        \end{aligned}\right).
    \]
    We implement the backpropagation of the mean squared error for $R_{\MLP}$ using a reverse handler $H_{\MLP} = \{ x \mapsto P \} \cup \{(x \mapsto Q^{\fwd}_{\opr} \mid y, z \mapsto Q^{\bwd}_{\opr})\}_{\opr \in \signatopr}$ where
    \[
        \begin{aligned}
            P
            & =
            \pret{\fromvec{\learningrate} \fnscalarmul (x \fnminus \fromvec{t})}
            \quad
            \text{where $t \in \Real^{n_{\out}}$ is the target vector and $\learningrate \in \Real$ is the learning rate,}
            \\
            Q^{\fwd}_{\opFC{\ell}}
            & =
            \plet \; w \pbind \opGet{\ell}() \; \pin \;
            \pret{\tuple{w \fnmatmul_{n, m} x, \tuple{w, x}}},
            \\
            Q^{\bwd}_{\opFC{\ell}}
            & =
            \plet \; \tuple{w, x} \pbind z \; \pin \;
            \plet \; \_ \pbind \opPut{\ell}(w - y \fnmatmul (x^{\fntranspose})) \; \pin \;
            \pret{w^{\fntranspose} \fnmatmul y}.
        \end{aligned}
    \]
    The commands and handlers have the following types:
    \begin{align*}
        & \emptyenv \asemicolon x_{\inp} : \typreal(n_{\inp})
        \vdash R_{\MLP} : \typreal(n_{\out}),
        \\
        & \emptyenv \asemicolon x : \typreal(n_{\out})
        \vdash P : \typreal(n_{\out}),
        \\
        & \emptyenv \asemicolon x : \typreal(n)
        \vdash Q^{\fwd}_{\opFC{\ell}} : \typreal(m) \times (\typreal(n \times m) \times \typreal(n)),
        \\
        & \emptyenv \asemicolon y : \typreal(m), z : \typreal(n \times m) \times \typreal(n)
        \vdash Q^{\bwd}_{\opFC{\ell}} : \typreal(n),
        \\
        & \revhandlerjudgment H_{\MLP} : \typreal(n_{\out}).
    \end{align*}
    An implementation of the backpropagation of the mean squared error for $R_{\MLP}$ using the reverse handler $H_{\MLP}$ is as follows:
    \[
    \emptyenv \asemicolon \emptyenv \vdash
    \prevhandle(\fromvec{v})\tuple{x_{\inp}}. R_{\MLP} \; \pwith \; H_{\MLP}
    : \typreal(n_{\inp})
    \quad
    \text{where $v \in \Real^{n_{\inp}}$}.
    \]
\end{example}

\begin{example}[residual neural network]\label{ex:residual-neural-network}
    We continue to work under the setting of Example~\ref{ex:mlp-backpropagation}.
    A \emph{residual neural network} \cite{He+2016} is a deep neural network architecture which contains \emph{residual blocks} as its building blocks.
    A residual block adds the input of the block to the output of a neural network inside the block.
    We can implement a residual block simply by
    $R_{\ResNet} = \plet \; z \pbind R_{\MLP} \; \pin \; \pret{x_{\inp} + z}$ where $n_\inp = n_\out$ and $R_{\MLP}$ is the multilayer perceptron in Example~\ref{ex:mlp-backpropagation}.
    The backpropagation for the residual block $R_{\ResNet}$ can be implemented using the same handler: $\prevhandle(\fromvec{v}) \tuple{x_{\inp}}. R_{\ResNet} \; \pwith \; H_{\MLP}$.
    We will use the strength of a strong promonad to interpret $R_{\ResNet}$.
\end{example}

\begin{example}[autoencoder]\label{ex:autoencoder}
    We give an example of a nested use of reverse handlers.
    An autoencoder \cite{HintonSalakhutdinov2006} is a neural network to reduce dimension of input data.
    For each pair of locations $\loctyp{\ell}{(n,m)}$ and $\loctyp{\ell'}{(m, n)}$, we add an operation
    $\opEnDe{\ell, \ell'} : \typreal(n) \opto \typreal(m)$
    to $\signatopr$ in Example~\ref{ex:mlp-backpropagation}.
    Intuitively, $n$ is the dimension of input data and $m$ is the dimension of the latent space where $n > m$.
    A simple autoencoder can be written as $R_{\autoencoder} = \opEnDe{\ell_0, \ell_1}(x_{\inp})$ where $\loctyp{\ell_0}{(n_{\inp}, m)}$ and $\loctyp{\ell_1}{(m, n_{\inp})}$.
    We exploit $Q^{\fwd}_{\opEnDe{\ell, \ell'}}$ to implement encoding layers and $Q^{\bwd}_{\opEnDe{\ell, \ell'}}$ to implement decoding layers.
    We define $H_{\autoencoder} = \{ x \mapsto \pret{x} \} \cup \{ (x \mapsto Q^{\fwd}_{\opEnDe{\ell, \ell'}} \mid y, z \mapsto Q^{\bwd}_{\opEnDe{\ell, \ell'}}) \}$ where
    \[
        Q^{\fwd}_{\opEnDe{\ell, \ell'}}
        =
        \plet \; y \pbind \opFC{\ell}(x) \; \pin \; \pret{\tuple{y, \fromvec{0}}}
        \quad \text{and} \quad
        Q^{\bwd}_{\opEnDe{\ell, \ell'}}
        =
        \opFC{\ell'}(y).
    \]
    We have
    \[
    \infer{
        \emptyenv \asemicolon x : \typreal(n_{\inp})
        \vdash \prevhandle(x)\tuple{x_{\inp}}. \opEnDe{\ell_0,\ell_1}(x_{\inp}) \; \pwith \; H_{\autoencoder}
        : \typreal(n_{\inp}).
    }{
        %\infer{
            \emptyenv \asemicolon x_{\inp} : \typreal(n_{\inp})
            \vdash \opEnDe{\ell_0, \ell_1}(x_{\inp}) : \typreal(m)
        %}{
        %    \emptyenv, x_{\inp} : \typreal(n_{\inp})
        %    \vdash x_{\inp} : \typreal(n_{\inp})
        %}
        &
        %\infer{
            \revhandlerjudgment H_{\autoencoder} : \typreal(m)
        %}{
        %    \begin{aligned}
        %        & x : \typreal(m) \vdash \pret{x} : \typreal(m) \\
        %        & x : \typreal(n) \vdash Q^{\fwd}_{\opEnDe{\ell, \ell'}} : \typreal(m) \times \typreal(0) \\
        %        & y : \typreal(m), z : \typreal(0) \vdash Q^{\bwd}_{\opEnDe{\ell, \ell'}} : \typreal(n)
        %    \end{aligned}
        %}
    }
    \]
    The command $\prevhandle(x)\tuple{x_{\inp}}.R_{\autoencoder} \; \pwith \; H_{\autoencoder}$ is an implementation of a simple autoencoder.
    This use of the reverse handler is interesting because $Q^{\bwd}_{\opEnDe{\ell, \ell'}}$ is not the derivative of $Q^{\fwd}_{\opEnDe{\ell,\ell'}}$ in the usual sense.
    Furthermore, we can implement the backpropagation for the autoencoder using a handler $H = \{ x \mapsto \pret{\fromvec{\learningrate} \cdot (x - \fromvec{v})} \} \cup \{ (x \mapsto Q^{\fwd}_{\opFC{\ell}} \mid y,z \mapsto Q^{\bwd}_{\opFC{\ell}}) \}$ where $v \in \Real^{n_{\inp}}$ is the input data, $\learningrate \in \Real$ is the learning rate and $Q^{\fwd}_{\opFC{\ell}}$ and $Q^{\bwd}_{\opFC{\ell}}$ are as in Example~\ref{ex:mlp-backpropagation}:
    \[
        \prevhandle(\fromvec{v}) \tuple{x_{\inp}}.
        \left(\prevhandle(x)\tuple{x_{\inp}}. R_{\autoencoder} \; \pwith \; H_{\autoencoder} \right)
        \; \pwith \; H_{\MLP}.
    \]
    Deeper autoencoders can be implemented by connecting the operation $\opEnDe{\ell, \ell'}$,
    for example $R'_{\autoencoder} = \plet \; x' \pbind \opEnDe{\ell_0, \ell_3}(x_{\inp})\; \pin \; \opEnDe{\ell_1, \ell_2}(x')$ where
    $\loctyp{\ell_0}{(n_{\inp}, m)}$,
    $\loctyp{\ell_1}{(m, m')}$,
    $\loctyp{\ell_2}{(m', m)}$,
    $\loctyp{\ell_3}{(m, n_{\inp})}$ and
    $n_{\inp} > m > m'$.

    The fact that an autoencoder can be implemented by pairing its encoding and decoding components via $H_{\autoencoder}$ is not merely a consequence of the reverse handler possessing unnecessarily high expressive power.
    There exists a relationship \cite{Cruttwell+2022} between CRDCs and lenses, the latter of which serve as a model of bidirectional computation.
    We believe that the implementation of the autoencoder is supported by constructions arising from this relationship with bidirectional computation.
    A detailed investigation of this relationship is left for future work.
\end{example}

\begin{example}[convolutional neural network]
    \label{ex:convolutional-neural-network}
    A convolutional neural network (CNN) is a neural network architecture based on \emph{convolution} of data and filters.
    A CNN has many practical applications such as image recognition \cite{Krizhevsky+2017}, genome analysis \cite{Zhou+2015}, and natural language processing \cite{Kim2014}.
    Here, we consider one-dimensional convolutions.
    We add the following operations to $\signatopr$ in Example~\ref{ex:mlp-backpropagation}:
    \[
    \begin{aligned}
        \opConv{n, \ell} & : \typreal(c n) \opto \typreal(c' (n-m+1))
        && \text{where $n \in \Nat$ and $\loctyp{\ell}{(m, c, c')} \in \Locations$},
        \\
        \opPool{n, m, c} & : \typreal(c n) \opto \typreal(c \lceil n/m \rceil)
        && \text{where $n, m, c \in \Nat$}.
    \end{aligned}
    \]
    Intuitively, the number $n$ in $\opConv{n, \ell}$ is the size of input data, $m$ is the size of the filter, and $c$ and $c'$ are the numbers of input and output channels, respectively, where a channel is one component of the data representation (e.g., $c=1$ for grayscale data and $c=3$ for RGB data).
    The number $n$ in $\opPool{n, m, c}$ is the size of input data, $m$ is the size of the pooling window, and $c$ is the number of channels. Here, the stride is assumed to be equal to the size of the pooling window.
    For example, a simple CNN consists of a convolution layer, an activation layer using the Swish function, a pooling layer and a fully connected layer can be written as
    \[
    R_{\CNN} =
    \left(
    \begin{aligned}
        & \plet \; y' \pbind \opConv{n_{\inp},\ell_0}(x_{\inp}) \; \pin \; \\
        & \plet \; x' \pbind \pret{\fnSwish_{c'n'}(y')} \; \pin \; \\
        & \plet \; x'' \pbind \opPool{n', m', c'}(x') \; \pin \;
        \opFC{\ell_1}(x'')
    \end{aligned}
    \right) \quad
    \text{where }
    \left\{
    \begin{aligned}
        & \loctyp{\ell_0}{(m,c,c')} \\
        & n' = n - m + 1 \\
        & \loctyp{\ell_1}{(c'\lceil n' / m' \rceil, n'')}.
    \end{aligned}
    \right.
    \]
    We can define $Q^{\fwd}_{\opConv{n, \ell}}$, $Q^{\bwd}_{\opConv{n, \ell}}$, $Q^{\fwd}_{\opPool{n,m,c}}$ and $Q^{\bwd}_{\opPool{n,m,c}}$ appropriately and obtain a handler $H_{\CNN}$ by adding them to the handler $H_{\MLP}$ in Example~\ref{ex:mlp-backpropagation}.
    See Example~\ref{ex:convolutional-neural-network-H-CNN} for the detail of $H_{\CNN}$.
    The following command trains the CNN $R_{\CNN}$ by backpropagation:
    \[
    \prevhandle(\fromvec{v}) \tuple{x_{\inp}}. R_{\CNN} \; \pwith \; H_{\CNN}.
    \]
    The operations $\opConv{n, \ell}$ and $\opPool{n,m,c}$ are implemented by the handler $H_{\CNN}$ using functions $\fnconv$ and $\fnpool$ in $\signatfunc$, respectively.
    See Example~\ref{ex:convolutional-neural-network-H-CNN} for details.
\end{example}

\begin{example}[U-net]
    \label{ex:U-net}
    A U-net \cite{Ronneberger+2015} is a neural network architecture for image segmentation, image denoising and image generation.
    A U-net has an encoder-decoder like structure with skip connections.
    We show that U-net can be implemented in our language using nested reverse handlers.
    We add the following function symbols to $\signatfunc$:
    \begin{gather*}
        \fnpadding_{c, n, m} : \typreal(c n) \to \typreal (c m),
        \quad
        \fnupscale_{c, n, m} : \typreal(c n) \to \typreal (c m),
        \\
        \fnconcat_{c, c', n} : \typreal(c n) \times \typreal(c' n) \to \typreal((c + c') n).
    \end{gather*}
    Intuitively, $\fnpadding_{c, n, m}(x)$ is a one-dimensional image of size $m$ with $c$ channels obtained by padding the image $x$ (of size $n$ with $c$ channels) with zeros.
    An image $\fnupscale_{c, n, m}(x)$ is an image of size $m$ with $c$ channels obtained by upscaling the image $x$ (of size $n$ with $c$ channels).
    An image $\fnconcat_{c, c', n}(\tuple{x, y})$ is an image of size $n$ with $c+c'$ channels obtained by concatenating images $x$ and $y$ of size $n$ with $c$ and $c'$ channels, respectively.

    Let us add operations for U-nets.
    For the sake of simplicity, we only use convolution operation of filter size $3$ and use pooling operation of window size $2$.
    For each $n, c, c' \in \Nat$ and $\loctyp{\ell}{(3, c, c')}, \loctyp{\ell'}{(1, c + c', c)}, \loctyp{\ell''}{(3, c, c)} \in \Locations$,
    we add an operation $\opU{n, \ell, \ell', \ell''} : \typreal(c n) \opto \typreal(c' \lceil (n - 2) / 2 \rceil)$ to $\signatopr$.
    A U-net has a symmetric architecture with encoder-like and decoder-like layers.
    With a reverse handler, it suffices to implement only the encoder; the decoder is generated automatically.
    We define $R_{\Unet}$ as follows:
    \[
    R_{\Unet} =
    \left(
    \begin{aligned}
        & \plet \; x' \pbind \opU{n_0, \loctyp{\ell_0}{(3, 1, 8)}, \loctyp{\ell'_0}{(1, 9, 8)}, \loctyp{\ell''_0}{(3, 8, 8)}}(x_{\inp}) \; \pin \; \\
        & \plet \; x'' \pbind \opU{n_1, \loctyp{\ell_1}{(3, 8, 16)}, \loctyp{\ell'_1}{(1, 24, 8)}, \loctyp{\ell''_1}{(3, 8, 8)}}(x') \; \pin \; \pret{x''}
    \end{aligned}
    \right)
    \]
    where $n_0 = 134$ and $n_1 = \lceil (n_0 - 2) / 2 \rceil$.
    We define a reverse handler $H_{\Unet} = \{ x \mapsto P \} \cup \{ (x \mapsto Q^{\fwd}_{\opU{n, \ell, \ell', \ell''}} \mid y, z \mapsto Q^{\bwd}_{\opU{n, \ell, \ell', \ell''}}) \}_{n,\ell, \ell', \ell''}$ where
    \[
    \begin{aligned}
        P
        & =
        \opConv{n_2, \loctyp{\ell_2}{(3, 16, 16)}}(\fnpadding_{16, 32, 34}(x))
        \quad \text{where} \quad n_2 = \lceil (n_1 - 2) / 2 \rceil,
        \\
        Q^{\fwd}_{\opU{n, \ell, \ell', \ell''}}
        & =
        \plet \; y \pbind \opConv{n, \ell}(x) \; \pin \;
        \plet \; z \pbind \opPool{n-2, 2, c'}(y) \; \pin \;
        \pret{\tuple{z, x}},
        \\
        Q^{\bwd}_{\opU{n, \ell, \ell', \ell''}}
        & =
        \left(
        \begin{aligned}
            & \plet \; x \pbind \opConv{n, \ell'}(\fnconcat_{c,c',n}(\tuple{y, \fnupscale_{c', \lceil (n - 2) / 2 \rceil, n}(z)})) \; \pin \\
            & \plet \; x' \pbind \opConv{n, \ell''}(x) \; \pin \; \pret{\fnpadding_{c, n - 2, n}(\fnSwish(x'))}
        \end{aligned}
        \right).
    \end{aligned}
    \]
    % The following judgments are derivable:
    % \[
    % \begin{aligned}
    %     & \emptyenv \asemicolon x_{\inp} : \typreal(134) \vdash R_{\Unet} : \typreal(16 \times 32),
    %     \\
    %     & \emptyenv \asemicolon x : \typreal(16 \times 32)
    %     \vdash P : \typreal(16 \times 32),
    %     \\
    %     & \emptyenv \asemicolon x : \typreal(c n)
    %     \vdash Q^{\fwd}_{\opU{n, \ell, \ell', \ell''}} : \typreal(c' \lceil (n - 2) / 2 \rceil) \times \typreal(c n),
    %     \\
    %     & \emptyenv \asemicolon y : \typreal(c' \lceil (n - 2) / 2 \rceil), z : \typreal(c n)
    %     \vdash Q^{\bwd}_{\opU{n, \ell, \ell', \ell''}} : \typreal(3 n),
    %     \\
    %     & \revhandlerjudgment H_{\Unet} : \typreal(16 \times 32).
    % \end{aligned}
    % \]
    We obtain a concrete U-net by combining $R_{\Unet}$ with the reverse handler $H_{\Unet}$:
    \[
    \emptyenv \asemicolon x : \typreal(134) \vdash
    \prevhandle(x) \tuple{x_{\inp}}. R_{\Unet} \; \pwith \; H_{\Unet} : \typreal(134).
    \]
    To train the U-net, we can use the reverse handler $H_{\CNN}$ in Example~\ref{ex:convolutional-neural-network}: for $v \in \Real^{134}$,
    \[
    \prevhandle(\fromvec{v}) \tuple{x}.
    \left(
        \prevhandle(x_{\inp}) \tuple{x_{\inp}}.
        R_{\Unet}
        \; \pwith \; H_{\Unet}
    \right)
    \; \pwith \; H_{\CNN} : \typreal(134).
    \]
\end{example}

\begin{example}[quantization aware training via straight-through estimator]
    \label{ex:ste}
    To demonstrate the flexibility of reverse handlers, we implement the straight-through estimator (STE) \cite{BengioLeonardCourville2013} for the backpropagation.
    STE is used as a technique for quantization-aware training (QAT) of neural networks.
    Suppose that a set of operation symbols $\signatopr$ contains an operation symbol $\oprlayer : \typreal(n) \opto \typreal(n)$ representing a programmable layer of a neural network, and
    that we want to use the layer $\oprlayer$ in order to apply a function $\func : \typreal(n) \to \typreal(n)$ to a vector during the forward pass.
    We can implement backpropagation of such behavior of $\oprlayer$ by a handler $H$ with $Q^{\fwd}_{\oprlayer}$ and $Q^{\bwd}_{\oprlayer}$, which are defined as follows:
    \[
    \begin{aligned}
        & \emptyenv \asemicolon x : \typreal(n) \vdash Q^{\fwd}_{\oprlayer} \defeq \pret{\tuple{\func(x), x}} : \typreal(n) \times \typreal(n),
        \\
        & \emptyenv \asemicolon y : \typreal(n), z : \typreal(n) \vdash Q^{\bwd}_{\oprlayer} \defeq \pret{\rd[\func](\tuple{z, y})} : \typreal(n).
    \end{aligned}
    \]

    Quantization is a technique to reduce the precision of weights or input data in neural networks to reduce the memory footprint and computation cost.
    There are several quantization methods and techniques.
    In QAT, one technique is that we apply a quantization function $\func$ to quantize the input vector.
    For example, a round function $\fnround_n : \typreal(n) \to \typreal(n)$ that maps each entry of a real vector to the nearest integer can be used as a quantization function $\func$.
    The problem is that the derivative of the round function is zero almost everywhere, so we cannot use the backpropagation algorithm to train the neural network with the round function.

    STE approximates the round function by the identity function.
    Hence, we can implement the operation $\oprlayer$ using STE by defining $Q^{\fwd}_{\oprlayer}$ and $Q^{\bwd}_{\oprlayer}$ as follows:
    \[
        \begin{aligned}
            & \emptyenv \asemicolon x : \typreal(n) \vdash Q^{\fwd}_{\oprlayer} \defeq \pret{\tuple{\fnround_n(x), \fromvec{0}}} : \typreal(n) \times \typreal(n),
            \\
            & \emptyenv \asemicolon y : \typreal(n), z : \typreal(n) \vdash Q^{\bwd}_{\oprlayer} \defeq \pret{y} : \typreal(n).
        \end{aligned}
    \]
%    The point here is that even if we employ the smooth approximation of the round function as the interpretation of $\fnround$,
%    $Q^{\bwd}_{\oprlayer}$ is not the derivative of the interpretation of $Q^{\fwd}_{\oprlayer}$ in the usual sense: 
%    we can intentionally choose $Q^{\bwd}_{\oprlayer}$ to be the identity function, instead of the derivative of the smooth approximation of the round function.
    % In the reverse differential category approach \cite{AbadiPlotkin2019,Cockett+2020}, we cannot capture functions with vanishing gradients such as the round function.
    This example shows that reverse handlers are flexible enough to implement various behaviors of neural networks, including QAT by the STE.
\end{example}

\section{Operational Semantics}\label{sec:operational-semantics}
We define call-by-value small-step operational semantics for our language.
The definition has two parts: reduction for terms and for commands.
\subsection{Operational Semantics for Terms}
The operational semantics for terms is a simplified version of that of the simple differentiable programming language \cite{AbadiPlotkin2019}.
We write $\Val_A = \{ V \mid {\emptyenv \vdash V : A}\}$ for the set of closed values of type $A$, and $\signatfunc(A, B) = \{ \func \in \signatfunc \mid \func : A \to B\}$ for the set of function symbols in $\signatfunc$ of type $A \to B$.
We assume that there is an \emph{evaluation} partial function
$\Eval_{A, B} \colon \signatfunc(A, B) \times \Val_A \to \Val_B$ ($A \ne 1$)
and that the set of constant symbols of type $\beta$ has a commutative monoid structure $(\signatfunc(1, \beta), +_{\beta}, 0_{\beta})$ for each base type $\beta \in \BType$.

\begin{definition}
    The \emph{evaluation contexts} $\ctxe$ for terms are defined by the following BNF:
    \begin{equation*}
        \begin{aligned}
            \ctxe
            & \bnfeq
            [\blank]
            \mid \func(\ctxe)
            \mid \ctxe + M
            \mid V + \ctxe
            \mid \tproj_i(\ctxe)
            \mid \tlet \; x \tbind \ctxe \; \tin \; M
            \mid \ctxe.\trd(x.N)(L)
            \mid W.\trd(x.N)(\ctxe).
        \end{aligned}
    \end{equation*}
\end{definition}

The most non-trivial part of the operational semantics for terms is term rewriting $\rewriterd{W}{V}(x.N)$, which is used in the reduction rule for reverse derivatives.
The original term rewriting \cite{AbadiPlotkin2019} is defined on \emph{trace terms} to trace conditional branches.
We do not have conditional branches in our language for simplicity.
Hence, we can give a simpler definition of term rewriting without trace terms.
\begin{definition}[\cite{AbadiPlotkin2019}]
    For well-typed values $\Gamma \vdash W : B$ and $\Gamma \vdash V : A$ and a well-typed term $\Gamma, x : A \vdash N : B$,
    we define a term $\rewriterd{W}{V}(x.N)$ in Fig.~\ref{fig:rewriting-rule-term}
    \begin{figure}
    \footnotesize
    \[
    \begin{aligned}
        %\rewriterd{W}{V}(x.y)
        %& =
        %\begin{cases}
        %    W & (x = y) \\
        %    0_\beta & (x \neq y)
        %\end{cases}
        \rewriterd{W}{V}(x.x)
        & = W
        \qquad
        \rewriterd{W}{V}(x.y)
        = 0_\beta \quad (x \neq y)
        %\\
        \qquad
        \rewriterd{W}{V}(x.\const)
        =
        0_{\beta}
        \\
        \rewriterd{W}{V}(x.\func(M))
        & =
        \tlet \; x \tbind V \; \tin \;
        \tlet \; y \tbind \rd[\func](\tuple{M, W}) \; \tin \;
        \rewriterd{y}{V}(x.M)
        \\
        \rewriterd{W}{V}(x. \tlet \; y \tbind M \; \tin \; N)
        & =
        \tlet \; x \tbind V \; \tin \;
        \tlet \; y \tbind M \; \tin \;
        \rewriterd{W}{V}(x.N) +
            \tlet \; y' \tbind \rewriterd{W}{y}(y.N) \; \tin \;
            \rewriterd{y'}{V}(x.M)
        \\
        \rewriterd{W}{V}(x.\tuple{M_1, \dots, M_n})
        & =
        \tlet \; \tuple{y_1, \dots, y_n} \tbind W \; \tin \;
        \rewriterd{y_1}{V}(x.M_1) + \cdots + \rewriterd{y_n}{V}(x.M_n)
        \\
        \rewriterd{W}{V}(x.\tproj_i(M))
        & =
        \tlet \; x \tbind V \; \tin \;
        \tlet \; y \tbind M \; \tin \;
        \rewriterd{\tuple{U_k}_{k=1}^n}{V}(x.M)
        \quad \text{where $U_i = W$ and $U_k = 0_{\beta}$ ($k \ne i$)
        %\begin{cases}
        %    W & (i = k) \\
        %    0_{\beta} & (i \neq k)
        %\end{cases}
        }
        \\
        \rewriterd{W}{V}(x.(M.\trd(y.N)(L)))
        & =
        \rewriterd{W}{V}\left(x.
            \tlet \; z \tbind M \; \tin \;
            \tlet \; w \tbind L \; \tin \;
            \rewriterd{z}{w}(y.N)
        \right)
    \end{aligned}
    \]
    \caption{Rewriting rules for terms, which is a simplified version of Abadi and Plotkin's language \cite{AbadiPlotkin2019}.}
    \label{fig:rewriting-rule-term}
    \Description{Rewriting rules for terms.}
    \end{figure}
\end{definition}
Intuitively, the term $\rewriterd{W}{V}(x.N)$ denote the reverse-mode derivative of the function $x : A \mapsto N : B$, at $V : A$, evaluated at $W : B$.
In particular, the definition for $\tlet$-terms captures the chain rule of differentiation.
\begin{lemma}[{\cite[Proposition~3.3]{AbadiPlotkin2019}}]
    For well-typed values $\Gamma \vdash W : B$ and $\Gamma \vdash V : A$ and a well-typed term $\Gamma, x : A \vdash N : B$,
    the judgment $\Gamma \vdash \rewriterd{W}{V}(x.N) : A$ is derivable.
\end{lemma}
The following definition gives the small-step reduction for terms.
\begin{definition}
    The \emph{small-step reduction} $(\redto) \subseteq \{ M \mid {\emptyenv \vdash M : A} \}^2$ for terms is defined as follows:
    \begin{gather*}
        \func(V) \redto \Eval(\func, V)
        \text{ when $\Eval(\func,V)$ is defined},
        \quad
        \const_1 + \const_2 \redto \const_1 +_{\beta} \const_2,
        \quad
        \tproj_i(\tuple{V_1, \dots, V_n}) \redto V_i,
        \\
        \tlet \; x \tbind V \; \tin \; M \redto M[V / x],
        \quad
        W.\trd(x.N)(V) \redto \rewriterd{W}{V}(x.N),
        \quad
        \ctxe[M] \redto \ctxe[N] \quad \text{if $M \redto N$}.
    \end{gather*}
    %\[
    %    \begin{aligned}
    %        \func(V) & \redto \Eval(\func, V)
    %        &
    %        \text{ (where $\Eval(\func,V)$ is defined.)}
    %        \\
    %        \const_1 + \const_2 & \redto \const_1 +_{\beta} \const_2
    %        \\
    %        \tproj_i(\tuple{V_1, \dots, V_n}) & \redto V_i
    %        &
    %        \tlet \; x \tbind V \; \tin \; M & \redto M[V / x]
    %        \\
    %        W.\trd(x.N)(V) & \redto \rewriterd{W}{V}(x.N)
    %        &
    %        \ctxe[M] & \redto \ctxe[N] \quad \text{if $M \redto N$}
    %    \end{aligned}
    %\]
\end{definition}

\subsection{Operational Semantics for Commands}
We define operational semantics for commands.
The most non-trivial reduction rule is the rule for reverse handlers.
\begin{definition}
    The \emph{evaluation contexts} $\ctxfc$ for commands are defined by the following BNF:
    \begin{equation*}
        \begin{aligned}
            \ctxfc
            & \bnfeq
            [\blank]
            \mid \plet \; x \pbind \ctxfc \; \pin \; Q
            \\
            \ctxft
            & \bnfeq \ctxfc[\pret{[\blank]}]
            \mid \ctxfc[\opr([\blank])]
            \mid \ctxfc[\prevhandle([\blank]) \tuple{x_1, \dots, x_n}.R \; \pwith \; H]
        \end{aligned}
    \end{equation*}
\end{definition}
We put a command into the hole of $\ctxfc$ to obtain a command, and put a term into the hole of $\ctxft$ to obtain a command.
The evaluation context $\ctxfc$ is standard, and the evaluation context $\ctxft$ is used to define the reduction rule for reverse handlers.

\begin{lemma}\label{lem:form-of-commands}
    Let $\Gamma = x_1 : A_1, \dots, x_n : A_n$ and $\Delta = y_1 : B_1, \dots, y_m : B_m$.
    For any well-typed command $\Gamma \asemicolon \Delta \vdash P : A$,
    $P$ is exactly one of the following forms:
    \begin{enumerate}
        \item $P = \ctxfc[\pret{z}]$ where $z \in \{ x_1, \dots, x_n, y_1, \dots, y_m \}$.
        \item $P = \ctxfc[\opr(z)]$ where $z \in \{ x_1, \dots, x_n, y_1, \dots, y_m \}$.
        \item $P = \ctxfc[\prevhandle(y_j) \tuple{y'_i}_{i=1}^{l}. P' \; \pwith \; H]$.
        \item $P = \ctxft[M]$ where $M$ is not a variable.
    \end{enumerate}
\end{lemma}
\begin{appendixproof}[Proof of Lemma~\ref{lem:form-of-commands}]
    By induction on the structure of $P$.
\end{appendixproof}

We introduce rewriting rules for commands in order to define the reduction rule for reverse handlers.
The rules are inspired by the construction of reverse algebras in \S\ref{sec:denotational-semantics}.
In the rewriting rule for $\ctxfc[\opr(y_i)]$, $Q^{\fwd}_{\opr}$ and $Q^{b}_{\opr}$ sandwich the continuation.
\begin{definition}
    Let $\Gamma$ and $\Delta = y_1 : B_1, \dots, y_m : B_m$ be a typing environment.
    For a command $\emptyenv \asemicolon \Delta \vdash R : C$,
    a value $\Delta' \vdash V : \typprod_{i=1}^{m} B_i$,
    and a handler $\revhandlerjudgment H = \{ x \mapsto P \} \cup \{ (x \mapsto Q^{\fwd}_{\opr} \mid y, z \mapsto Q^{\bwd}_{\opr}) \}_{\opr \in \signatopr} : C$,
    we define a command $\rewriterevh{V}{H}(\tuple{y_i}_{i=1}^m . R)$ by Lemma~\ref{lem:form-of-commands} and the recursion on $(\max\{\rhdepth{R}, \rhdepth{H}\}, \dsize{\Gamma \asemicolon \Delta \vdash R : C}) \in (\Nat^2, \le_{\Nat^2})$ in Fig.~\ref{fig:rewriterevh}.
    \begin{figure}
        \footnotesize
        \begin{align*}
            & \rewriterevh{V}{H}(\tuple{y_i}_{i=1}^n.\pret{y_k})
            = \plet \; y \pbind P[V_k/x] \; \pin \; \pret{\tuple{W_i}_{i=1}^n}
            \quad \text{where $W_k = y$ and $W_i = \fromvec{0}$ $(i \neq k)$}
            \\
            & \rewriterevh{V}{H}(\tuple{y_i}_{i=1}^m.\ctxfc[\plet \; x \pbind \pret{y_j} \; \pin \; R])
            = \rewriterevh{V}{H}(\tuple{y_i}_{i=1}^m.\ctxfc[R[y_j / x]])
            \\
            & \rewriterevh{V}{H}(\tuple{y_i}_{i=1}^m.\ctxfc[\opr(y_i)])
            = \left(
                \begin{aligned}
                    & \plet \; \tuple{y,z'} \pbind Q^{\fwd}_{\opr}[V_i/x] \; \pin \\
                    & \plet \; \tuple{y', y'_1, \dots, y'_m} \pbind \rewriterevh{\tuple{y, V_1, \dots, V_m}}{H}(\tuple{y, y_1, \dots, y_m}.\ctxfc[\pret{y}]) \; \pin\\
                    & \plet \; y''_i \pbind Q^{\bwd}_{\opr}[y'/y, z'/z] \; \pin \\
                    & \pret{\langle y'_1, \dots, y'_{i-1}, y'_i + y''_i, y'_{i+1}, \dots, y'_m \rangle}
                \end{aligned}
            \right)
            \\
            & \rewriterevh{V}{H}(\tuple{y_i}_{i=1}^m. \ctxfc[\prevhandle(W) \tuple{z_j}_{j=1}^l. R \; \pwith \; H'])
            = \rewriterevh{V}{H}(\tuple{y_i}_{i=1}^m. \ctxfc[\rewriterevh{W}{H'}(\tuple{z_j}_{j=1}^l. R]))
            \\
            & \rewriterevh{V}{H}(\tuple{y_i}_{i=1}^m.\ctxft[M])
            = \left(
                \begin{aligned}
                    & \plet \; y \pbind \pret{M[V_1 / y_1, \dots, V_n / y_n]} \; \pin \\
                    & \plet \; (z, y'_1, \dots, y'_n) \pbind \rewriterevh{\tuple{y, V_1 \dots, V_n}}{H}(\tuple{y, y_1, \dots, y_n}.\ctxft[y]) \; \pin \\
                    & \pret{\langle y'_i + z.\trd(y_i.M[V_1/y_1, \dots, V_{i-1}/y_{i-1}, V_{i+1}/y_{i+1}, \dots, V_m/y_m])(V_i) \rangle_{i = 1}^m}
                \end{aligned}
            \right)
            \\
            & \qquad \text{where $M$ is not a variable}
        \end{align*}
    \caption{Rewriting rules for reverse handlers}
    \label{fig:rewriterevh}
    \Description{Rewriting rules for reverse handlers}
    \end{figure}
\end{definition}
Intuitively, $\rewriterevh{V}{H}(\tuple{y_i}_{i=1}^m.R)$ gives an operational account of the interpretation of reverse handlers in~\eqref{eq:reverse-handler-algebra}, instantiated with $X = \itp{\Delta}$ and $Z = \itp{C}$.
It describes how the computation $R$ is executed under the handler $H$, and then applies the resulting transformation on $\Delta$ to the value $V$.

The most interesting rewriting rule is the one for $\ctxfc[\opr(y)]$.
For example, under the setting of Example~\ref{ex:mlp-backpropagation}, we have
\[
\begin{aligned}
    & \phandle(\fromvec{v}) \tuple{x}. \opFC{\ell}(x) \; \pwith \; H_{\MLP}
    \qquad \text{where $v \in \Real^n$ and $\loctyp{\ell}{(n, m)}$}
    \\
    & \redto
    \rewriterevh{\fromvec{v}}{H_{\MLP}}(\opFC{\ell}(x))
    =
    \left(
    \begin{aligned}
        & \plet \; \tuple{y, z'} \pbind Q^{\fwd}_{\opFC{\ell}}[\fromvec{v} / x] \; \pin \; \\
        & \plet \; \tuple{y', y'_1} \pbind
        (\plet \; y \pbind P[y / x] \; \pin \; \pret{\tuple{y, \fromvec{0}}}) \; \pin \; \\
        & \plet \; y''_1 \pbind Q^{\bwd}_{\opFC{\ell}}[y' / y, z' / z] \; \pin \;
        \pret{y'_1 + y''_1}
    \end{aligned}
    \right).
\end{aligned}
\]
In the above reduction, the operation $\opFC{\ell}$ is handled and is implemented by the forward computation $Q^{\fwd}_{\opFC{\ell}}$ and the backward computation $Q^{\bwd}_{\opFC{\ell}}$ defined in the handler $H_{\MLP}$.

The following lemma says that the rewriting $\rewriterevh{V}{H}(\tuple{y_i}_{i=1}^m.R)$ strictly reduces the depth of reverse handlers. It is crucial in the induction for the proof of type preservation.
\begin{lemma}\label{lem:RH-depth}
    Let $\Delta = y_1 : B_1, \dots, y_m : B_m$ be a typing environment.
    For a command $\emptyenv \asemicolon \Delta \vdash R : C$,
    a value $\Delta' \vdash V : \typprod_{i=1}^m B_i$ and
    a handler $\revhandlerjudgment H : C$,
    we have $\rhdepth{R} + \rhdepth{H} \ge \rhdepth{\rewriterevh{V}{H}(\tuple{y_j}_{j=1}^m. R)}$.
    Hence, we have
    \[ \rhdepth{\prevhandle(V) \tuple{y_i}_{i=1}^m.R \; \pwith \; H} > \rhdepth{\rewriterevh{V}{H}(\tuple{y_i}_{i=1}^m.R)}. \]
\end{lemma}
\begin{appendixproof}[Proof of Lemma~\ref{lem:RH-depth}]
    By induction on $(\max\{\rhdepth{P}, \rhdepth{H}\}, \dsize{\emptyenv \asemicolon \Delta \vdash P : A}) \in (\Nat^2, \le_{\Nat^2})$.
\end{appendixproof}

The following lemma states that the rewriting $\rewriterevh{V}{H}(\tuple{y_i}_{i=1}^m.R)$ preserves typing.
\begin{lemma}\label{lem:type-preservation-RH}
    Let $\Delta = y_1 : B_1, \dots, y_m : B_m$ be a typing environment.
    For a command $\emptyenv \asemicolon \Delta \vdash R : C$,
    a value $\Delta' \vdash V : \typprod_{i=1}^m B_i$ and
    a handler $\revhandlerjudgment H : C$,
    we have
    $\emptyenv \asemicolon \Delta' \vdash \rewriterevh{V}{H}(\tuple{y_i}_{i=1}^m.R) : \typprod_{i=1}^m B_i$.
\end{lemma}
\begin{appendixproof}[Proof of Lemma~\ref{lem:type-preservation-RH}]
    By induction on $(\max\{\rhdepth{R}, \rhdepth{H}\}, \dsize{\emptyenv \asemicolon \Delta \vdash R : C}) \in (\Nat^2, \le_{\Nat^2})$.
    We perform case analysis on the structure of $\emptyenv \asemicolon \Delta \vdash R : C$.
    Let $H = \{ x \mapsto P \} \cup \{ (x \mapsto Q_{\opr}^{\fwd} \mid y, z \mapsto Q_{\opr}^{\bwd}) \}_{\opr \in \signatopr}$.

    \textbf{Case} $R = \pret{y_k}$ for some $1 \le k \le m$.
    We have $\emptyenv \asemicolon \Delta \vdash \pret{y_k} : B_k$ and
    \[
        \rewriterevh{V}{H}(\tuple{y_i}_{i=1}^m.\pret{y_k})
        = \plet \; y \pbind P[V_k/x] \; \pin \; \pret{\tuple{W_i}_{i=1}^m}
        \quad \text{where $W_i = \begin{cases}
            \fromvec{0} & (i \neq k) \\
            y & (i = k)
        \end{cases}$}.
    \]
    We have
    \[
    \infer[\tyclet]{
        \emptyenv \asemicolon \Delta' \vdash \plet \; y \pbind P[V_i/x] \; \pin \; \pret{\tuple{W_i}_{i=1}^m} : \typprod_{i=1}^m B_i.
    }{
        \emptyenv \asemicolon \Delta' \vdash P[V_k/x] : B_k
        &
        \infer[\tycpure]{
            \emptyenv \asemicolon y : B_k, \Delta'  \vdash \pret{\tuple{W_i}_{i=1}^m} : \typprod_{i=1}^m B_i
        }{
            \infer[\tyttuple]{
                y : B_k, \Delta' \vdash \tuple{W_i}_{i=1}^m : \typprod_{i=1}^m B_i
            }{
                y : B_k, \Delta' \vdash W_i : B_i \quad (1 \leq i \leq m)
            }
        }
    }
    \]

    \textbf{Case} $R = \ctxfc[\plet \; x \pbind \pret{y_i} \; \pin \; R']$.
    We have
    \[
        \rewriterevh{V}{H}(\tuple{y_i}_{i=1}^m.\ctxfc[\plet \; x \pbind \pret{y_i} \; \pin \; R'])
        = \rewriterevh{V}{H}(\tuple{y_i}_{i=1}^m.\ctxfc[R'[y_i / x]])
    \]
    and
    $(\rhdepth{R}, \dsize{R} ) > (\rhdepth{\ctxfc[R'[y_i/x]]}, \dsize{\ctxfc[R'[y_i/x]]} )$.
    Thus, by the induction hypothesis, we have
    $\emptyenv \asemicolon \Delta' \vdash \rewriterevh{V}{H}(\tuple{y_i}_{i=1}^m.\ctxfc[R'[y_i / x]]) : \typprod_{i=1}^m B_i$.

    \textbf{Case} $R = \ctxfc[\opr(y_i)]$.
    We have
    \[
    \rewriterevh{V}{H}(\tuple{y_i}_{i=1}^m.\ctxfc[\opr(y_i)])
    =
    \left(\begin{aligned}
        & \plet \; \tuple{y,z'} \pbind Q^{\fwd}_{\opr}[V_i/x] \; \pin \\
        & \plet \; \tuple{y', y'_1, \dots, y'_l} \pbind \rewriterevh{\tuple{y, V_1, \dots, V_m}}{H}(\tuple{y, y_1, \dots, y_m}.\ctxfc[\pret{y}]) \; \pin\\
        & \plet \; y''_i \pbind Q^{\bwd}_{\opr}[y'/y, z'/z] \; \pin \\
        & \pret{\langle y'_1, \dots, y'_{i-1}, y'_i + y''_i, y'_{i+1}, \dots, x'_m \rangle}
    \end{aligned}\right).
    \]
    Let $R' = \rewriterevh{\tuple{y, V_1, \dots, V_m}}{H}(\tuple{y, y_1, \dots, y_m}.\ctxfc[\pret{y}])$,
    $R'' = \pret{\langle y'_1, \dots, y'_{i-1}, y'_i + y''_i, y'_{i+1}, \dots, y'_m \rangle}$,
    $R_1 = \plet \; y''_i \pbind Q^{\bwd}_{\opr}[y'/y, z'/z] \; \pin \; R''$,
    $R_2 = \plet \; \tuple{y', y'_1, \dots, y'_l} \pbind R' \; \pin \; R_1$,
    $R_3 = \plet \; \tuple{y,z'} \pbind Q^{\fwd}_{\opr}[V_i/x] \; \pin \; R_2$ and
    $\Delta'' = y' : B_1, \dots, y'_m : B_m$.
    
    We have $(\rhdepth{R}, \dsize{R}) > (\rhdepth{\ctxfc[\pret{y}]}, \dsize{\ctxfc[\pret{y}]})$.
    By the induction hypothesis, we have
    $ \emptyenv \asemicolon y : B, z : D_{\opr}\vdash R' : C \times \typprod_{i=1}^m B_i$.
    Hence, we have
    \[
    \scalebox{0.9}{$
    \infer{
        \emptyenv \asemicolon y : B, z : D_{\opr}, y' : C, \Delta'' \vdash
        R_1
        : \typprod_{i=1}^m B_i,
    }{
        \emptyenv \asemicolon y : B, z : D_{\opr}, y' : C, \Delta'' \vdash
        Q^{\bwd}_{\opr}[y'/y, z'/z] : B_i
        &
        \emptyenv \asemicolon y''_i : B_i, y : B, z : D_{\opr}, y' : C, \Delta'' \vdash
        R''
        : \typprod_{i=1}^m B_i
    }$}
    \]
    \[
        \infer{
            \emptyenv \asemicolon y : B, z : D_{\opr}, \Delta' \vdash 
            R_2
            : \typprod_{i=1}^m B_i
        }{
            \emptyenv \asemicolon y : B, z : D_{\opr}, \Delta' \vdash R' : C \times \typprod_{i=1}^m B_i
            &
            \emptyenv \asemicolon  y' : C, \Delta'', y : B, z : D_{\opr}, \Delta' \vdash
            R_1
            : \typprod_{i=1}^m B_i
        }
    \]
    and
    \[
    \infer{
        \emptyenv \asemicolon \Delta' \vdash
        R_3
        : \typprod_{i=1}^m B_i.
    }{
        \emptyenv \asemicolon \Delta' \vdash Q^{\fwd}_{\opr}[V_i/x] : B \times D_{\opr}
        &
        \emptyenv \asemicolon y : B, z : D_{\opr}, \Delta' \vdash 
        R_2
        : \typprod_{i=1}^m B_i
    }
    \]

    \textbf{Case} $R = \ctxfc[\prevhandle(W) \tuple{z_j}_{j=1}^l. R' \; \pwith \; H']$.
    We have
    \[
        \rewriterevh{V}{H}(\tuple{y_i}_{i=1}^m.\ctxfc[\prevhandle(W) \tuple{z_j}_{j=1}^l. R' \; \pwith \; H'])
        = \rewriterevh{V}{H}(\tuple{y_i}_{i=1}^m. \ctxfc[\rewriterevh{W}{H'}(\tuple{z_j}_{j=1}^l. R')]).
    \]
    By Lemma~\ref{lem:RH-depth}, we have
    $\rhdepth{R} > \rhdepth{\ctxfc[\rewriterevh{W}{H'}(\tuple{z_j}_{j=1}^l. R')]}$.
    Thus, by the induction hypothesis, we have
    $\emptyenv \asemicolon \Delta' \vdash \rewriterevh{V}{H}(\tuple{y_i}_{i=1}^m. \ctxfc[\rewriterevh{W}{H'}(\tuple{z_j}_{j=1}^l. R')]) : \typprod_{i=1}^m B_i$.

    \textbf{Case} $R = \ctxft[M]$.
    We have
    \[
    \begin{aligned}
        \rewriterevh{V}{H}(\tuple{y_i}_{i=1}^m.R)
        & =
        \left(\begin{aligned}
            & \plet \; y \pbind M[V_1/y_1, \dots, V_m/y_m] \; \pin \\
            & \plet \; (z, y'_1, \dots, y'_m) \pbind \rewriterevh{\tuple{y, V_1, \dots, V_m}}{H}(\tuple{y, y_1, \dots, y_m}.\ctxft[\pret{y}]) \; \pin \\
            & \pret{\langle y'_i + z.\trd(y_i.M[V_1/y_1, \dots, V_{i-1}/y_{i-1}, V_{i+1}/y_{i+1}, \dots, V_m/y_m])(V_i) \rangle_{i = 1}^m}
        \end{aligned}\right).
    \end{aligned}
    \]
    By the induction hypothesis, we have
    $ \emptyenv \asemicolon \Delta' \vdash \rewriterevh{\tuple{y, V_1, \dots, V_m}}{H}(\tuple{y, y_1, \dots, y_m}.\ctxft[\pret{y}]) : C \times \typprod_{i=1}^m B_i$.
    Hence, we can derive $\emptyenv \asemicolon \Delta' \vdash
    \rewriterevh{V}{H}(\tuple{y_i}_{i=1}^m.R)
    : \typprod_{i=1}^m B_i$.
\end{appendixproof}

\begin{definition}\label{def:small-step-reduction-commands}
    The \emph{small-step reduction} $P \redto P'$ for closed commands is defined as follows:
    \begin{align*}
        & \ctxft[M]
        \redto
        \ctxft[M']
        \quad \text{if $M \redto M'$},\quad\quad\quad\quad
      \ctxfc[\plet \; x \pbind \pret{V} \; \pin \; R]
        \redto
        \ctxfc[R[V/x]],
        \\
        & \ctxfc[\prevhandle(V) \tuple{x_i}_{i=1}^n. R \; \pwith \; H]
        \redto
        \ctxfc[\rewriterevh{V}{H}(\tuple{x_i}_{i=1}^n.R)].
    \end{align*}
\end{definition}

\begin{toappendix}
\begin{lemma}\label{lem:depth-reduction}
    If $P \redto P'$, then $\rhdepth{P} \ge \rhdepth{P'}$ holds.
    Furthermore, it holds that
    \[ \rhdepth{\ctxfc[\prevhandle(V) \tuple{x_i}_{i=1}^n. R \; \pwith \; H]} > \rhdepth{\ctxfc[\rewriterevh{V}{H}(\tuple{x_i}_{i=1}^n.R)]}. \]
\end{lemma}
\begin{proof}[Proof of Lemma~\ref{lem:depth-reduction}]
    By case analysis on the definition of $\redto$.
    It is obvious that
    $\rhdepth{\ctxft[M]} = \rhdepth{\ctxft[M']}$ and
    $\rhdepth{\ctxfc[\plet \; x \pbind \pret{V} \; \pin \; R]} = \rhdepth{\ctxfc[R[V/x]]}$.
    By Lemma~\ref{lem:RH-depth}, we have
    \[
    \rhdepth{\prevhandle(V) \tuple{x_i}_{i=1}^n. R \; \pwith \; H}
    =
    1 + \rhdepth{R} + \rhdepth{H}
    >
    \rhdepth{\rewriterevh{V}{H}(\tuple{x_i}_{i=1}^n.R)}.
    \]
    Hence, we have
    $\rhdepth{\ctxfc[\prevhandle(V) \tuple{x_i}_{i=1}^n. R \; \pwith \; H]}
    \ge
    \rhdepth{\ctxfc[\rewriterevh{V}{H}(\tuple{x_i}_{i=1}^n.R)]}$.
\end{proof}
\end{toappendix}

The following proposition states that the small-step reduction for commands is type-preserving.
\begin{proposition}[Type preservation]\label{prop:type-preservation-commands}
    If $\emptyenv \asemicolon \emptyenv \vdash P : A$ and $P \redto P'$, then $\emptyenv \asemicolon \emptyenv \vdash P' : A$.
\end{proposition}
\begin{appendixproof}[Proof of Proposition~\ref{prop:type-preservation-commands}]
    By the induction of the structure of $\ctxfc$, we can show that
    if $\emptyenv \asemicolon \emptyenv \vdash \ctxfc[Q] : A$ and $\emptyenv \asemicolon \emptyenv \vdash Q : B$,
    then $\emptyenv \asemicolon \emptyenv \vdash \ctxfc[Q'] : A$ holds for any command $Q'$ satisfying $\emptyenv \asemicolon \emptyenv \vdash Q' : B$.
    Hence, it suffices to show the type preservation for each reduction rule in Definition~\ref{def:small-step-reduction-commands} without the contexts.
    It follows from Lemma~\ref{lem:type-preservation-RH}.
\end{appendixproof}

\subsection{Examples}
\label{subsec:examples-operational-semantics}
We show reductions of commands defined in \S\ref{sec:examples-neural-network}.
We assume that the evaluation function $\Eval$ satisfies the following equations:
\begin{gather*}
    \Eval(\fnSwish_n, \fromvec{v}) = \fromvec{(v_i \cdot \sigmoid(v_i))_{i=1}^n},
    \\
    \Eval(\rd[\fnSwish_n], \tuple{\fromvec{v}, \fromvec{u}}) = \fromvec{(u_i (\sigmoid(v_i) + v_i \sigmoid(v_i)(1 - \sigmoid(v_i))))_
    {i=1}^n},
    \\
    \Eval(\fnscalarmul_n, \tuple{\fromvec{a}, \fromvec{v}}) = \fromvec{a v},
    \quad
    \Eval(\fnminus_n, \tuple{\fromvec{u}, \fromvec{v}}) = \fromvec{u - v},
    \quad
    \Eval(\fnmatmul, \tuple{\fromvec{m}, \fromvec{v}}) = \fromvec{m v},
    \quad
    \Eval((\blank)^\fntranspose, \fromvec{m}) = \fromvec{m^\fntranspose}.
\end{gather*}

\begin{example}[MLP, continued from Example~\ref{ex:mlp-backpropagation}]\label{ex:mlp-reduction}
    The signature $\signatopr$ contains $\opGet{\ell}$ and $\opPut{\ell}$, which read and write the memory location $\ell$.
    Hence, we extend the operational semantics as follows.
    A \emph{heap} is a function $\heap \colon \Locations \to \{ \fromvec{v} \mid v \in \Real^n, \ n \in \Nat \}$ satisfying $\heap(\ell) = \fromvec{v}$ for some $v \in \Real^{n_t \cdots n_1}$ for each $\loctyp{\ell}{(n_1, \dots, n_t)} \in \Locations$.
    We write $\heap[\ell \mapsto \fromvec{v}]$ for the standard heap update, % defined by
    % \[
    % \heap[\ell \mapsto \fromvec{v}](\ell') = \begin{cases}
    %     \heap(\ell') & (\ell \ne \ell')
    %     \\
    %     \fromvec{v} & (\ell = \ell')
    % \end{cases}
    % \]
    and $[\ell_i \mapsto \fromvec{v_i}]_{i \in I}$ for the heap defined by
    $[\ell_i \mapsto \fromvec{v_i}]_{i \in I}(\ell_i) = \fromvec{v_i}$ and
    $[\ell_i \mapsto \fromvec{v_i}]_{i \in I}(\ell) = \fromvec{0}$ if $\ell \ne \ell_i$ for every $i \in I$.

    We extend the reduction relation to pairs of a heap and a command:
    $(\heap, P) \redto (\heap, P')$ if $P \redto P'$.
    We also add the following rule to the reduction rules.
    \[
        (\heap, \ctxfc[\opGet{\ell}()])
        \redto
        (\heap, \ctxfc[\pret{\heap(\ell)}]),
        \quad
        (\heap, \ctxfc[\opPut{\ell}](\fromvec{v}))
        \redto
        (\heap[\ell \mapsto \fromvec{v}], \ctxfc[\pret{\fromvec{0}}]).
    \]

    For $m_0 \in \Real^{n_{\inp} n_{\hid}}$ and $m_1 \in \Real^{n_{\hid} n_{\out}}$, we have
    \[
    \begin{aligned}
        & ([\ell_0 \mapsto \fromvec{m_0}, \ell_1 \mapsto \fromvec{m_1}], \prevhandle(\fromvec{v})\tuple{x_{\inp}}. R_{\MLP} \; \pwith \; H_{\MLP}) \\
        & \redto^*
        \bigl(
            \bigl[ \ell_0 \mapsto \fromvec{m_0 - u'_1 (v^{\transpose}) }, \ell_1 \mapsto \fromvec{m_1 - \learningrate(m_1 v_1 - t) (v_1^{\transpose}) } \bigr],
            \pret{V}
        \bigr)
        \quad\text{where $V$ is a value.}
    \end{aligned}
    \]
    The entire reduction sequence is shown in Fig.~\ref{fig:mlp-reduction}.
    The resulting heap
    %$\left[ \ell_0 \mapsto \fromvec{m_0 - u'_1 (v^{\transpose}) }, \ell_1 \mapsto \fromvec{m_1 - \learningrate(m_1 v_1 - t) (v_1^{\transpose}) } \right]$
    after the reduction is the result of one step of gradient descent on the parameters $m_0$ and $m_1$ of the MLP with the learning rate $\learningrate$.
    See Example~\ref{ex:mlp-gradient-calc} for the detailed calculations of the coincidence between the resulting heap and the usual gradient descent.
\end{example}
\section{Categorical Preliminaries}
We review categorical notions, promonads and reverse differential restriction categories.
These notions will be used to define the denotational semantics of our language.
%Readers who are not interested in denotational semantics and are only concerned with type systems or operational semantics may skip this section.
\begin{notation}
    \label{notation:general}
    Let $\Real$ be the set of real numbers and $\Nat$ the set of natural numbers.
    For totally ordered sets $(X, \le_X)$ and $(Y, \le_Y)$, we write $(X \times Y, \le_{X,Y})$ for their Cartesian product as a set equipped with the lexicographic order $\le_{X,Y}$.
    For morphisms $f \colon X \to Y$ and $g \colon Y \to Z$ in a category $\catc$,
    we write their composition as $(f \semicolon g) \colon X \to Z$.
    For a Cartesian category $\catc$, we write
    \begin{itemize}
        \item $\pi_i \colon \prod_{i=1}^{n} X_i \to X_i$ for the $i$-th projection,
        \item $\diag_X : X \to X \times X$ for the diagonal morphism for an object $X$ in $\catc$,
        \item $\sym_{X,Y} \colon X \times Y \to Y \times X$ for the symmetry morphism for objects $X$ and $Y$ in $\catc$,
        \item $\terminal_{X} \colon X \to 1$ for the unique morphism from an object $X$ to the terminal object $1$ in $\catc$,
        \item $\assoc_{X,Y,Z} \colon X \times (Y \times Z) \to (X \times Y) \times Z$ for the associativity for objects $X, Y, Z$ in $\catc$.
    \end{itemize}
\end{notation}

\subsection{Promonads and Algebras}
\label{subsec:promonads-and-algebras}
We review promonads and their algebras.
Strong promonads have been studied as a categorical structure for arrows \cite{JacobsHeunenHasuo2009,Asada2010,Sanada2024}.
Homomorphisms between promonad algebras are used to interpret handlers \cite{Sanada2024}.
We assume familiarity with basic notions of category theory \cite{MacLane1971,Leinster2014}.

A \emph{profunctor} $\mathcal{P} \colon \catc\profto \catd$ is a functor $\mathcal{P} \colon \catc^{\opposite} \times \catd \to \Sets$.\footnote{%
    Note that in literature including \cite{Sanada2024}, such a profunctor is instead denoted by $\catd \profto \catc$.
    One reason for our choice is that it aligns with the hom-functor $\catc(\blank, \blank) \colon \catc^{\opposite} \times \catc \to \Sets$.
}
For fixed categories $\catc$ and $\catd$, profunctors $\catc \profto \catd$ and natural transformations form a category, denoted by $\Prof(\catc, \catd)$.
Moreover, small categories and profunctors form a bicategory $\Prof$, whose hom-category from $\catc$ to $\catd$ is $\Prof(\catc, \catd)$,
whose identity $1$-cell $\idprof_\catc$ on a category $\catc$ is the hom-functor $\catc(\blank, \blank) \colon \catc^{\opposite} \times \catc \to \Sets$.
The horizontal composition of profunctors is defined using \emph{coends} \cite{Loregian2021}.
However, we will not introduce the explicit definition of coends or of the composition of profunctors here; instead,
we describe $2$-cells in $\Prof$ whose source is a composite of profunctors.
Suppose that we have profunctors $\mathcal{F}\colon\catc\profto\catd$, $\mathcal{G}\colon\catd\profto\cate$, and $\mathcal{H}\colon\catc\profto\cate$.
Then, a $2$-cell $\zeta \colon \mathcal{F}\odot\mathcal{G}\natto\mathcal{H}\colon\catc\profto\cate$ in $\Prof$ can be described as
a family of functions $\{ \zeta_{X, Y, Z} \colon  \mathcal{F}(X, Y) \times \mathcal{G}(Y, Z) \to \mathcal{H}(X, Z) \}_{X,Y,Z}$ with the following properties:
\begin{itemize}
    \item %
        it is natural in $X$ and $Z$, and
    \item %
        it is \emph{extranatural} in $Y$; for every morphism $f\colon Y'\to Y$ in $\catd$, the following diagram commutes.
        \[
            \begin{tikzcd}[column sep=50pt, row sep=small]
                \mathcal{F}(X,Y')\times\mathcal{G}(Y,Z)
                \ar[r,"{\mathcal{F}(X,f)\times\idmor_{\mathcal{G}(Y,Z)}}"]
                \ar[d,"{\idmor_{\mathcal{F}(X,Y')}\times\mathcal{G}(f,Z)}"']
                    &
                    \mathcal{F}(X,Y)\times\mathcal{G}(Y,Z)
                    \ar[d,"\zeta_{X,Y,Z}"]
                \\
                \mathcal{F}(X,Y')\times\mathcal{G}(Y',Z)
                \ar[r,"\zeta_{X,Y',Z}"]
                    &
                    \mathcal{H}(X,Z)
            \end{tikzcd}
        \]
\end{itemize}
This characterization indeed defines the composite $\mathcal{F}\odot\mathcal{G}$ as the universal profunctor $\mathcal{H}$ equipped with such a family $\zeta$.
For simplicity, we often refer to $2$-cells in $\Prof$ as \emph{natural transformations}.

A \emph{promonad} $\mathcal{A}$ on a category $\catc$ is a monad on $\catc$ in the bicategory $\Prof$;
it consists of an endo-profunctor $\mathcal{A}\colon\catc\profto\catc$,
the \emph{unit} $\eta\colon\idprof_{\catc}\natto\mathcal{A}$, and the \emph{multiplication} $\mu\colon\mathcal{A}\odot\mathcal{A}\natto\mathcal{A}$,
satisfying the usual monad axioms. By the Yoneda lemma, the unit is equivalently described as a family
$\{ \eta_X (:=\eta(\idmor_X)) \in\mathcal{A}(X,X) \}_{X\in\obj(\catc)}$. Using the above description of $2$-cells in $\Prof$, the multiplication is equivalently described as a family $\{ \mu_{X,Y,Z} \colon \mathcal{A}(X,Y) \times \mathcal{A}(Y,Z) \to \mathcal{A}(X,Z) \}_{X,Y,Z}$.
For objects $X,Y,Z$ in $\catc$ and elements $a\in\mathcal{A}(X,Y)$ and $b\in\mathcal{A}(Y,Z)$, we write $\mu_{X,Y,Z}(a, b) \in \promonad(X, Z)$ as $a \musemicolon b$.
These data define a category whose objects are those of $\catc$, and morphisms from $X$ to $Y$ are elements of $\mathcal{A}(X,Y)$, with identities given by $\eta_X$ and composition given by $\musemicolon$.
This is equipped with an identity-on-objects functor from $\catc$, and indeed the converse is also true:
A promonad is equivalently defined as an identity-on-objects functor from $\catc$.

\begin{definition}[strong promonads]
%    For a small category $\catc$,
%    a \emph{promonad} on $\catc$ is a monad in the bicategory $\mathbf{Prof}$ on the object $\catc$.
%    That is, a promonad $\promonad$ on $\catc$ is a profunctor $\promonad \colon \catc \profto \catc$ equipped with
%    a 2-cell $\eta \colon \idprof_{\catc} \natto \promonad$ called the \emph{unit}, and
%    a 2-cell $\mu \colon \promonad \circ \promonad \natto \promonad$ called the \emph{multiplication},
%    satisfying the usual monad laws.
    A \emph{strong promonad} $\promonad$ on a (small) monoidal category $(\catc,1,\otimes)$ is a promonad $\promonad$ on $\catc$ equipped with
    a family $\{ \strength^{X, Y}_{Z} \colon \promonad(X, Y) \to \promonad(X \otimes Z, Y \otimes Z) \}_{X,Y,Z}$ natural in $X$ and $Y$, and extranatural in $Z$,
    called a \emph{strength}. It satisfies the usual strength laws; see \cite[Fig.~1]{Sanada2024} for the laws.
%    Note that such a family forms a natural transformation
%    $\promonad\times\idprof_\catc\natto\promonad(\blank\otimes\blank,\blank\otimes\blank)\colon\catc\times\catc\profto\catc\times\catc$ by the Yoneda lemma.
\end{definition}
As for promonads, strong promonads are equivalently described as an identity-on-objects functor from $\catc$ equipped with some structure, which is called
a \emph{(monoidal) Freyd category} (see \cite[Theorem~5.4]{HeunenJacobs06} and \cite[§9]{GarnerLopezFranco16}) or an \emph{effectful category} \cite{Roman2023}.
%In particular, for a strong promonad, the collage is a \emph{premonoidal category}, where the tensor product is functorial only in morphisms of $\catc$ but not in morphisms of the collage.

As explained in \S\ref{subsubsec:monads-vs-arrows:promonads},
we think of the family $\{\promonad(X,Y)\}_{X,Y \in \obj(\catc)}$ of a strong promonad as a family of sets of \emph{computations},
and those computations coming from morphisms in $\catc$ via the unit $\eta\colon\idprof_\catc\to\promonad$ as \emph{pure computations}.

We use a string-diagrammatic notation to represent elements in $\promonad(X,Y)$ for a strong promonad $\promonad$ on $\catc$.
The notation is used in \cite{AsadaHasuo2010,Sanada2024} and has a mathematical justification \cite{Roman2023}.
We omit drawing wires whose type is the tensor unit $1$.
Figure~\ref{fig:promonad-string-diagrams} summarizes the notation for a general strong promonad and, in the Cartesian or Cartesian restriction setting of \S\ref{subsec:RDRC}, for the pure structural morphisms induced by products.
\begin{figure}
\footnotesize
\[
\resizebox{\textwidth}{!}{$
\begin{array}{c@{\qquad}c@{\qquad}c@{\qquad}c}
\begin{tikzpicture}[baseline=(current bounding box.center), scale=0.8, transform shape]
    \pgfmathsetmacro{\xd}{0.3}
    \pgfmathsetmacro{\xw}{0.5}
    \pgfmathsetmacro{\xStart}{0}
    \pgfmathsetmacro{\xAL}{\xStart + \xd}
    \pgfmathsetmacro{\xAR}{\xAL + \xw}
    \pgfmathsetmacro{\xEnd}{\xAR + \xd}
    \draw (\xStart,  0.3) node[left] {$X_1$} -- (\xEnd,  0.3) node[right] {$Y_1$};
    \draw (\xStart, -0.3) node[left] {$X_n$} -- (\xEnd, -0.3) node[right] {$Y_m$};
    \node (midL) at (\xStart + \xd/2, 0.1) {$\vdots$};
    \node (midR) at (\xEnd   - \xd/2, 0.1) {$\vdots$};
    \draw[comp] (\xAL, -0.5) rectangle node {$a$} (\xAR, 0.5);
\end{tikzpicture}
\in \promonad(\bigotimes_{i=1}^n X_i, \bigotimes_{j=1}^m Y_j),
&
\eta(f) =
\begin{tikzpicture}[baseline=(current bounding box.center), scale=0.8, transform shape]
    \pgfmathsetmacro{\xd}{0.2}
    \pgfmathsetmacro{\xw}{0.5}
    \pgfmathsetmacro{\xStart}{0}
    \pgfmathsetmacro{\xAL}{\xStart + \xd}
    \pgfmathsetmacro{\xAR}{\xAL + \xw}
    \pgfmathsetmacro{\xEnd}{\xAR + \xd}
    \draw (\xStart, 0) node[left] {$X$} -- (\xEnd, 0) node[right] {$Y$};
    \draw[purecomp] (\xAL, -0.3) rectangle node {$f$} (\xAR, 0.3);
\end{tikzpicture},
&
\eta(\diag_X) =
\begin{tikzpicture}[baseline=(current bounding box.center), scale=0.8, transform shape]
    \pgfmathsetmacro{\xd}{0.2}
    \pgfmathsetmacro{\xw}{0.3}
    \pgfmathsetmacro{\xStart}{0}
    \pgfmathsetmacro{\xDiagL}{\xStart + \xd}
    \pgfmathsetmacro{\xDiagR}{\xDiagL + \xw}
    \pgfmathsetmacro{\xEnd}{\xDiagR + \xd}
    \draw (\xStart, 0) node[left] {$X$} -- (\xDiagL, 0);
    \draw (\xDiagL, 0) -- (\xDiagR,  0.4) -- (\xEnd,  0.4) node[right] {$X$};
    \draw (\xDiagL, 0) -- (\xDiagR, -0.4) -- (\xEnd, -0.4) node[right] {$X$};
\end{tikzpicture},
&
\eta(\sym_{X,Y}) =
\begin{tikzpicture}[baseline=(current bounding box.center), scale=0.8, transform shape]
    \pgfmathsetmacro{\xd}{0.2}
    \pgfmathsetmacro{\xw}{0.3}
    \pgfmathsetmacro{\xStart}{0}
    \pgfmathsetmacro{\xSymL}{\xStart + \xd}
    \pgfmathsetmacro{\xSymR}{\xSymL + \xw}
    \pgfmathsetmacro{\xEnd}{\xSymR + \xd}
    \draw (\xStart,  0.4) node[left] {$X$} -- (\xSymL,  0.4) -- (\xSymR, -0.4) -- (\xEnd, -0.4) node[right] {$X$};
    \draw (\xStart, -0.4) node[left] {$Y$} -- (\xSymL, -0.4) -- (\xSymR,  0.4) -- (\xEnd,  0.4) node[right] {$Y$};
\end{tikzpicture},
\\[1.2ex]
&
a \musemicolon b =
\begin{tikzpicture}[baseline=(current bounding box.center), scale=0.8, transform shape]
    \pgfmathsetmacro{\xd}{0.2}
    \pgfmathsetmacro{\xw}{0.5}
    \pgfmathsetmacro{\xStart}{0}
    \pgfmathsetmacro{\xAL}{\xStart + \xd}
    \pgfmathsetmacro{\xAR}{\xAL + \xw}
    \pgfmathsetmacro{\xBL}{\xAR + \xd}
    \pgfmathsetmacro{\xBR}{\xBL + \xw}
    \pgfmathsetmacro{\xEnd}{\xBR + \xd}
    \draw (\xStart, 0) node[left] {$X$} -- (\xEnd, 0) node[right] {$Z$};
    \draw[comp] (\xAL, -0.3) rectangle node {$a$} (\xAR, 0.3);
    \draw[comp] (\xBL, -0.3) rectangle node {$b$} (\xBR, 0.3);
\end{tikzpicture},
&
\strength^{X,Y}_{Z}(a) =
\begin{tikzpicture}[baseline=(current bounding box.center), scale=0.8, transform shape]
    \pgfmathsetmacro{\xd}{0.2}
    \pgfmathsetmacro{\xw}{0.5}
    \pgfmathsetmacro{\xStart}{0}
    \pgfmathsetmacro{\xAL}{\xStart + \xd}
    \pgfmathsetmacro{\xAR}{\xAL + \xw}
    \pgfmathsetmacro{\xEnd}{\xAR + \xd}
    \draw (\xStart,    0) node[left] {$X$} -- (\xEnd,    0) node[right] {$Y$};
    \draw (\xStart, -0.5) node[left] {$Z$} -- (\xEnd, -0.5) node[right] {$Z$};
    \draw[comp] (\xAL, -0.3) rectangle node {$a$} (\xAR, 0.3);
\end{tikzpicture}
%\eta(\pi_i) =
%\begin{tikzpicture}[baseline=(current bounding box.center), scale=0.8, transform shape]
%    \pgfmathsetmacro{\xd}{0.3}
%    \pgfmathsetmacro{\xw}{0.6}
%    \pgfmathsetmacro{\xStart}{0}
%    \pgfmathsetmacro{\xProj}{\xStart + \xd}
%    \pgfmathsetmacro{\xEnd}{\xProj + \xd}
%    \draw (\xStart,  0.6) node[left] {$X_1$} -- (\xProj,  0.6);
%    \draw (\xStart,  0.2)                    -- (\xProj,  0.2);
%    \draw (\xStart,    0) node[left] {$X_i$} -- (\xEnd,     0) node[right] {$X_i$};
%    \draw (\xStart, -0.2)                    -- (\xProj, -0.2);
%    \draw (\xStart, -0.6) node[left] {$X_n$} -- (\xProj, -0.6);
%    \node (dotsU) at (\xStart + \xd/2,  0.45) {\scalebox{0.5}{$\vdots$}};
%    \node (dotsD) at (\xStart + \xd/2, -0.35) {\scalebox{0.5}{$\vdots$}};
%    \draw[del] (\xProj,  0.6) circle (0.05);
%    \draw[del] (\xProj,  0.2) circle (0.05);
%    \draw[del] (\xProj, -0.2) circle (0.05);
%    \draw[del] (\xProj, -0.6) circle (0.05);
%\end{tikzpicture},
%&
%&
%\eta(\terminal_X) =
%\begin{tikzpicture}[baseline=(current bounding box.center), scale=0.8, transform shape]
%    \pgfmathsetmacro{\xd}{0.3}
%    \pgfmathsetmacro{\xw}{0.6}
%    \pgfmathsetmacro{\xStart}{0}
%    \pgfmathsetmacro{\xEnd}{\xStart + \xd}
%    \draw (\xStart, 0) node[left] {$X$} -- (\xEnd, 0);
%    \draw[del] (\xEnd, 0) circle (0.05);
%\end{tikzpicture}
\end{array}
$}
\]
\caption{String-diagrammatic notation for strong promonads.}% The top row shows a general element, pure morphisms, composition, and strength. The bottom row shows the pure structural morphisms used in the Cartesian or Cartesian restriction setting.}
\label{fig:promonad-string-diagrams}
\Description{String-diagrammatic notation for strong promonads and for the structural pure morphisms induced by products.}
\end{figure}

The rewriting rules for the string diagrams of promonads are derived from the monad laws and the strength laws.
Figure~\ref{fig:promonad-rewriting-contrast} contrasts the non-pure and pure cases: computations do not commute in general, but pure morphisms from $\catc$ do.
For details of the rewriting rules, see \cite{Roman2023}.
\begin{figure}
\footnotesize
\[
\resizebox{\textwidth}{!}{$
\begin{array}{c}
    \begin{tikzpicture}[baseline=(current bounding box.center), scale=0.8, transform shape]
        \pgfmathsetmacro{\xd}{0.2}
        \pgfmathsetmacro{\xw}{0.5}
        \pgfmathsetmacro{\xStart}{0}
        \pgfmathsetmacro{\xAL}{\xStart + \xd}
        \pgfmathsetmacro{\xAR}{\xAL + \xw}
        \pgfmathsetmacro{\xDiagL}{\xAR + \xd}
        \pgfmathsetmacro{\xDiagR}{\xDiagL + \xd}
        \pgfmathsetmacro{\xBL}{\xDiagR + \xd}
        \pgfmathsetmacro{\xBR}{\xBL + \xw}
        \pgfmathsetmacro{\xEnd}{\xBR + \xd}
        % a
        \draw (\xStart,    0) node[left] {$X$} -- (\xAR,    0);
        \draw (\xStart, -0.5) node[left] {$Z$} -- (\xAR, -0.5);
        \draw[comp] (\xAL, -0.3) rectangle node {$a$} (\xAR, 0.3);
        % diag
        \draw (\xAR,    0) -- (\xDiagL,    0) -- (\xDiagR,  -0.5);
        \draw (\xAR, -0.5) -- (\xDiagL, -0.5) -- (\xDiagR,    0);
        % b
        \draw (\xDiagR,    0) -- (\xEnd,    0) node[right] {$Y$};
        \draw (\xDiagR, -0.5) -- (\xEnd, -0.5) node[right] {$W$};
        \draw[comp] (\xBL, -0.3) rectangle node {$b$} (\xBR, 0.3);
    \end{tikzpicture}
    \ne
    \begin{tikzpicture}[baseline=(current bounding box.center), scale=0.8, transform shape]
        \pgfmathsetmacro{\xd}{0.2}
        \pgfmathsetmacro{\xw}{0.5}
        \pgfmathsetmacro{\xStart}{0}
        \pgfmathsetmacro{\xDiagAL}{\xStart + \xd}
        \pgfmathsetmacro{\xDiagAR}{\xDiagAL + \xd}
        \pgfmathsetmacro{\xAL}{\xDiagAR + \xd}
        \pgfmathsetmacro{\xAR}{\xAL + \xw}
        \pgfmathsetmacro{\xDiagL}{\xAR + \xd}
        \pgfmathsetmacro{\xDiagR}{\xDiagL + \xd}
        \pgfmathsetmacro{\xBL}{\xDiagR + \xd}
        \pgfmathsetmacro{\xBR}{\xBL + \xw}
        \pgfmathsetmacro{\xDiagBL}{\xBR + \xd}
        \pgfmathsetmacro{\xDiagBR}{\xDiagBL + \xd}
        \pgfmathsetmacro{\xEnd}{\xDiagBR + \xd}
        % diag A
        \draw (\xStart,    0) node[left] {$X$} -- (\xDiagAL,    0) -- (\xDiagAR,  -0.5);
        \draw (\xStart, -0.5) node[left] {$Z$} -- (\xDiagAL, -0.5) -- (\xDiagAR,     0);
        % b
        \draw (\xDiagAR,    0) -- (\xAR,    0);
        \draw (\xDiagAR, -0.5) -- (\xAR, -0.5);
        \draw[comp] (\xAL, -0.3) rectangle node {$b$} (\xAR, 0.3);
        % diag
        \draw (\xAR,    0) -- (\xDiagL,    0) -- (\xDiagR,  -0.5);
        \draw (\xAR, -0.5) -- (\xDiagL, -0.5) -- (\xDiagR,    0);
        % a
        \draw (\xDiagR,    0) -- (\xBR,    0);
        \draw (\xDiagR, -0.5) -- (\xBR, -0.5);
        \draw[comp] (\xBL, -0.3) rectangle node {$a$} (\xBR, 0.3);
        % diag B
        \draw (\xBR,    0) -- (\xDiagBL,    0) -- (\xDiagBR, -0.5) -- (\xEnd, -0.5) node[right] {$W$};
        \draw (\xBR, -0.5) -- (\xDiagBL, -0.5) -- (\xDiagBR,    0) -- (\xEnd,    0) node[right] {$Y$};
    \end{tikzpicture}
    \hspace{1em}
    \vline
    \hspace{1em}
    \begin{tikzpicture}[baseline=(current bounding box.center), scale=0.8, transform shape]
        \pgfmathsetmacro{\xd}{0.3}
        \pgfmathsetmacro{\xw}{0.6}
        \pgfmathsetmacro{\xStart}{0}
        \pgfmathsetmacro{\xAL}{\xStart + \xd}
        \pgfmathsetmacro{\xAR}{\xAL + \xw}
        \pgfmathsetmacro{\xDiagL}{\xAR + \xd}
        \pgfmathsetmacro{\xDiagR}{\xDiagL + \xd}
        \pgfmathsetmacro{\xBL}{\xDiagR + \xd}
        \pgfmathsetmacro{\xBR}{\xBL + \xw}
        \pgfmathsetmacro{\xEnd}{\xBR + \xd}
        % a
        \draw (\xStart,    0) node[left] {$X$} -- (\xAR,    0);
        \draw (\xStart, -0.5) node[left] {$Z$} -- (\xAR, -0.5);
        \draw[comp] (\xAL, -0.3) rectangle node {$a$} (\xAR, 0.3);
        % diag
        \draw (\xAR,    0) -- (\xDiagL,    0) -- (\xDiagR,  -0.5);
        \draw (\xAR, -0.5) -- (\xDiagL, -0.5) -- (\xDiagR,    0);
        % f
        \draw (\xDiagR,    0) -- (\xEnd,    0) node[right] {$Y$};
        \draw (\xDiagR, -0.5) -- (\xEnd, -0.5) node[right] {$W$};
        \draw[purecomp] (\xBL, -0.3) rectangle node {$f$} (\xBR, 0.3);
    \end{tikzpicture}
    =
    \begin{tikzpicture}[baseline=(current bounding box.center), scale=0.8, transform shape]
        \pgfmathsetmacro{\xd}{0.3}
        \pgfmathsetmacro{\xw}{0.6}
        \pgfmathsetmacro{\xStart}{0}
        \pgfmathsetmacro{\xDiagAL}{\xStart + \xd}
        \pgfmathsetmacro{\xDiagAR}{\xDiagAL + \xd}
        \pgfmathsetmacro{\xAL}{\xDiagAR + \xd}
        \pgfmathsetmacro{\xAR}{\xAL + \xw}
        \pgfmathsetmacro{\xDiagL}{\xAR + \xd}
        \pgfmathsetmacro{\xDiagR}{\xDiagL + \xd}
        \pgfmathsetmacro{\xBL}{\xDiagR + \xd}
        \pgfmathsetmacro{\xBR}{\xBL + \xw}
        \pgfmathsetmacro{\xDiagBL}{\xBR + \xd}
        \pgfmathsetmacro{\xDiagBR}{\xDiagBL + \xd}
        \pgfmathsetmacro{\xEnd}{\xDiagBR + \xd}
        % diag A
        \draw (\xStart,    0) node[left] {$X$} -- (\xDiagAL,    0) -- (\xDiagAR,  -0.5);
        \draw (\xStart, -0.5) node[left] {$Z$} -- (\xDiagAL, -0.5) -- (\xDiagAR,     0);
        % f
        \draw (\xDiagAR,    0) -- (\xAR,    0);
        \draw (\xDiagAR, -0.5) -- (\xAR, -0.5);
        \draw[purecomp] (\xAL, -0.3) rectangle node {$f$} (\xAR, 0.3);
        % diag
        \draw (\xAR,    0) -- (\xDiagL,    0) -- (\xDiagR,  -0.5);
        \draw (\xAR, -0.5) -- (\xDiagL, -0.5) -- (\xDiagR,    0);
        % a
        \draw (\xDiagR,    0) -- (\xBR,    0);
        \draw (\xDiagR, -0.5) -- (\xBR, -0.5);
        \draw[comp] (\xBL, -0.3) rectangle node {$a$} (\xBR, 0.3);
        % diag B
        \draw (\xBR,    0) -- (\xDiagBL,    0) -- (\xDiagBR, -0.5) -- (\xEnd, -0.5) node[right] {$W$};
        \draw (\xBR, -0.5) -- (\xDiagBL, -0.5) -- (\xDiagBR,    0) -- (\xEnd,    0) node[right] {$Y$};
    \end{tikzpicture}
\end{array}
$}
\]
\caption{Rewriting contrast for strong promonads.}% The top row shows a non-pure configuration that cannot be interchanged in general; the bottom row shows the corresponding interchange law when one component is pure.}
\label{fig:promonad-rewriting-contrast}
\Description{A contrast between non-pure and pure rewriting behavior in the string diagrams of strong promonads.}
\end{figure}

Following the paradigm of algebraic effects, we treat effects on computations as \emph{algebras} for the promonad.
Note that any presheaf $G\colon \catc^{\opposite}\to\Sets$ can be regarded as a profunctor $G\colon\catc\profto\mathbf{1}$ to the terminal category.
\begin{definition}[algebras and homomorphisms, \cite{Sanada2024}]
    \label{def:algebra}
    Let $\promonad$ be a promonad on a category $\catc$.
    An \emph{$\promonad$-algebra} is a pair $(G, \alpha)$ of a presheaf $G$ and a natural transformation
    $\alpha = \{\alpha_{X,Y} \colon \promonad(X, Y) \times G(Y) \to G(X) \}_{X,Y}\colon\promonad\odot G\natto G$ satisfying
    $\alpha(\eta_X, k) = k$ for $X \in \catc$ and $k \in G(X)$ and
    $\alpha(a \musemicolon b, k) = \alpha(a, \alpha(b, k))$ for $a \in \promonad(X, Y)$, $b \in \promonad(Y, Z)$ and $k \in G(Z)$.
    A \emph{homomorphism} between $\promonad$-algebras $(G, \alpha)$ and $(G', \alpha')$ is a natural transformation $h \colon G \natto G'$ satisfying
    $\alpha'(a, h_Y(k)) = h_X(\alpha(a, k))$
    for $a \in \promonad(X, Y)$ and $k \in G(Y)$.
\end{definition}
Intuitively, the set $G(X)$ is thought of as a set of ``continuations'' of the shape determined by $X$ in a sense.
%This aligns with the idea of \emph{contexts} in \cite{CurienHerbelin00}.
On the other hand, an algebra structure describes how a computation $R \in \promonad(Y,X)$ is executed as a continuation transformation $\alpha_{Y,X}(R,-)\colon G(X)\to G(Y)$,
where this execution realizes the effects occurring in $R$.

For each $Y \in \obj(\catc)$, $(\promonad(\blank, Y), \mu)$ is an $\promonad$-algebra.
This algebra is \emph{free} in the following sense. This will later be used to package the semantics of reverse handlers as homomorphisms from free algebras.
\begin{lemma}[{\cite[Theorem~3.2]{Sanada2024}}]\label{lem:free-algebra}
    Let $Y$ be an object in $\catc$,
    $(G, \alpha)$ an $\promonad$-algebra,
    and
    let $p\in G(Y)$.
    Then, there is a unique homomorphism $h \colon (\promonad(\blank, Y), \mu) \to (G, \alpha)$ such that the following diagram commutes for any object $X$ in $\catc$:
    \[
    \begin{tikzcd}[row sep=small]
        \catc(X, Y) \arrow[r, "\eta_{X,Y}"] \arrow[dr, "{\phi_X \defeq \lambda f . G(f)(p)}"']
        & \promonad(X, Y) \arrow[d, "h_X"]
        \\
        & G(X)
    \end{tikzcd}
    \]
\end{lemma}

\subsection{Reverse Differential Restriction Categories}
\label{subsec:RDRC}
Reverse differential restriction categories (RDRCs) \cite{Cruttwell+2021} are categorical models of reverse-mode automatic differentiation for partially defined smooth functions.
They generalize Cartesian reverse differential categories \cite{Cockett+2020} to the setting of restriction categories \cite{CocketLack2002},
and were developed to model Abadi and Plotkin's \emph{simple differentiable programming language} \cite{AbadiPlotkin2019}.
In many examples in this paper, we employ CRDCs as models; therefore, if you are not interested in the partiality of functions, you can safely ignore the term ``restriction'' in the following definitions and examples.

Here we briefly recall the definition of RDRCs and give some examples. For details, see \cite{Cruttwell+2021}.
\begin{definition}
    A \emph{Cartesian restriction category} \cite{CLIII} is a category $\catx$ equipped with a restriction structure
    $(f:A\to B)\mapsto (\restr{f}:A\to A)$ satisfying various axioms (see \cite{CocketLack2002}),
    and a monoidal structure $(\times,1)$ called \emph{restriction product} and \emph{restriction terminal object} satisfying various axioms.
\end{definition}
Intuitively, the endomorphism $\restr{f}$ represents the domain of a partial function $f$.
Given two parallel morphisms $f,g\colon A\to B$, we write $f\leq g$ if $\restr{f}\semicolon g = f$.
A morphism $f\colon A\to B$ in a Cartesian restriction category $\catx$ is called \emph{total} if $\restr{f} = \idmor_A$.
A restriction product and a restriction terminal object come equipped with canonical diagonal morphisms, projections, and a unique total morphism to $1$.
When working with restriction categories, we use Notation~\ref{notation:general} for these restriction versions of the usual categorical products and terminal object.
By a \emph{total point} of an object $X$, we mean a total morphism $1 \to X$.
\begin{definition}
    A \emph{Cartesian left additive restriction category} is a Cartesian restriction category $\catx$
    equipped with a structure of commutative monoid on each hom-set $\catx(X,Y)$ satisfying $x\semicolon(f+g)=(x\semicolon f)+(x\semicolon g)$
    and $x\semicolon0_{X,Y}\leq 0_{X',Y}$, and projections fully preserves the addition.
\end{definition}
\begin{definition}
    \label{def:RDRC}
    A \emph{reverse differential restriction category} (RDRC) is a Cartesian left additive restriction category $\catx$,
    which has an operation on maps:
    $\RD[-]:\catx(A,B)\to\catx(A\times B,A)
        % f:A\to B
        % \mapsto
        % \RD\lbrack f\rbrack:A\times B\to A
    $
    satisfying nine axioms \axiomRD{1}--\axiomRD{9} as in \cite[Definition~3.1]{Cruttwell+2021}.
    We write $\RD'[f] \colon Y \times X \to X$ for the morphism defined by $\RD'[f] = \sym_{Y,X} \semicolon \RD[f]$.
\end{definition}
Here, we only present \axiomRD{1} and \axiomRD{5}.
The axiom \axiomRD{1} states that the reverse derivative combinator preserves the monoid structure of hom-sets:
\[
    \RD[f + g] = \RD[f] + \RD[g],
    \quad
    \RD[0_{X,Y}] = 0_{X \times Y, X}.
\]
The axiom \axiomRD{5} formalizes the chain rule of reverse derivatives:
\[
    \RD[f \semicolon g]
%    = \left( X \times Z \xrightarrow{\langle \pi_1 , \langle \pi_1 \semicolon f, \pi_2 \rangle \rangle} X \times (Y \times Z)
%        \xrightarrow{\idmor_X \times \RD[g]} X \times Y
    = \bigl( X \times Z
        \xrightarrow{\Delta_X\times \idmor_Z}
        (X \times X) \times Z
        \cong
        X \times (X \times Z)
        \xrightarrow{\idmor_X\times((f\times\idmor_Z)\semicolon\RD[g])}
        X \times Y
        \xrightarrow{\RD[f]}
    X \bigr).
\]
The difference between RDRCs and CRDCs is the universality of products: in RDRCs, the product $X\times Y$ is not a categorical product.
%However, RDRCs are equipped with a diagonal morphism and projections that do not satisfy the usual equalities; for example, the projections do not form a natural transformation.
%If you use only diagonals, the product behaves just like an ordinary product; this is why we describe \axiomRD{5} using the diagonal morphism $\Delta_X$ instead of the projections $\pi_1$ and $\pi_2$, as in \cite{Cockett+2020,Cruttwell+2021}.
This point is clarified by the characterization of Cartesian restriction categories as \emph{$p$-categories}, observed in \cite{CLIII}.
\begin{example}
    \label{ex:PartSmooth-is-RDRC}
    Define the category $\PartSmooth$ so that it has the natural numbers $n \in \Nat$ as objects
    and a morphism $f\colon n\to m$ is a partial function $f\colon \Real^n\to\Real^m$ that is smooth on an open subset of $\Real^n$.
    This category $\PartSmooth$ is an RDRC.
    For the Cartesian (restriction) structure, the product $n \times m$ in $\PartSmooth$ is defined as $n +_{\Nat} m$, and the terminal object $1$ in $\PartSmooth$ is $0_{\Nat}$.

    The reverse differential combinator $\RD$ is defined by
    \[
        \RD[f](x_1, \dots, x_n, y_1, \dots, y_m)
        = \Bigl( \sum_{j=1}^m y_j \frac{\partial f_j}{\partial x_i}(x_1, \dots, x_n) \Bigr)_{i=1}^{n}
        = \left( \Jacobian_x f \right)^{\transpose} y
    \]
    where $(\Jacobian_x f)^{\transpose}$ is the transpose of the Jacobian matrix $\Jacobian_x f$ of $f$ at $x$.
    Alternatively, a morphism $f\colon n\to m$ is defined as a continuous partial function such that the above $\RD[f]$ is defined on
    $\mathrm{Dom}(f)\times\Real^m$ and is continuous there, where $\mathrm{Dom}(f)$ is the domain of $f$
    (see the definition of \emph{continuously differentiable} functions in \cite[\S{}4.1]{AbadiPlotkin2019}).
\end{example}
\begin{example}
    \label{ex:CRDC-is-RDRC}
    Any \emph{Cartesian reverse differential category} (CRDC) is an RDRC.
    In particular, the wide subcategory $\Smooth\subseteq\PartSmooth$ of total smooth functions is a CRDC, hence an RDRC.
\end{example}

\section{Categorical Semantics}
\label{sec:denotational-semantics}
%\inred{{TODO1: describe the difference from CGP. We need not interpret boolean terms, so that we need not assume the existence of joins in the RDRC.}}
%\inred{{TODO2: we need the semantics to be ``good'': in the sense that the section 5 of CGP.}}
%\inred{{TODO3: we need to note that all the equations are ``Kleene's'': that is, the left and right hand sides are defined at the same time.}}
%\inred{{TODO4: we need to change some examples remarking about non-smoothness, according to the semantics and the new version of the App E.}}
Let $\catx$ be a small RDRC.
The aim of this section is to give interpretation $\itp{P}$ for commands $\Gamma \asemicolon \Delta \vdash P : A$.
We write $\prod_{i=1}^nX_i$ for the (restriction) product of the list of objects $X_1, \dots, X_n$ in $\catx$.

\begin{definition}[interpretation of types]
    Given a map $\itp{\blank}_{\BType} \colon \BType \to \obj(\catx)$,
    we extend the map $\itp{\blank}_{\BType}$ to a map $\itp{\blank} \colon \Type \to \obj(\catx)$ by defining
    $\itp{\beta} = \itp{\beta}_{\BType}$ for $beta \in \BType$ and
    $\itp{\typprod_{i=1}^n A_i} = \prod_{i=1}^n \itp{A_i}$.
    For a typing environment $\Gamma = x_1 : A_1, \dots, x_n : A_n$, we define $\itp{\Gamma} = \prod_{i=1}^n \itp{A_i}$.
\end{definition}

\begin{definition}[interpretation of function symbols]
    Let $\signatfunc$ be a reverse derivative closed signature of function symbols.
    A family $\{ \itp{\func} \in \catx(\itp{A}, \itp{B}) \}_{\func : A \to B \in \signatfunc}$ is an \emph{interpretation} of $\signatfunc$ in $\catx$ if
    $\RD[\itp{\func}] = \itp{\rd[\func]}$ holds for each $\func : A \to B \in \signatfunc$.
\end{definition}

\begin{definition}[interpretation of operation symbols]
    Let $\signatopr$ be a signature of operation symbols and
    $\promonad$ be a promonad on a monoidal category $\catc$.
    An interpretation of $\signatopr$ in $\catc$ is a family
    $\{ \itp{\opr} \in \promonad(\itp{A}, \itp{B})\}_{\opr : A \opto B \in \signatopr}$.
\end{definition}

We define models, in which our language without reverse handlers is interpreted.
\begin{definition}[model]
    A \emph{model} of arrow calculus with signature $\signatfunc$ and $\signatopr$ consists of
    the following data:
    \begin{itemize}
        \item a small RDRC $\catx$,
        \item an interpretation of base types $\itp{\blank}_{\BType} \colon \BType \to \catx$,
        \item an interpretation $\{\itp{\func} \colon \itp{A} \to \itp{B} \}_{\func : A \to B \in \signatfunc}$ of $\signatfunc$,
        \item a strong promonad $\promonad \colon \catx \profto \catx$,
        \item an interpretation $\{ \itp{\opr} \in \promonad(\itp{A}, \itp{B}) \}_{\opr : A \opto B \in \signatopr}$ of $\signatopr$.
    \end{itemize}
\end{definition}

Note that, since our syntax does not involve predicates (unlike \cite{Cruttwell+2021}), we do not need to assume the existence of \emph{joins} in the reverse differential restriction category $\catx$.

\subsection{A Construction of the Free Model from Signatures}
To interpret the syntax of our language (with reverse handlers), we need to take in particular free models.
Let $\catc$ be a small monoidal category and $\signatopr$ be a signature of operation symbols,
and suppose that an interpretation of base types $\itp{\blank}_{\BType} \colon \BType \to \obj(\catc)$ is given.
Here we give a construction of a strong promonad $\promonad_{\signatopr}$ on $\catc$ from $\signatopr$,
the \emph{free model}, where our denotational semantics work.
\begin{definition}[\cite{Sanada2024}]
	\label{def:promonad-from-signature}
    We define a map $\ATerm_{\signatopr} \colon \obj(\catc^{\opposite}) \times \obj(\catc) \to \obj(\Sets)$ as the smallest set satisfying the rules in Fig.~\ref{fig:formation-rules-ATerm}.
    \begin{figure}
        \footnotesize
        \begin{gather*}
            \infer{
                \apure(f) \in \ATerm_{\signatopr}(X, Y)
            }{
                f \in \catc(X, Y)
            }
            \qquad
            \infer{
                \opr \in \ATerm_{\signatopr}(\itp{A}, \itp{B})
            }{
                \opr : A \opto B \in \signatopr
            }
            \quad
            \infer{
                a \acomp b \in \ATerm_{\signatopr}(X, Z)
            }{
                \begin{aligned}
                    a & \in \ATerm_{\signatopr}(X, Y) \\
                    b & \in \ATerm_{\signatopr}(Y, Z)
                \end{aligned}
            }
            \qquad
            \infer{
                \astr_Z(a) \in \ATerm_{\signatopr}(X \otimes Z, Y \otimes Z).
            }{
                a \in \ATerm_{\signatopr}(X, Y)
            }
        \end{gather*}
        \caption{Definition of $\ATerm_{\signatopr}(X, Y)$}\label{fig:formation-rules-ATerm}
        \Description{Definition of the class $\ATerm_{\signatopr}(X, Y)$ of terms.}
    \end{figure}
    We define a congruence relation $(\sim_{X,Y})$ on $\ATerm_{\signatopr}(X, Y)$ as the smallest congruence relation generated by the axioms in Fig.~\ref{fig:axioms-congruence-relation}.
    \begin{figure}
        \footnotesize
        \begin{gather*}
            (a \acomp b) \acomp c \sim a \acomp (b \acomp c),
            \quad
            \apure(f \semicolon g) \sim \apure(f) \acomp \apure(g),
            \quad
            \apure(\idmor_X) \acomp a \sim a,
            \quad
            a \acomp \apure(\idmor_Y) \sim a,
            \\
            \astr(a) \acomp \apure(\idmor \times f) \sim \apure(\idmor \times f) \acomp \astr(a),
            \qquad
            \astr(a) \acomp \apure(\pi_1) \sim \apure(\pi_1) \acomp a,
            \\
            \astr(a) \acomp \apure(\assoc) \sim \apure(\assoc) \acomp \astr(\astr(a)),
            \\
            \astr(\apure(f)) \sim \apure(f \times \idmor),
            \qquad
            \astr(a \acomp b) \sim \astr(a) \acomp \astr(b).
        \end{gather*}
        \caption{Axioms of the congruence relation $\sim_{X,Y}$ on $\ATerm_{\signatopr}(X,Y)$}\label{fig:axioms-congruence-relation}
        \Description{Axioms of the congruence relation $\sim_{X,Y}$ on $\ATerm_{\signatopr}(X, Y)$.}
    \end{figure}
    A map $\promonad_{\signatopr} \colon \obj(\catc^{\opposite}) \times \obj(\catc) \to \obj(\Sets)$ is defined by
    $\promonad_{\signatopr}(X, Y) = \ATerm_{\signatopr}(X, Y) / {\sim_{X,Y}}$.
\end{definition}
\begin{lemma}[{\cite[Proposition~5.6]{Sanada2024}}]\label{lem:promonad-from-signature}
    The map $\promonad_{\signatopr} \obj(\catc^{\opposite}) \times \obj(\catc) \to \obj(\Sets)$ defined in Definition~\ref{def:promonad-from-signature} can be extended to a strong promonad $\promonad_{\signatopr} \colon \catc \profto \catc$, that is, a functor $\promonad_{\signatopr} \colon \catc^{\opposite} \times \catc \to \Sets$ with a unit, a multiplication and a strength.
\end{lemma}
%\begin{appendixproof}[Proof of Lemma~\ref{lem:promonad-from-signature}]
%    The unit $\eta \colon \idprof_{\catc} \to \promonad_{\signatopr}$ is defined by
%    $\eta_{X,Y}(f) = [\apure(f)]$ for $f \in \catc(X, Y)$.
%    The multiplication $\mu \colon \promonad_{\signatopr} \circ \promonad_{\signatopr} \to \promonad_{\signatopr}$ is defined by
%    $\mu_{X,Y}([a], [b]) = [a \acomp b]$ for $[a] \in \promonad_{\signatopr}(X, Y)$ and $[b] \in \promonad_{\signatopr}(Y, Z)$.
%    The strength $\strength_{Z} \colon \promonad_{\signatopr}(X, Y) \to \promonad_{\signatopr}(X \otimes Z, Y \otimes Z)$ is defined by
%    $\strength_{Z}([a]) = [\astr_Z(a)]$ for $[a] \in \promonad_{\signatopr}(X, Y)$.
%    It is straightforward to check that the above definitions satisfy the axioms of strong promonads.
%\end{appendixproof}

The above definition indicates that the strong promonad $\promonad_{\signatopr}$ is \emph{freely generated} from the signature $\signatopr$, by the unit, multiplication, strength, and pure computations coming from $\catc$.

\subsection{Reverse Algebras}
\label{subsec:reverse-algebras}
Let $\promonad \colon \catx \profto \catx$ be a strong promonad on $\catx$.
For each $X \in \catx$, we introduce a new $\promonad$-algebra structure on $\promonad(X,X)$, the set of computations from $X$ to $X$.
Recall from the remark after \ref{def:algebra} that
one sees the set $\promonad(X,X)$ as the set of ``endomorphic continuations'' at $X$, and
such an algebra structure describes how a computation $R \in \promonad(Y,X)$ is executed as a continuation transformation $\promonad(X,X)\to\promonad(Y,Y)$.
However, even the execution of \emph{pure} computations, that is, computations coming from $\catx$, is non-trivial. This is exactly where we need the \emph{reverse derivative} combinator of the RDRC $\catx$:
a pure computation $Y\to X$ is executed as a continuation transformation $\promonad(X,X)\to\promonad(Y,Y)$ defined by its \emph{reverse-mode automatic differentiation}.
\begin{definition}
    \label{def:reverse-algebra}
    We define a function $\ralg{\promonad} \colon \obj(\catx^{\opposite}) \to \obj(\Sets)$ by
    $\ralg{\promonad}(X) = \promonad(X, X)$.
    For a morphism $f \colon X \to Y$ in $\catx$,
    we define a function $\ralg{\promonad}(f) \colon \ralg{\promonad}(Y) \to \ralg{\promonad}(X)$ by
    \[
        \ralg{\promonad}(f)(k)
        =
        \eta(\diag_X \semicolon (f \times \idmor_X))
        \musemicolon
        \strength_{X}(k)
        \musemicolon
        \eta(\RD'[f])
        =
        \begin{tikzpicture}[baseline=(current bounding box.center), scale=0.8, transform shape]
            \pgfmathsetmacro{\ytop}{0.2}
            \pgfmathsetmacro{\ybot}{-0.2}
            \pgfmathsetmacro{\yd}{0.2}
            \pgfmathsetmacro{\xd}{0.3}
            \pgfmathsetmacro{\xw}{0.6}
            \pgfmathsetmacro{\xwDiag}{0.2}
            \pgfmathsetmacro{\xStart}{0}
            \pgfmathsetmacro{\xDiagL}{\xStart + \xd}
            \pgfmathsetmacro{\xDiagR}{\xDiagL + \xwDiag}
            \pgfmathsetmacro{\xFL}{\xDiagR + \xd}
            \pgfmathsetmacro{\xFR}{\xFL + \xw}
            \pgfmathsetmacro{\xKL}{\xFR + \xd}
            \pgfmathsetmacro{\xKR}{\xKL + \xw}
            \pgfmathsetmacro{\xRfL}{\xKR + \xd}
            \pgfmathsetmacro{\xRfR}{\xRfL + 1.6 * \xw}
            \pgfmathsetmacro{\xEnd}{\xRfR + \xd}
            % Start
            \draw (\xStart, 0) node[left] {$X$} -- (\xDiagL, 0);
            % Diagonal
            \draw (\xDiagL, 0) -- (\xDiagR, \ytop);
            \draw (\xDiagL, 0) -- (\xDiagR, \ybot);
            % f times id
            \draw (\xDiagR, \ytop) -- (\xFR, \ytop);
            \draw (\xDiagR, \ybot) -- (\xFR, \ybot);
            \draw[purecomp] (\xFL, \ytop + \yd) rectangle node {$f$} (\xFR, \ytop - \yd);
            % st(k)
            \draw (\xFR, \ytop) -- (\xKR, \ytop);
            \draw (\xFR, \ybot) -- (\xKR, \ybot);
            \draw[comp] (\xKL, \ytop + \yd) rectangle node {$k$} (\xKR, \ytop - \yd);
            % RD'[f]
            \draw (\xKR, \ytop) -- (\xRfR, \ytop);
            \draw (\xKR, \ybot) -- (\xRfR, \ybot);
            \draw (\xRfR, 0) -- (\xEnd, 0) node[right] {$X$};
            \draw[purecomp] (\xRfL, \ytop + \yd) rectangle node {$\RD'[f]$} (\xRfR, \ybot - \yd);
        \end{tikzpicture}
    \]
    where $\RD'[f] = \sym_{Y,X} \semicolon \RD[f]$ (Definition~\ref{def:RDRC}).
\end{definition}
\begin{lemma}\label{lem:ralg-A-is-functor}
    The map $\ralg{\promonad}$ is a functor $\catx^{\opposite} \to \Sets$.
\end{lemma}
\begin{appendixproof}[Proof of Lemma~\ref{lem:ralg-A-is-functor}]
    For an identity morphism $\idmor_X \colon X \to X$,
    we have
    \[
    \begin{aligned}
        \ralg{\promonad}(\idmor_X)(k)
        & =
        \eta(\diag_X \semicolon (\idmor_X \times \idmor_X)) 
        \musemicolon
        \strength_{X}(k)
        \musemicolon
        \eta(\RD'[\idmor_X])
        \\
        & =
        \eta(\diag_X)
        \musemicolon
        \strength_{X}(k)
        \musemicolon
        \pi_1
        &
        \text{by \axiomRD{3}}
        \\
        & =
        \eta(\diag_X)
        \musemicolon
        \strength_{X}(k)
        \musemicolon
        \pi_1
        \\
        & =
        \eta(\diag_X)
        \musemicolon
        \pi_1
        \musemicolon
        k
        \\
        & =
        \eta(\diag_X)
        \musemicolon
        \pi_1
        \musemicolon
        k
        \\
        & =
        \idmor_X
        \musemicolon
        k
        \\
        & =
        k
    \end{aligned}
    \]
    for every $k \in \ralg{\promonad}(X)$.
    Hence, $\ralg{\promonad}(\idmor_X) = \idmor_{\ralg{\promonad}(X)}$.
    For morphisms $f \colon X \to Y$ and $g \colon Y \to Z$ in $\catx$,
    we have
    \[
    \begin{aligned}
        & \ralg{\promonad}(f \semicolon g)(k)
        \\
        & =
        \eta(\diag_X \semicolon ((f \semicolon g) \times \idmor_X))
        \musemicolon
        \strength_{X}(k)
        \musemicolon
        \eta(\RD'[f \semicolon g])
        \\
        & =
        \eta(\diag_X \semicolon (f \times \idmor_X) \semicolon (g \times \idmor_X))
        \musemicolon
        \strength_{X}(k)
        \musemicolon
        \eta((\idmor_Z \times (\diag_X \semicolon (f \times \idmor_X)) ) \semicolon (\RD'[g] \times \idmor_X) \semicolon \RD'[f])
        & \text{by \axiomRD{5}}
        \\
        & =
        \eta(\delta_X \semicolon (f \times \idmor_X))
        \musemicolon
        \strength_{X}(
            \eta(\diag_Y \semicolon (g \times \idmor_Y))
            \musemicolon
            \strength_{Y}(k)
        )
        \musemicolon
        \eta(\RD'[g])
        \musemicolon
        \eta(\RD'[f])
        \\
        & =
        \eta(\delta_X \semicolon (f \times \idmor_X))
        \musemicolon
        \strength_{X}(
            \ralg{\promonad}(g)(k)
        )
        \musemicolon
        \eta(\RD'[f])
        \\
        & =
        \ralg{\promonad}(f)(\ralg{\promonad}(g)(k))
    \end{aligned}
    \]
    for every $k \in \ralg{\promonad}(Z)$.
    Hence, $\ralg{\promonad}(f \semicolon g) = \ralg{\promonad}(g) \semicolon \ralg{\promonad}(f)$.
    Therefore, $\ralg{\promonad}$ is a functor $\catx^{\opposite} \to \Sets$.
\end{appendixproof}

We next describe how general computations in $\promonad(Y,X)$ are executed as continuation transformations $\promonad(X,X)\to\promonad(Y,Y)$ when $\promonad$ is freely generated in the above sense. This is the core of the semantics of reverse handlers.
Intuitively, we can determine how a computation is executed just by determining how the operations in $\signatopr$ execute, because every computation is generated from those operations.
\begin{definition}[{cf.~\cite[p28]{Sanada2024}}]
    \label{def:alpha-from-q-opr-Y}
    Let $\signatopr$ be a signature.
    Given a family $\{ q_{\opr} \colon \ralg{\promonad}(\itp{B} \times \blank) \to \ralg{\promonad}(\itp{A} \times \blank)\}_{\opr : A \opto B \in \signatopr}$ of natural transformations, we define a family
    $\alpha = \{ \alpha_{X,Y} \colon \ATerm_{\signatopr}(X,Y) \times \ralg{\promonad}(Y) \to \ralg{\promonad}(X) \}_{X,Y}$ of functions by
    \begin{align*}
        & \alpha(\apure(f), k)
        =
        \ralg{\promonad}(f)(k)
%        \eta(\diag_X \semicolon (f \times \idmor_X))
%        \musemicolon
%        \strength_X(k)
%        \musemicolon
%        \eta(\RD'[f]),
        & &
        \alpha(a \acomp b, k)
        =
        \alpha(a, \alpha(b, k)),
        \\
        & \alpha(\astr_Z(a \acomp b), k)
        =
        \alpha(\astr_Z(a) , \alpha(\astr_Z(b) , k)),
        & &
        \alpha(\astr_Z(f), k)
        =
        \alpha(f \times \idmor_Z , k),
        \\
        & \alpha(\astr_Z(\opr) , k)
        =
        q_{\opr,Z}(k),
        & &
        \alpha(\opr, k)
        =
        q_{\opr}(k),
        \\
        & \alpha(\astr_Z(\astr_{Z'}(f)) , k)
        =
        \alpha(\astr_{Z \otimes Z'}(f) , k).
    \end{align*}
\end{definition}
%\inred{check (maybe abrupt to some extent):}
It is straight forward to check the family $\alpha$ induces a natural transformation
of type $\promonad_{\signatopr} \odot \ralg{\promonad} \natto \ralg{\promonad}$,
and moreover, it becomes an algebra structure:
\begin{lemma}\label{lem:ralg-is-promonad-algebra}
    The natural transformation
    $\alpha \colon \promonad_{\signatopr}\odot\ralg{\promonad} \natto \ralg{\promonad}$
    determines an $\promonad_{\signatopr}$-algebra structure on $\ralg{\promonad}$.
\end{lemma}
\begin{appendixproof}[Proof of Lemma~\ref{lem:ralg-is-promonad-algebra}]
    We have
    \[
    \begin{aligned}
        \alpha(\eta_X, k)
        & =
        \eta(\diag_X \semicolon (\idmor_X \times \idmor_X))
        \musemicolon
        \strength_X(k)
        \musemicolon
        \eta(\RD'[\idmor_X])
        \\
        & =
        \eta(\diag_X)
        \musemicolon
        \strength_X(k)
        \musemicolon
        \eta(\pi_1)
        \\
        & =
        \eta(\diag_X ; \pi_1)
        \musemicolon
        \strength_X(k)
        \\
        & =
        k
    \end{aligned}
    \]
    and
    $\alpha((a \musemicolon b) , k)
    =
    \alpha(a , \alpha(b , k))$.
    This completes the proof.
\end{appendixproof}
In fact, we can replace the presheaf $\ralg{\promonad}$ with an arbitrary presheaf $G$. In particular, if we take $G=\promonad(-,D)$ for some $D\in\catx$, then we recover the construction of an algebra given in \cite[p28]{Sanada2024}.

This construction of $\alpha$ illustrates how continuation transformations for arbitrary computations are determined
from a smaller data $\{q_{\opr,Y}\colon \ralg{\promonad}(\itp{B} \times Y) \to \ralg{\promonad}(\itp{A} \times Y)\}_{\opr : A \opto B \in \signatopr,Y\in\catc}$,
which does not merely describe how basic operations are executed as continuation transformations; it is additionally parameterized by objects $Y \in \catc$.

Finally, to show that the programmer can obtain such an algebra structure via the above construction, we explain how our reverse handlers (see $\tyhandler$ in Figure~\ref{fig:typing-commands-handlers}) determine such parameterized continuation transformations $\{q_{\opr,Y}\}_{\opr : A \opto B \in \signatopr,\,Y\in\catc}$. This explains our understanding of backpropagation and how it is \emph{programmed} via reverse handlers.
\begin{definition}
    \label{def:q-opr-Y-from-reverse-handler}
    For families
    $\{ q^{\fwd}_{\opr} \in \promonad(\itp{A}, \itp{B} \times X_{\opr}) \}_{\opr : A \opto B \in \signatopr}$ and
    $\{ q^{\bwd}_{\opr} \in \promonad(\itp{B} \times X_{\opr}, \itp{A}) \}_{\opr : A \opto B \in \signatopr}$,
    we define a family
    $\{ q_{\opr} \colon \ralg{\promonad}(\itp{B} \times \blank) \to \ralg{\promonad}(\itp{A} \times \blank)\}_{\opr : A \opto B \in \signatopr}$
    of natural transformations by
    \[
    \begin{aligned}
        q_{\opr,Y}(k)
        & =
        \astr_{Y}(q^{\fwd}_{\opr})
        \musemicolon
        \eta(\idmor_{\itp{B}} \times \sym_{X_{\opr}, Y})
        \musemicolon
        \strength_{X_{\opr}}(k)
        \musemicolon
        \eta(\idmor_{\itp{B}} \times \sym_{Y, X_{\opr}})
        \musemicolon
        \astr_{Y}(q^{\bwd}_{\opr})
        \\
        & =
        \begin{tikzpicture}[baseline=(current bounding box.center), scale=0.8, transform shape]
            \pgfmathsetmacro{\ytop}{0.4}
            \pgfmathsetmacro{\ymid}{0}
            \pgfmathsetmacro{\ybot}{-0.4}
            \pgfmathsetmacro{\yd}{0.2}
            \pgfmathsetmacro{\xd}{0.3}
            \pgfmathsetmacro{\xw}{1}
            \pgfmathsetmacro{\xwDiag}{0.2}
            \pgfmathsetmacro{\xwSym}{0.2}
            \pgfmathsetmacro{\xStart}{0}
            \pgfmathsetmacro{\xQfL}{\xStart + \xd}
            \pgfmathsetmacro{\xQfR}{\xQfL + \xw}
            \pgfmathsetmacro{\xSymlL}{\xQfR + \xd}
            \pgfmathsetmacro{\xSymlR}{\xSymlL + \xwSym}
            \pgfmathsetmacro{\xStKL}{\xSymlR + \xd}
            \pgfmathsetmacro{\xStKR}{\xStKL + \xw}
            \pgfmathsetmacro{\xSymrL}{\xStKR + \xd}
            \pgfmathsetmacro{\xSymrR}{\xSymrL + \xwSym}
            \pgfmathsetmacro{\xQbL}{\xSymrR + \xd}
            \pgfmathsetmacro{\xQbR}{\xQbL + \xw}
            \pgfmathsetmacro{\xEnd}{\xQbR + \xd}
            % q^fwd
            \draw (\xStart, \ytop) node[left] {$\itp{A}$} -- (\xQfR, \ytop);
            \draw (\xStart, \ybot) node[left] {$Y$}       -- (\xQfR, \ybot);
            \draw[comp] (\xQfL, \ytop + \yd) rectangle node {$q^{\fwd}_{\opr}$} (\xQfR, \ymid - \yd);
            % id x sym
            \draw (\xQfR, \ytop) -- (\xSymlR, \ytop);
            \draw (\xQfR, \ymid) -- (\xSymlL, \ymid) -- (\xSymlR, \ybot);
            \draw (\xQfR, \ybot) -- (\xSymlL, \ybot) -- (\xSymlR, \ymid);
            % st(k)
            \draw (\xSymlR, \ytop) -- (\xStKR, \ytop);
            \draw (\xSymlR, \ymid) -- (\xStKR, \ymid);
            \draw (\xSymlR, \ybot) -- (\xStKR, \ybot);
            \draw[comp] (\xStKL, \ytop + \yd) rectangle node {$k$} (\xStKR, \ymid - \yd);
            % id x sym
            \draw (\xStKR, \ytop) -- (\xSymrR, \ytop);
            \draw (\xStKR, \ymid) -- (\xSymrL, \ymid) -- (\xSymrR, \ybot);
            \draw (\xStKR, \ybot) -- (\xSymrL, \ybot) -- (\xSymrR, \ymid);
            % q^bwd
            \draw (\xSymrR, \ytop) -- (\xEnd, \ytop) node[right] {$\itp{A}$};
            \draw (\xSymrR, \ymid) -- (\xQbR, \ymid);
            \draw (\xSymrR, \ybot) -- (\xEnd, \ybot) node[right] {$Y$};
            \draw[comp] (\xQbL, \ytop + \yd) rectangle node {$q^{\bwd}_{\opr}$} (\xQbR, \ymid - \yd);
        \end{tikzpicture}.
    \end{aligned}
    \]
\end{definition}

\subsection{Interpretation}
Let $(\catx, \promonad_{\signatopr}, \itp{\blank})$ be a model of arrow calculus with reverse handlers.
For a well-typed term $\Gamma, \Delta \vdash M : A$ and a well-typed command $\Gamma \asemicolon \Delta \vdash P : A$, we define their interpretation $\itp{M}^{\Gamma}_{\Delta} \colon \catx(1, \itp{\Gamma}) \to \catx(\itp{\Delta}, \itp{A})$ and $\itp{P} \colon \catx(1, \itp{\Gamma}) \to \promonad_{\signatopr}(\itp{\Delta}, \itp{A})$.

\begin{definition}
    For well-typed terms $\Gamma , \Delta \vdash M : A$, the interpretation $\itp{M}^{\Gamma}_{\Delta} \colon \catx(1, \itp{\Gamma}) \to \catx(\itp{\Delta}, \itp{A})$ is defined recursively as Figure~\ref{fig:interpretation-of-terms}.
    \begin{figure}[t]
        \footnotesize
        \begin{align*}
            \itp{\Gamma , \Delta \vdash x : A}^{\Gamma}_{\Delta}(c)
            & =
            \terminal_{\itp{\Delta}} \semicolon c \semicolon \pi_{x}
            \quad
            \text{if $x : A \in \Gamma$},
            \qquad
            \itp{\Gamma , \Delta \vdash y : B}^{\Gamma}_{\Delta}(c)
            =
            \pi_{y}
            \quad
            \text{if $y : B \in \Delta$}
            \\
            \itp{\Gamma, \Delta \vdash \const : \beta}^{\Gamma}_{\Delta}(c)
            & =
            \terminal_{\itp{\Delta}} \semicolon \itp{\const}
            \\
            \itp{\Gamma , \Delta \vdash \func(M) : B}^{\Gamma}_{\Delta}(c)
            & =
            \itp{M}^{\Gamma}_{\Delta}(c) \semicolon \itp{\func}
            \\
            \itp{\Gamma , \Delta \vdash \tlet \; x \tbind M \; \tin \; N : B}^{\Gamma}_{\Delta}(c)
            & =
            \langle\itp{M}^{\Gamma}_{\Delta}(c), \idmor_{\itp{\Delta}} \rangle \semicolon \itp{N}^{\Gamma}_{\Delta, x : A}(c)
            \\
            \itp{\Gamma , \Delta \vdash \tuple{M_i}_{i=1}^n : \typprod_{i=1}^n A_i}^{\Gamma}_{\Delta}(c)
            & =
            \langle \itp{M_i}^{\Gamma}_{\Delta}(c) \rangle_{i=1}^n
            \\
            \itp{\Gamma , \Delta \vdash \tproj_i(M) : A_i}^{\Gamma}_{\Delta}(c)
            & =
            \itp{M}^{\Gamma}_{\Delta}(c) \semicolon \pi_i
            \\
            \itp{\Gamma , \Delta \vdash M_1 + M_2 : A}^{\Gamma}_{\Delta}(c)
            & = \itp{M_1}^{\Gamma}_{\Delta}(c) + \itp{M_2}^{\Gamma}_{\Delta}(c)
            \\
            \itp{\Gamma, \Delta \vdash M.\trd(x.N)(L) : A}^{\Gamma}_{\Delta}(c)
            & =
            \langle \idmor_{\itp{\Delta}}, \itp{L}^{\Gamma}_{\Delta}(c), \itp{M}^{\Gamma}_{\Delta}(c)\rangle \semicolon \RD'[\itp{N}^{\Gamma}_{\Delta, x : A}(c)] \semicolon \pi_{x}
        \end{align*}
        \caption{Recursive definition of the interpretation of terms}
        \Description{Recursive equations defining the interpretation of variables, constants, function application, let, tuples, projections, addition, and reverse derivative terms.}
        \label{fig:interpretation-of-terms}
    \end{figure}
\end{definition}
The last is the interpretation of $\trd$ using reverse derivative operation, which is defined so that realize the intuition
explained after Definition~\ref{def:syntax:terms-commands}.
\begin{definition}
    For a well-typed command $\Gamma \asemicolon \Delta \vdash R : A$ and a well-typed handler $\revhandlerjudgment H : C$ where $H = \{ x \mapsto P \} \cup \{ (x \mapsto Q^{\fwd}_{\opr}, y,z \mapsto Q^{\bwd}_{\opr})_{\opr : A \opto B \in \signatopr} \}$,
    we define their interpretation
    $\itp{R} \colon \catx(1, \itp{\Gamma}) \to \promonad_{\signatopr}(\itp{\Delta}, \itp{A})$
    and
    $\itp{H} \colon \promonad_{\signatopr}(-, \itp{C}) \to \ralg{\promonad_{\signatopr}}$
    by the following recursion.
    \begin{align*}
        & \itp{\Gamma \asemicolon \Delta \vdash \pret{M} : A}(c)
        =
        \eta(\itp{M}^{\Gamma}_{\Delta}(c))
        =
        \begin{tikzpicture}[baseline=(current bounding box.center), scale=0.8, transform shape]
            \pgfmathsetmacro{\yd}{0.3}
            \pgfmathsetmacro{\xd}{0.2}
            \pgfmathsetmacro{\xStart}{0}
            \pgfmathsetmacro{\xML}{\xStart + \xd}
            \pgfmathsetmacro{\xMR}{\xML + 1.4}
            \pgfmathsetmacro{\xEnd}{\xMR + \xd}
            \draw (\xStart, 0) node[left] {$\itp{\Delta}$} -- (\xEnd, 0) node[right] {$\itp{A}$};
            \draw[purecomp] (\xML, \yd) rectangle node {$\itp{M}^{\Gamma}_{\Delta}(c)$} (\xMR, -\yd);
        \end{tikzpicture}
        \\
        & \itp{\Gamma \asemicolon \Delta \vdash \opr(M) : B}(c)
        =
        \eta(\itp{M}^{\Gamma}_{\Delta}(c)) \musemicolon \opr
        =
        \begin{tikzpicture}[baseline=(current bounding box.center), scale=0.8, transform shape]
            \pgfmathsetmacro{\yd}{0.3}
            \pgfmathsetmacro{\xd}{0.2}
            \pgfmathsetmacro{\xStart}{0}
            \pgfmathsetmacro{\xML}{\xStart + \xd}
            \pgfmathsetmacro{\xMR}{\xML + 1.4}
            \pgfmathsetmacro{\xOpL}{\xMR + \xd}
            \pgfmathsetmacro{\xOpR}{\xOpL + 0.6}
            \pgfmathsetmacro{\xEnd}{\xOpR + \xd}
            \draw (\xStart, 0) node[left] {$\itp{\Delta}$} -- (\xEnd, 0) node[right] {$\itp{B}$};
            \draw[purecomp] (\xML, \yd) rectangle node {$\itp{M}^{\Gamma}_{\Delta}(c)$} (\xMR, -\yd);
            \draw[comp] (\xOpL, \yd) rectangle node {$\opr$} (\xOpR, -\yd);
        \end{tikzpicture}
        \\
        & \itp{\Gamma \asemicolon \Delta \vdash \plet \; x \pbind P \; \pin \; Q}(c)
        =
        \eta(\diag_{\itp{\Delta}}) \musemicolon \strength_{\itp{\Delta}}(\itp{P}(c)) \musemicolon \itp{Q}(c)
        =
        \begin{tikzpicture}[baseline=(current bounding box.center), scale=0.8, transform shape]
            \pgfmathsetmacro{\yd}{0.3}
            \pgfmathsetmacro{\xd}{0.2}
            \pgfmathsetmacro{\xStart}{0}
            \pgfmathsetmacro{\xDiagL}{\xStart + \xd}
            \pgfmathsetmacro{\xDiagR}{\xDiagL + 0.2}
            \pgfmathsetmacro{\xPL}{\xDiagR + \xd}
            \pgfmathsetmacro{\xPR}{\xPL + 1.1}
            \pgfmathsetmacro{\xQL}{\xPR + \xd}
            \pgfmathsetmacro{\xQR}{\xQL + 1.1}
            \pgfmathsetmacro{\xEnd}{\xQR + \xd}
            % Diagonal
            \draw (\xStart, 0) node[left] {$\itp{\Delta}$} -- (\xDiagL, 0);
            \draw (\xDiagL, 0) -- (\xDiagR,  \yd);
            \draw (\xDiagL, 0) -- (\xDiagR, -\yd);
            % st(P)
            \draw (\xDiagR,  \yd) -- (\xPR,  \yd);
            \draw (\xDiagR, -\yd) -- (\xPR, -\yd);
            \draw[comp] (\xPL, \yd + 0.3) rectangle node {$\itp{P}(c)$} (\xPR, \yd - 0.3);
            % Q
            \draw (\xPR,  \yd) -- (\xQR,  \yd);
            \draw (\xPR, -\yd) -- (\xQR, -\yd);
            \draw (\xQR, 0) -- (\xEnd, 0) node[right] {$\itp{A}$};
            \draw[comp] (\xQL, \yd + 0.2) rectangle node {$\itp{Q}(c)$} (\xQR, -\yd - 0.2);
        \end{tikzpicture}
        \\
        & \itp{\Gamma \asemicolon \Delta \vdash \prevhandle(M)\tuple{x_i}.R'\; \pwith \; H}(c)
        =
        \eta(\itp{M}^{\Gamma}_{\Delta}(c)) \musemicolon \itp{H}_{\itp{\Delta'}}(\itp{R'}(c))
    \end{align*}
    
    Let
    $q^{\fwd}_{\opr} = \itp{Q^{\fwd}_{\opr}} \in \promonad_{\signatopr}(\itp{A}, \itp{B} \times \itp{D_{\opr}})$ and
    $q^{\bwd}_{\opr} = \itp{Q^{\bwd}_{\opr}} \in \promonad_{\signatopr}(\itp{B} \times \itp{D_{\opr}}, \itp{A})$.
    By the constructions~\ref{def:alpha-from-q-opr-Y}, \ref{lem:ralg-is-promonad-algebra}, and \ref{def:q-opr-Y-from-reverse-handler},
    we have an $\promonad_{\signatopr}$-algebra structure $\alpha$ on $\ralg{\promonad_{\signatopr}}$.
    The interpretation $\itp{H}$ is given as the homomorphism $(\promonad_{\signatopr}(-,\itp{C}), \mu) \to (\ralg{\promonad_{\signatopr}}, \alpha)$ induced by the freeness of the domain (Lemma~\ref{lem:free-algebra}), applied to
    $Y=\itp{C}$, $(G,\alpha)=(\ralg{\promonad_{\signatopr}},\alpha))$, and $p = \itp{P} \in \promonad_{\signatopr}(\itp{C}, \itp{C})$.
%    \[
%    \begin{tikzcd}[row sep=small]
%        {\catx(X, \itp{C})}
%        \ar[r, "\eta"]
%        \ar[rd, "\phi"']
%        &
%        {\promonad_{\signatopr}(X, \itp{C})}
%        \ar[d, dashed, "\itp{H}_X"]
%        \\
%        &
%        {\ralg{\promonad_{\signatopr}}(X)}
%    \end{tikzcd}
%    \quad
%    \begin{tikzcd}[row sep=small]
%        {\promonad_{\signatopr}(X,Y) \times \promonad_{\signatopr}(Y, \itp{C})}
%        \ar[d, dashed, "{\idmor \times \itp{H}_Y}"']
%        \ar[r, "\mu"]
%        &
%        {\promonad_{\signatopr}(X, \itp{C})}
%        \ar[d, dashed, "\itp{H}_X"]
%        \\
%        {\promonad_{\signatopr}(X,Y) \times \ralg{\promonad_{\signatopr}}(Y)}
%        \ar[r, "\alpha"]
%        &
%        {\ralg{\promonad_{\signatopr}}(X)}
%    \end{tikzcd}
%    \]
%    where $\phi = \lambda f . \eta(\diag_X \semicolon (f \times \idmor_X)) \musemicolon \strength_X(p) \musemicolon \eta(\RD'[f])$.
%    where $\phi = \lambda f . \ralg{\promonad_{\signatopr}}(f)(p)$.
%    \inred{Rephrasing the \ref{lem:free-algebra} via the Yoneda lemma, we might omit those diagrams.}
\end{definition}
In the terminology of \S\ref{subsec:promonads-and-algebras} and \S\ref{subsec:reverse-algebras}, where effects are understood as continuation transformations, the interpretation of the last command involving $\prevhandle$ can be described as follows:
\begin{itemize}
    \item %
        first, the computation $\itp{R'}(c)$ is executed as a continuation transformation
        $\promonad_{\signatopr}(\itp{C},\itp{C})\to\promonad_{\signatopr}(\itp{\Delta'},\itp{\Delta'})$,
        determined by $\itp{Q^{\fwd}_{\opr}}$ and $\itp{Q^{\bwd}_{\opr}}$ via Definitions~\ref{def:q-opr-Y-from-reverse-handler} and \ref{def:alpha-from-q-opr-Y};
    \item %
        next, this is followed by the left action
        $\promonad_{\signatopr}(\itp{\Delta'},\itp{\Delta'})\to\promonad_{\signatopr}(\itp{\Delta},\itp{\Delta'})$
        by $\itp{M}^{\Gamma}_{\Delta}(c)\colon\itp{\Delta}$ $\to\itp{\Delta'}$;
    \item %
        finally, the resulting continuation transformation is evaluated at
        $p = \itp{P}\in \promonad_{\signatopr}(\itp{C},\itp{C})$.
\end{itemize}

\subsection{Examples of Interpretation}
We show interpretation of some examples in \S\ref{sec:examples-neural-network}.
In Example~\ref{ex:mlp-denotation}, \ref{ex:resnet-denotation} and \ref{ex:autoencoder-denotation} below, we use $\Smooth$ as an RDRC $\catx$ (see Example~\ref{ex:CRDC-is-RDRC}).
In Example~\ref{ex:U-net-denotation}, we use $\PartSmooth$ (see Example~\ref{ex:PartSmooth-is-RDRC}), mainly because the map $\max(x,y)$ is not differentiable at the line $y = x$.

\begin{example}[MLP, continued from Example~\ref{ex:mlp-backpropagation}]
    \label{ex:mlp-denotation}
    Our framework can be understood as providing a rigorous language and a clear categorical semantics for the architectural diagrams that appear in the neural network literature.
    The interpretation $\itp{R_{\MLP}} \in \promonad_{\signatopr}(\Real^{n_{\inp}}, \Real^{n_{\out}})$ is as follows:
    \[
    \begin{tikzpicture}[scale=0.8, transform shape]
        \pgfmathsetmacro{\yw}{0.3}
        \pgfmathsetmacro{\xd}{0.3}
        \pgfmathsetmacro{\xw}{2.9}
        \pgfmathsetmacro{\xwDiag}{0.2}
        \pgfmathsetmacro{\xwSym}{0.2}
        \pgfmathsetmacro{\xStart}{0}
        \pgfmathsetmacro{\xLinearAL}{\xStart + \xd}
        \pgfmathsetmacro{\xLinearAR}{\xLinearAL + \xw}
        \pgfmathsetmacro{\xSwishL}{\xLinearAR + \xd}
        \pgfmathsetmacro{\xSwishR}{\xSwishL + 0.8 * \xw}
        \pgfmathsetmacro{\xLinearBL}{\xSwishR + \xd}
        \pgfmathsetmacro{\xLinearBR}{\xLinearBL + \xw}
        \pgfmathsetmacro{\xEnd}{\xLinearBR + \xd}
        % Linear A
        \draw (\xStart, 0) node[left] {$\Real^{n_{\inp}}$} -- (\xLinearAR, 0);
        \draw[comp] (\xLinearAL, \yw) rectangle node {$\opFC{\loctyp{\ell_0}{ (n_{\inp}, n_{\hid})}}$} (\xLinearAR, -\yw);
        % Swish
        \draw (\xLinearAR, 0) -- (\xSwishR, 0);
        \draw[purecomp] (\xSwishL, \yw) rectangle node {$\itp{\fnSwish_{n_{\hid}}}$} (\xSwishR, -\yw);
        % Linear B
        \draw (\xSwishR, 0) -- (\xEnd, 0) node[right] {$\Real^{n_{\out}}$};
        \draw[comp] (\xLinearBL, \yw) rectangle node {$\opFC{\loctyp{\ell_1}{( n_{\hid}, n_{\out})}}$} (\xLinearBR, -\yw);
    \end{tikzpicture}.
    \]
    The above string diagram coincides with the typical architectural diagram of MLPs.
    The interpretation $\itp{\prevhandle(\fromvec{v}) \tuple{x_{\inp}}. R_{\MLP} \; \pwith \; H_{\MLP}} \in \promonad_{\signatopr}(1, \Real^{n_{\inp}})$ is as follows:
    \[
    \begin{tikzpicture}[scale=0.8, transform shape]
        \pgfmathsetmacro{\yw}{0.5}
        \pgfmathsetmacro{\yd}{0.4}
        \pgfmathsetmacro{\xd}{0.2}
        \pgfmathsetmacro{\xw}{1.6}
        \pgfmathsetmacro{\xwsmall}{0.4}
        \pgfmathsetmacro{\xwDiag}{0.15}
        \pgfmathsetmacro{\xwSym}{0.2}
        \pgfmathsetmacro{\xStart}{0}
        \pgfmathsetmacro{\xVL}{\xStart}
        \pgfmathsetmacro{\xVR}{\xVL + \xwsmall}
        \pgfmathsetmacro{\xQfAL}{\xVR + \xd}
        \pgfmathsetmacro{\xQfAR}{\xQfAL + \xw}
        \pgfmathsetmacro{\xDiagL}{\xQfAR + \xd}
        \pgfmathsetmacro{\xDiagR}{\xDiagL + \xwDiag}
        \pgfmathsetmacro{\xSwishL}{\xDiagR + \xd}
        \pgfmathsetmacro{\xSwishR}{\xSwishL + 0.8 * \xw}
        \pgfmathsetmacro{\xQfBL}{\xSwishR + \xd}
        \pgfmathsetmacro{\xQfBR}{\xQfBL + \xw}
        \pgfmathsetmacro{\xPL}{\xQfBR + \xd}
        \pgfmathsetmacro{\xPR}{\xPL + 0.5 * \xw}
        \pgfmathsetmacro{\xQbBL}{\xPR + \xd}
        \pgfmathsetmacro{\xQbBR}{\xQbBL + \xw}
        \pgfmathsetmacro{\xRSwishL}{\xQbBR + \xd}
        \pgfmathsetmacro{\xRSwishR}{\xRSwishL + 1.2 * \xw}
        \pgfmathsetmacro{\xQbAL}{\xRSwishR + \xd}
        \pgfmathsetmacro{\xQbAR}{\xQbAL + \xw}
        \pgfmathsetmacro{\xEnd}{\xQbAR + \xd}
        % v
        \draw[purecomp] (\xVL, \yw) rectangle node {$v$} (\xVR, -\yw);
        % Q^fwd A
        \draw (\xVR, 0) -- (\xQfAR, 0);
        \draw[comp] (\xQfAL, \yw) rectangle node {$\itp{Q^{\fwd}_{\opFC{\ell_0}}}$} (\xQfAR, -\yw);
        % diag x id
        \draw (\xQfAR,   \yd) -- (\xDiagL,  \yd) -- (\xDiagR, 2 * \yd);
        \draw (\xDiagL,  \yd) -- (\xDiagR,  0);
        \draw (\xQfAR,  -\yd) -- (\xDiagR, -\yd);
        % Swish x id
        \draw (\xDiagR, 2 * \yd) -- (\xSwishR, 2 * \yd);
        \draw (\xDiagR, 0) -- (\xSwishR, 0);
        \draw (\xDiagR, -\yd) -- (\xSwishR, -\yd);
        \draw[purecomp] (\xSwishL, 2 * \yd + 0.6 * \yw) rectangle node {$\itp{\fnSwish}$} (\xSwishR, 2 * \yd - 0.6 * \yw);
        % Q^fwd B
        \draw (\xSwishR, 2 * \yd) -- (\xQfBR, 2 * \yd);
        \draw (\xSwishR,       0) -- (\xQfBR,       0);
        \draw (\xSwishR,    -\yd) -- (\xQfBR,    -\yd);
        \draw[comp] (\xQfBL, 2 * \yd + \yw) rectangle node {$\itp{Q^{\fwd}_{\opFC{\ell_1}}}$} (\xQfBR, 2 * \yd - \yw);
        % P
        \draw (\xQfBR, 2 * \yd + 0.6 * \yd) -- (\xPR, 2 * \yd + 0.6 * \yd);
        \draw (\xQfBR, 2 * \yd - 0.6 * \yd) -- (\xPR, 2 * \yd - 0.6 * \yd);
        \draw (\xQfBR,       0) -- (\xPR,       0);
        \draw (\xQfBR,    -\yd) -- (\xPR,    -\yd);
        \draw[comp] (\xPL, 2 * \yd + 0.6 * \yd + 0.6 * \yw) rectangle node {$\itp{P}$} (\xPR, 2 * \yd + 0.6 * \yd - 0.6 * \yw);
        % Q^bwd B
        \draw (\xPR, 2 * \yd + 0.6 * \yd) -- (\xQbBR, 2 * \yd + 0.6 * \yd);
        \draw (\xPR, 2 * \yd - 0.6 * \yd) -- (\xQbBR, 2 * \yd - 0.6 * \yd);
        \draw (\xPR,       0) -- (\xQbBR,       0);
        \draw (\xPR,    -\yd) -- (\xQbBR,    -\yd);
        \draw[comp] (\xQbBL, 2 * \yd + \yw) rectangle node {$\itp{Q^{\bwd}_{\opFC{\ell_1}}}$} (\xQbBR, 2 * \yd - \yw);
        % R[Swish]
        \draw (\xQbBR, 2 * \yd) -- (\xRSwishR, 2 * \yd);
        \draw (\xQbBR, 0) -- (\xRSwishR, 0);
        \draw (\xQbBR, -\yd) -- (\xRSwishR, -\yd);
        \draw[purecomp] (\xRSwishL, 2 * \yd + 0.5 * \yw) rectangle node {$\RD'[\itp{\fnSwish}]$} (\xRSwishR, 0 - 0.5 * \yw);
        % Q^bwd A
        \draw (\xRSwishR,  \yd) -- (\xQbAR,  \yd);
        \draw (\xRSwishR, -\yd) -- (\xQbAR, -\yd);
        \draw (\xQbAR, 0) -- (\xEnd, 0) node[right] {$\Real^{n_{\inp}}$};
        \draw[comp] (\xQbAL, \yw) rectangle node {$\itp{Q^{\bwd}_{\opFC{\ell_0}}}$} (\xQbAR, -\yw);
    \end{tikzpicture}
    \]
    where $\itp{P} \in \promonad_{\signatopr}(\Real^{n_{\out}})$,
    $\itp{Q^{\fwd}_{\opFC{\loctyp{\ell}{(n,m)}}}} \in \promonad_{\signatopr}(\Real^{n}, \Real^{m} \times \Real^{nm} \times \Real^{n})$ and
    $\itp{Q^{\bwd}_{\opFC{\loctyp{\ell}{(n,m)}}}} \in \promonad_{\signatopr}(\Real^{m} \times \Real^{nm} \times \Real^{n}, \Real^{n})$
    are given as follows:
    \begin{align*}
        \itp{P}
        & =
        \begin{tikzpicture}[baseline=(current bounding box.center), scale=0.75, transform shape]
            \pgfmathsetmacro{\yw}{0.4}
            \pgfmathsetmacro{\xd}{0.2}
            \pgfmathsetmacro{\xw}{2}
            \pgfmathsetmacro{\xStart}{0}
            \pgfmathsetmacro{\xPL}{\xStart + \xd}
            \pgfmathsetmacro{\xPR}{\xPL + \xw}
            \pgfmathsetmacro{\xEnd}{\xPR + \xd}
            % P
            \draw (\xStart, 0) node[left] {$\Real^{n_{\out}}$} -- (\xEnd, 0) node[right] {$\Real^{n_{\out}}$};
            \draw[purecomp] (\xPL, \yw) rectangle node {$\lambda x . \learningrate (x - t)$} (\xPR, -\yw);
        \end{tikzpicture}
        \\
        \itp{Q^{\fwd}_{\opFC{\ell}}}
        & =
        \begin{tikzpicture}[baseline=(current bounding box.center), scale=0.7, transform shape]
            \pgfmathsetmacro{\yd}{0.45}
            \pgfmathsetmacro{\yw}{0.3}
            \pgfmathsetmacro{\xd}{0.2}
            \pgfmathsetmacro{\xw}{1.5}
            \pgfmathsetmacro{\xStart}{0}
            \pgfmathsetmacro{\xGetL}{\xStart + \xd}
            \pgfmathsetmacro{\xGetR}{\xGetL + 0.8 * \xw}
            \pgfmathsetmacro{\xDiagL}{\xGetR + \xd}
            \pgfmathsetmacro{\xDiagR}{\xDiagL + 0.2}
            \pgfmathsetmacro{\xMatL}{\xDiagR + \xd}
            \pgfmathsetmacro{\xMatR}{\xMatL + 1.2 * \xw}
            \pgfmathsetmacro{\xEnd}{\xMatR + \xd}
            % Get
            \draw (\xStart, 0) node[left] {$\Real^{n}$} -- (\xGetR, 0);
            \draw[comp] (\xGetL, \yd + \yw) rectangle node {$\opGet{\ell}$} (\xGetR, \yd - \yw);
            % diag
            \draw (\xGetR, \yd) -- (\xDiagL, \yd);
            \draw (\xGetR,   0) -- (\xDiagL,   0);
            \draw (\xDiagL, \yd) -- (\xDiagR,  2 * \yd);
            \draw (\xDiagL, \yd) -- (\xDiagR,  0 * \yd);
            \draw (\xDiagL,   0) -- (\xDiagR,  1 * \yd);
            \draw (\xDiagL,   0) -- (\xDiagR, -1 * \yd);
            % Mat
            \draw (\xDiagR,  2 * \yd) -- (\xMatR, 2 * \yd);
            \draw (\xDiagR,  1 * \yd) -- (\xMatR, 1 * \yd);
            \draw (\xMatR, 1.5 * \yd) -- (\xEnd, 1.5 * \yd) node[right] {$\Real^{m}$};
            \draw (\xDiagR,  0 * \yd) -- (\xEnd,   0 * \yd) node[right] {$\Real^{nm}$};
            \draw (\xDiagR, -1 * \yd) -- (\xEnd,  -1 * \yd) node[right] {$\Real^{n}$};
            \draw[purecomp] (\xMatL, 2 * \yd + 0.8 * \yw) rectangle node {$\lambda (w, x) . wx$} (\xMatR, 1 * \yd - 0.8 * \yw);
        \end{tikzpicture}
        &
        \itp{Q^{\bwd}_{\opFC{\ell}}}
        & =
        \begin{tikzpicture}[baseline=(current bounding box.center), scale=0.7, transform shape]
            \pgfmathsetmacro{\yd}{0.4}
            \pgfmathsetmacro{\yw}{0.2}
            \pgfmathsetmacro{\xd}{0.2}
            \pgfmathsetmacro{\xw}{3}
            \pgfmathsetmacro{\xStart}{0}
            \pgfmathsetmacro{\xDiagL}{\xStart + \xd}
            \pgfmathsetmacro{\xDiagR}{\xDiagL + 0.2}
            \pgfmathsetmacro{\xMatL}{\xDiagR + \xd}
            \pgfmathsetmacro{\xMatR}{\xMatL + 0.9 * \xw}
            \pgfmathsetmacro{\xPutL}{\xMatR + \xd}
            \pgfmathsetmacro{\xPutR}{\xPutL + 0.45 * \xw}
            \pgfmathsetmacro{\xMYL}{\xPutR + \xd}
            \pgfmathsetmacro{\xMYR}{\xMYL + 0.7 * \xw}
            \pgfmathsetmacro{\xEnd}{\xMYR + \xd}
            % diag
            \draw (\xStart, 3 * \yd) node[left] {$\Real^{m}$} -- (\xDiagL, 3 * \yd);
            \draw (\xStart, 2 * \yd) node[left] {$\Real^{nm}$} -- (\xDiagL, 2 * \yd);
            \draw (\xStart, 1 * \yd) node[left] {$\Real^{n}$} -- (\xDiagL, 1 * \yd);
            \draw (\xDiagL, 3 * \yd) -- (\xDiagR, 4 * \yd);
            \draw (\xDiagL, 3 * \yd) -- (\xDiagR, 1 * \yd);
            \draw (\xDiagL, 2 * \yd) -- (\xDiagR, 3 * \yd);
            \draw (\xDiagL, 2 * \yd) -- (\xDiagR, 0 * \yd);
            \draw (\xDiagL, 1 * \yd) -- (\xDiagR, 2 * \yd);
            % Mat
            \draw (\xDiagR, 4 * \yd) -- (\xMatR, 4 * \yd);
            \draw (\xDiagR, 3 * \yd) -- (\xMatR, 3 * \yd);
            \draw (\xDiagR, 2 * \yd) -- (\xMatR, 2 * \yd);
            \draw (\xDiagR, 1 * \yd) -- (\xMatR, 1 * \yd);
            \draw (\xDiagR, 0 * \yd) -- (\xMatR, 0 * \yd);
            \draw[purecomp] (\xMatL, 4 * \yd + \yw) rectangle node {$\lambda(y, w, x) . w - yx^{\transpose}$} (\xMatR, 2 * \yd - \yw);
            % Put
            \draw (\xMatR, 3 * \yd) -- (\xPutR, 3 * \yd);
            \draw (\xMatR, 1 * \yd) -- (\xPutR, 1 * \yd);
            \draw (\xMatR, 0 * \yd) -- (\xPutR, 0 * \yd);
            \draw[comp] (\xPutL, 3 * \yd + 1.2 * \yw) rectangle node {$\opPut{\ell}$} (\xPutR, 3 * \yd - 1.2 * \yw);
            % m^T * y
            \draw (\xPutR, 1 * \yd) -- (\xMYL, 1 * \yd);
            \draw (\xPutR, 0 * \yd) -- (\xMYL, 0 * \yd);
            \draw (\xMYL, 0.5 * \yd) -- (\xEnd, 0.5 * \yd) node[right] {$\Real^{n}$};
            \draw[purecomp] (\xMYL, 1 * \yd + 0.6 * \yw) rectangle node {$\lambda(y, w) . w^{\transpose} y$} (\xMYR, 0 * \yd - 0.6 * \yw);
        \end{tikzpicture}.
    \end{align*}
\end{example}
\begin{example}[residual neural network, continued from Example~\ref{ex:residual-neural-network}]\label{ex:resnet-denotation}
    The interpretation $\itp{R_{\ResNet}} \in \promonad_{\signatopr}(\Real^{n_{\inp}}, \Real^{n_{\out}})$ is as follows:
    \[
    \begin{tikzpicture}[scale=0.8, transform shape]
        \pgfmathsetmacro{\yw}{0.3}
        \pgfmathsetmacro{\yd}{0.7}
        \pgfmathsetmacro{\xd}{0.3}
        \pgfmathsetmacro{\xw}{2.3}
        \pgfmathsetmacro{\xwDiag}{0.2}
        \pgfmathsetmacro{\xwSym}{0.2}
        \pgfmathsetmacro{\xStart}{0}
        \pgfmathsetmacro{\xDiagL}{\xStart + \xd}
        \pgfmathsetmacro{\xDiagR}{\xDiagL + \xwDiag}
        \pgfmathsetmacro{\xLinearAL}{\xDiagR + \xd}
        \pgfmathsetmacro{\xLinearAR}{\xLinearAL + \xw}
        \pgfmathsetmacro{\xSwishL}{\xLinearAR + \xd}
        \pgfmathsetmacro{\xSwishR}{\xSwishL + 0.8 * \xw}
        \pgfmathsetmacro{\xLinearBL}{\xSwishR + \xd}
        \pgfmathsetmacro{\xLinearBR}{\xLinearBL + \xw}
        \pgfmathsetmacro{\xPlusL}{\xLinearBR + \xd}
        \pgfmathsetmacro{\xPlusR}{\xPlusL + \xwDiag}
        \pgfmathsetmacro{\xEnd}{\xPlusR + \xd}
        % diag
        \draw (\xStart, -0.5 * \yd) node[left] {$\Real^{n_{\inp}}$} -- (\xDiagL, -0.5 * \yd);
        \draw (\xDiagL, -0.5 * \yd) -- (\xDiagR,  0 * \yd);
        \draw (\xDiagL, -0.5 * \yd) -- (\xDiagR, -1 * \yd);
        % Linear A x id
        \draw (\xDiagR, 0) -- (\xLinearAR, 0);
        \draw (\xDiagR, -1 * \yd) -- (\xLinearAR, -1 * \yd);
        \draw[comp] (\xLinearAL, \yw) rectangle node {$\opFC{\ell_0}$} (\xLinearAR, -\yw);
        % Swish x id
        \draw (\xLinearAR, 0) -- (\xSwishR, 0);
        \draw (\xLinearAR, -1 * \yd) -- (\xSwishR, -1 * \yd);
        \draw[purecomp] (\xSwishL, \yw) rectangle node {$\itp{\fnSwish_{n_{\hid}}}$} (\xSwishR, -\yw);
        % Linear B x id
        \draw (\xSwishR, 0) -- (\xLinearBR, 0);
        \draw (\xSwishR, -1 * \yd) -- (\xLinearBR, -1 * \yd);
        \draw[comp] (\xLinearBL, \yw) rectangle node {$\opFC{\ell_1}$} (\xLinearBR, -\yw);
        % plus
        \draw (\xLinearBR, 0) -- (\xPlusL, 0);
        \draw (\xLinearBR, -1 * \yd) -- (\xPlusL, -1 * \yd);
        \draw (\xPlusL,        0) -- (\xPlusR, -0.5 * \yd);
        \draw (\xPlusL, -1 * \yd) -- (\xPlusR, -0.5 * \yd);
        \draw (\xPlusR, -0.5 * \yd) -- (\xEnd, -0.5 * \yd) node[right] {$\Real^{n_{\out}}$};
        \draw[purecomp] (\xPlusR, -0.5 * \yd) circle [x radius=0.1cm, y radius=0.1cm] node {$+$};
    \end{tikzpicture}
    \]
    where
    \begin{tikzpicture}[baseline=(current bounding box.center)]
        \pgfmathsetmacro{\yw}{0.3}
        \pgfmathsetmacro{\yd}{0.3}
        \pgfmathsetmacro{\xd}{0.3}
        \pgfmathsetmacro{\xw}{2.3}
        \pgfmathsetmacro{\xwDiag}{0.2}
        \pgfmathsetmacro{\xStart}{0}
        \pgfmathsetmacro{\xPlusL}{\xStart + \xd}
        \pgfmathsetmacro{\xPlusR}{\xPlusL + \xwDiag}
        \pgfmathsetmacro{\xEnd}{\xPlusR + \xd}
        \draw (\xStart, 0) -- (\xPlusL, 0);
        \draw (\xStart, -1 * \yd) -- (\xPlusL, -1 * \yd);
        \draw (\xPlusL,        0) -- (\xPlusR, -0.5 * \yd);
        \draw (\xPlusL, -1 * \yd) -- (\xPlusR, -0.5 * \yd);
        \draw (\xPlusR, -0.5 * \yd) -- (\xEnd, -0.5 * \yd);
        \draw[purecomp] (\xPlusR, -0.5 * \yd) circle [x radius=0.1cm, y radius=0.1cm] node {$+$};
    \end{tikzpicture}
    is $\eta(\pi_1 + \pi_2) \in \promonad(X \times X, X)$.
    The diagram is essentially the same as the typical illustrations of residual networks, for example, those in \cite[Figure~2]{He+2016}.
\end{example}

\begin{example}[Autoencoder, continued from Example~\ref{ex:autoencoder}]
    \label{ex:autoencoder-denotation}
    The interpretation
    of the command $\itp{\prevhandle(x)\tuple{x_{\inp}}.R_{\autoencoder}\; \pwith \; H_{\autoencoder}}$
    % \in \promonad_{\signatopr}(\Real^{n_{\inp}}, \Real^{n_{\out}})
    is the left diagram below, and
    the interpretation of deeper autoencoder $\itp{\prevhandle(x) \tuple{x_{\inp}}. R'_{\autoencoder} \; \pwith \; H_{\autoencoder}}$ is the right diagram below:
    \[
        \begin{tikzpicture}[scale=0.8, transform shape]
            \pgfmathsetmacro{\yw}{0.3}
            \pgfmathsetmacro{\yd}{0.2}
            \pgfmathsetmacro{\xd}{0.3}
            \pgfmathsetmacro{\xw}{1.6}
            \pgfmathsetmacro{\xwDiag}{0.2}
            \pgfmathsetmacro{\xwSym}{0.2}
            \pgfmathsetmacro{\xStart}{0}
            \pgfmathsetmacro{\xLinearAL}{\xStart + \xd}
            \pgfmathsetmacro{\xLinearAR}{\xLinearAL + \xw}
            \pgfmathsetmacro{\xLinearBL}{\xLinearAR + \xd}
            \pgfmathsetmacro{\xLinearBR}{\xLinearBL + \xw}
            \pgfmathsetmacro{\xEnd}{\xLinearBR + \xd}
            % Linear A
            \draw (\xStart, 0) node[left] {$\Real^{n_{\inp}}$} -- (\xLinearAR, 0);
            \draw[comp] (\xLinearAL, \yw) rectangle node {$\opFC{\ell_0}$} (\xLinearAR, -\yw);
            \draw decorate[decoration={brace, amplitude=2mm}] { (\xLinearAL-0.1,\yw+\yd) to node[above=1mm] {encoding} (\xLinearAR+0.1,\yw+\yd) };
            % Linear B
            \draw (\xLinearAR, 0) -- (\xEnd, 0) node[right] {$\Real^{n_{\inp}}$};
            \draw[comp] (\xLinearBL, \yw) rectangle node {$\opFC{\ell_1}$} (\xLinearBR, -\yw);
            \draw decorate[decoration={brace, amplitude=2mm}] { (\xLinearBL-0.1,\yw+\yd) to node[above=1mm] {decoding} (\xLinearBR+0.1,\yw+\yd) };
        \end{tikzpicture},
    \begin{tikzpicture}[scale=0.8, transform shape]
        \pgfmathsetmacro{\yw}{0.3}
        \pgfmathsetmacro{\yd}{0.2}
        \pgfmathsetmacro{\xd}{0.3}
        \pgfmathsetmacro{\xw}{1.6}
        \pgfmathsetmacro{\xwDiag}{0.2}
        \pgfmathsetmacro{\xwSym}{0.2}
        \pgfmathsetmacro{\xStart}{0}
        \pgfmathsetmacro{\xLinearAL}{\xStart + \xd}
        \pgfmathsetmacro{\xLinearAR}{\xLinearAL + \xw}
        \pgfmathsetmacro{\xLinearBL}{\xLinearAR + \xd}
        \pgfmathsetmacro{\xLinearBR}{\xLinearBL + \xw}
        \pgfmathsetmacro{\xLinearCL}{\xLinearBR + \xd}
        \pgfmathsetmacro{\xLinearCR}{\xLinearCL + \xw}
        \pgfmathsetmacro{\xLinearDL}{\xLinearCR + \xd}
        \pgfmathsetmacro{\xLinearDR}{\xLinearDL + \xw}
        \pgfmathsetmacro{\xEnd}{\xLinearDR + \xd}
        % Linear A
        \draw (\xStart, 0) node[left] {$\Real^{n_{\inp}}$} -- (\xLinearAR, 0);
        \draw[comp] (\xLinearAL, \yw) rectangle node {$\opFC{\ell_0}$} (\xLinearAR, -\yw);
        % Linear B
        \draw (\xLinearAR, 0) -- (\xLinearBR, 0);
        \draw[comp] (\xLinearBL, \yw) rectangle node {$\opFC{\ell_1}$} (\xLinearBR, -\yw);
        % Linear C
        \draw (\xLinearBR, 0) -- (\xLinearCR, 0);
        \draw[comp] (\xLinearCL, \yw) rectangle node {$\opFC{\ell_2}$} (\xLinearCR, -\yw);
        % Linear D
        \draw (\xLinearCR, 0) -- (\xEnd, 0) node[right] {$\Real^{n_{\inp}}$};
        \draw[comp] (\xLinearDL, \yw) rectangle node {$\opFC{\ell_3}$} (\xLinearDR, -\yw);
        \draw decorate[decoration={brace, amplitude=2mm}] { (\xLinearAL-0.1,\yw+\yd) to node[above=1mm] {encoding} (\xLinearBR+0.1,\yw+\yd) };
        \draw decorate[decoration={brace, amplitude=2mm}] { (\xLinearCL-0.1,\yw+\yd) to node[above=1mm] {decoding} (\xLinearDR+0.1,\yw+\yd) };
    \end{tikzpicture}.
    \]
    The interpretation
    $\itp{\prevhandle(\fromvec{v}) \tuple{x_{\inp}}.
    \left(\prevhandle(x)\tuple{x_{\inp}}. R_{\autoencoder} \; \pwith \; H_{\autoencoder} \right)
    \; \pwith \; H_{\MLP}}$
    describes training of the autoencoder $R_{\autoencoder}$.
\end{example}

In Example~\ref{ex:convolutional-neural-network},
the operation $\oprpool$ is implemented by $\fnpool$, intended to be interpreted as the usual max-pooling function (see Definition~\ref{def:interpretation-cnn-fnconv-fnpool}).
Such a function is not available in $\Smooth$; therefore, in interpreting the CNN examples (Examples~\ref{ex:convolutional-neural-network} and~\ref{ex:U-net}),
we use $\PartSmooth$ as the model.
See Appendix~\ref{sec:details-examples} for details.
For the handler $H_{\Unet}$, we use the following interpretation:
\begin{example}[U-net, continued from Example~\ref{ex:U-net}]
    \label{ex:U-net-denotation}
    The string diagram of the interpretation
    is:
    \[
        \itp{\prevhandle(x_{\inp}) \tuple{x_{\inp}}. R_{\Unet} \; \pwith \; H_{\Unet}}
        =
        \begin{tikzpicture}[baseline=(current bounding box.center), scale=0.75, transform shape]
            \pgfmathsetmacro{\yw}{0.1}
            \pgfmathsetmacro{\yd}{0.4}
            \pgfmathsetmacro{\xd}{0.2}
            \pgfmathsetmacro{\xw}{1}
            \pgfmathsetmacro{\xwDiag}{0.2}
            \pgfmathsetmacro{\xStart}{0}
            \pgfmathsetmacro{\xQfAL}{\xStart + \xd}
            \pgfmathsetmacro{\xQfAR}{\xQfAL + \xw}
            \pgfmathsetmacro{\xQfBL}{\xQfAR + \xd}
            \pgfmathsetmacro{\xQfBR}{\xQfBL + \xw}
            \pgfmathsetmacro{\xPL}{\xQfBR + \xd}
            \pgfmathsetmacro{\xPR}{\xPL + \xw}
            \pgfmathsetmacro{\xQbBL}{\xPR + \xd}
            \pgfmathsetmacro{\xQbBR}{\xQbBL + \xw}
            \pgfmathsetmacro{\xQbAL}{\xQbBR + \xd}
            \pgfmathsetmacro{\xQbAR}{\xQbAL + \xw}
            \pgfmathsetmacro{\xEnd}{\xQbAR + \xd}
            % Q^fwd A
            \draw (\xStart, 0) node[left] {$\Real^{n_{\inp}}$} -- (\xQfAR, 0);
            \draw[comp] (\xQfAL, \yd+\yw) rectangle node {$\itp{Q^{\fwd}_{\opru}}$} (\xQfAR, -\yw);
            % Q^fwd B
            \draw (\xQfAR, \yd) -- (\xQfBR, \yd);
            \draw (\xQfAR,   0) -- (\xQfBR,   0);
            \draw[comp] (\xQfBL, 2 * \yd + \yw) rectangle node {$\itp{Q^{\fwd}_{\opru}}$} (\xQfBR, \yd - \yw);
            % P
            \draw (\xQfBR, 2 * \yd) -- (\xPR, 2 * \yd);
            \draw (\xQfBR,     \yd) -- (\xPR,     \yd);
            \draw (\xQfBR,       0) -- (\xPR,       0);
            \draw[comp] (\xPL, 3 * \yd + \yw) rectangle node {$\itp{P}$} (\xPR, 2 * \yd - \yw);
            % Q^bwd B
            \draw (\xPR, 2 * \yd) -- (\xQbBR, 2 * \yd);
            \draw (\xPR,     \yd) -- (\xQbBR,     \yd);
            \draw (\xPR,       0) -- (\xQbBR,       0);
            \draw[comp] (\xQbBL, 2 * \yd + \yw) rectangle node {$\itp{Q^{\bwd}_{\opru}}$} (\xQbBR, \yd - \yw);
            % Q^bwd A
            \draw (\xQbBR,     \yd) -- (\xQbAR,     \yd);
            \draw (\xQbBR,       0) -- (\xEnd, 0) node[right] {$\Real^{n_{\inp}}$};
            \draw[comp] (\xQbAL, \yd+\yw) rectangle node {$\itp{Q^{\bwd}_{\opru}}$} (\xQbAR, -\yw);
        \end{tikzpicture}.
    \]
    The name ``U-net'' comes from the fact that a diagram of its architecture has the shape of the letter ``U.'' One can see that the string diagram above indeed has the shape of a ``U,'' albeit upside down.
\end{example}

Finally, for the STE example (Example~\ref{ex:ste}), we interpret $\fnround$ by the usual round function; therefore, the interpretation is also taken in $\PartSmooth$.

%In the STE example (Example~\ref{ex:ste}), the round function is also not differentiable at $n + 1/2$ for each integer $n$.
%If we interpret $\fnround$ by a smooth approximation of the round function, for example, $\itp{\fnround}(x) = x - \sin(2\pi x) / (2\pi)$, then we can give interpretation of the STE example.
%
\section{Soundness and Adequacy}\label{sec:soundness-adequacy}
We prove the soundness and adequacy of the denotational semantics defined in \S\ref{sec:denotational-semantics} with respect to the operational semantics defined in \S\ref{sec:operational-semantics}.

Let $(\catx, \promonad_{\signatopr})$ be a model.
We assume that the evaluation function $\Eval_A \colon \signatfunc(A, B) \times \Val_A \to \Val_B$ and the commutative monoid structure $(\signatfunc(1, \beta), +_{\beta}, 0_{\beta})$ are compatible with the semantic structure in the following sense. Such system is called \emph{differentially denotational interpretation structure} in \cite[\S{}5]{Cruttwell+2021}.
\begin{enumerate}
    \item
        For each closed value $V$ of type $A$, the interpretation $\itp{V}:1\to\itp{A}$ is total.
    \item
        For $\func : A \to B$ and $V \in \Val_A$,
        $\Eval(\func,V)$ is defined if and only if $\itp{V}\semicolon\itp{\func}$ is total, and if $\Eval(\func,V)$ is defined, then
        $\itp{\Eval(\func, V)} = \itp{V} \semicolon \itp{\func}$ holds.
    \item
        $\itp{\const_1 +_{\beta} \const_2} = \itp{\const_1} +_{\catx(1, \itp{\beta})} \itp{\const_2}$
        holds for each constant symbols $\const_1, \const_2 \in \signatfunc(1, \beta)$.
\end{enumerate}
The soundness theorem means the invariance of the semantics under the small-step reduction:
\begin{theorem}[Soundness]\label{thm:soundness}
    The following hold:
    \begin{enumerate}
        \item\label{item:soundness-term} For any well-typed term $\emptyenv \vdash M : A$, if $M \redto M'$, then $\itp{M} = \itp{M'}$ holds.
        \item\label{item:soundness-command} For any well-typed command $\emptyenv \asemicolon \emptyenv \vdash P : A$, if $P \redto P'$, then $\itp{P} = \itp{P'}$ holds.
    \end{enumerate}
\end{theorem}
\begin{appendixproof}[Proof of Theorem~\ref{thm:soundness}]
    \itemref{item:soundness-term}
    This statement is essentially a simplified version of the soundness of the categorical semantics of a simple differentiable programming language; see \cite{Cruttwell+2021}.

    \itemref{item:soundness-command}
    First, we show $\eta(\itp{V}^{\emptyenv}_{\Delta'}) \musemicolon \itp{H}_{\itp{\Delta}}(\itp{R}) = \itp{\rewriterevh{V}{H} \tuple{y_j}_{j=1}^m . R}$ for
    $\emptyenv \asemicolon \Delta \vdash R : C$,
    $\Delta' \vdash V : \typprod_{j=1}^m B_j$ and
    $\revhandlerjudgment H : C$
    by induction on
    $(\max(\rhdepth{R},\rhdepth{H}), \dsize{\emptyenv \asemicolon \Delta \vdash R : C} + \dsize{\revhandlerjudgment H : C}) \in (\Nat^2, \le_{\Nat^2})$.
    \begin{itemize}
        \item $R = \pret{y_j}$. We have
            \[
            \begin{aligned}
                & \eta(\itp{V}^{\emptyenv}_{\Delta'}) \musemicolon \itp{H}_{\Delta}(\itp{\pret{y_j}})
                \\
                & = \eta(\langle\itp{V_j}^{\emptyenv}_{\Delta'}\rangle_{j=1}^m) \musemicolon \itp{H}_{\Delta}(\eta(\pi_j))
                \\
                & = \eta(\langle\itp{V_j}^{\emptyenv}_{\Delta'}\rangle_{j=1}^m) \musemicolon \eta(\delta_{\itp{\Delta}} \semicolon (\pi_j \times \idmor_{\itp{\Delta}})) \musemicolon \strength_{\itp{\Delta}}(\itp{P}) \musemicolon \eta(\RD'[\pi_j])
                \\
                & = \eta(\langle\itp{V_j}^{\emptyenv}_{\Delta'}\rangle_{j=1}^m) \musemicolon \eta(\delta_{\itp{\Delta}} \semicolon (\pi_j \times \idmor_{\itp{\Delta}})) \musemicolon \strength_{\itp{\Delta}}(\itp{P}) \musemicolon \eta(\lambda(y, \_). (w_{j'})_{j' = 1}^m)
                \\
                & \qquad \text{where $w_j = y$ and $w_{j'} = 0$ for $j' \ne j$}
                \\
                & = \eta(\itp{V_j}^{\emptyenv}_{\Delta}) \musemicolon \itp{P} \musemicolon \eta(\lambda y . (w_{j'})_{j'=1}^m)
                \\
                & = \itp{\plet \; y \pbind P[V_k / x] \; \pin \; \pret{\tuple{W_{j'}}_{j'=1}^m}}
                \\
                & \qquad \text{where $W_j = y$ and $W_{j'} = \fromvec{0}$ for $j' \ne j$.}
                \\
                & =
                \itp{\rewriterevh{V}{H} \tuple{y_j}_{j=1}^m . \pret{y_j}}.
            \end{aligned}
            \]
        \item $R = \ctxfc[\plet \; x \pbind \pret{y_j} \; \pin \; R']$.
            We have $\itp{\plet \; x \pbind \pret{y_j} \; \pin \; R'} = \itp{R'[y_j/x]}$.
            Hence, by the induction hypothesis, we have
            \[
            \begin{aligned}
                \eta(\itp{V}^{\emptyenv}_{\Delta'}) \musemicolon \itp{H}_{\Delta}(\itp{\ctxfc[\plet \; x \pbind \pret{y_j} \; \pin \; R']})
                & =
                \eta(\itp{V}^{\emptyenv}_{\Delta'}) \musemicolon \itp{H}_{\Delta}(\itp{\ctxfc[R'[y_j / x]]})
                \\
                & =
                \itp{\rewriterevh{V}{H} \tuple{y_j}_{j=1}^m . \ctxfc[R'[y_j / x]]}.
            \end{aligned}
            \]
        \item $R = \ctxfc[\opr(y_j)]$. By the induction hypothesis, we have
            \[
            \begin{aligned}
                & \eta(\itp{V}^{\emptyenv}_{\Delta'}) \musemicolon \itp{H}_{\itp{\Delta}}(\itp{\ctxfc[\opr(y_j)]})
                \\
                & = \eta(\itp{V}^{\emptyenv}_{\Delta'}) \musemicolon \itp{H}_{\itp{\Delta}}\left(\eta(\delta_{\itp{\Delta}}) \musemicolon \strength_{\itp{\Delta}}(\eta(\pi_j) \musemicolon \opr) \musemicolon \itp{\ctxfc[y]}\right)
                \\
                & = \eta(\itp{V}^{\emptyenv}_{\Delta'}) \musemicolon \itp{H}_{\itp{\Delta}}\left(
                    \begin{tikzpicture}[baseline=(current bounding box.center)]
                        \pgfmathsetmacro{\yd}{0.3}
                        \pgfmathsetmacro{\xd}{0.2}
                        \pgfmathsetmacro{\xStart}{0}
                        \pgfmathsetmacro{\xDiagL}{\xStart + \xd}
                        \pgfmathsetmacro{\xDiagR}{\xDiagL + 0.2}
                        \pgfmathsetmacro{\xProjL}{\xDiagR + \xd}
                        \pgfmathsetmacro{\xProjR}{\xProjL + 0.5}
                        \pgfmathsetmacro{\xOpL}{\xProjR + \xd}
                        \pgfmathsetmacro{\xOpR}{\xOpL + 0.8}
                        \pgfmathsetmacro{\xRL}{\xOpR + \xd}
                        \pgfmathsetmacro{\xRR}{\xRL + 2}
                        \pgfmathsetmacro{\xEnd}{\xRR + \xd}
                        % Diagonal
                        \draw (\xStart, 0) node[left] {$\itp{\Delta}$} -- (\xDiagL, 0);
                        \draw (\xDiagL, 0) -- (\xDiagR,  \yd);
                        \draw (\xDiagL, 0) -- (\xDiagR, -\yd);
                        % st(\pi_j)
                        \draw (\xDiagR,  \yd) -- (\xProjR,  \yd);
                        \draw (\xDiagR, -\yd) -- (\xProjR, -\yd);
                        \draw[purecomp] (\xProjL, \yd + 0.3) rectangle node {$\pi_j$} (\xProjR, \yd - 0.3);
                        % st(Op)
                        \draw (\xProjR,  \yd) -- (\xOpR,  \yd);
                        \draw (\xProjR, -\yd) -- (\xOpR, -\yd);
                        \draw[comp] (\xOpL, \yd + 0.3) rectangle node {$\opr$} (\xOpR, \yd - 0.3);
                        % R'
                        \draw (\xOpR,  \yd) -- (\xRR,  \yd);
                        \draw (\xOpR, -\yd) -- (\xRR, -\yd);
                        \draw (\xRR, 0) -- (\xEnd, 0) node[right] {$\itp{C}$};
                        \draw[comp] (\xRL, \yd + 0.2) rectangle node {$\itp{\ctxfc[\pret{y}]}$} (\xRR, -\yd - 0.2);
                    \end{tikzpicture}
                \right)
                \\
                & =
                \begin{tikzpicture}[baseline=(current bounding box.center)]
                    \pgfmathsetmacro{\ytop}{0.4}
                    \pgfmathsetmacro{\ymid}{0}
                    \pgfmathsetmacro{\ybot}{-0.4}
                    \pgfmathsetmacro{\yd}{0.2}
                    \pgfmathsetmacro{\xd}{0.2}
                    \pgfmathsetmacro{\xw}{1.1}
                    \pgfmathsetmacro{\xwDiag}{0.2}
                    \pgfmathsetmacro{\xwSym}{0.2}
                    \pgfmathsetmacro{\xStart}{0}
                    \pgfmathsetmacro{\xVL}{\xStart + \xd}
                    \pgfmathsetmacro{\xVR}{\xVL + 1}
                    \pgfmathsetmacro{\xDiagL}{\xVR + \xd}
                    \pgfmathsetmacro{\xDiagR}{\xDiagL + 0.2}
                    \pgfmathsetmacro{\xProjL}{\xDiagR + \xd}
                    \pgfmathsetmacro{\xProjR}{\xProjL + 0.5}
                    \pgfmathsetmacro{\xQfL}{\xProjR + \xd}
                    \pgfmathsetmacro{\xQfR}{\xQfL + \xw}
                    \pgfmathsetmacro{\xSymlL}{\xQfR + \xd}
                    \pgfmathsetmacro{\xSymlR}{\xSymlL + \xwSym}
                    \pgfmathsetmacro{\xStKL}{\xSymlR + \xd}
                    \pgfmathsetmacro{\xStKR}{\xStKL + 2.7 * \xw}
                    \pgfmathsetmacro{\xSymrL}{\xStKR + \xd}
                    \pgfmathsetmacro{\xSymrR}{\xSymrL + \xwSym}
                    \pgfmathsetmacro{\xQbL}{\xSymrR + \xd}
                    \pgfmathsetmacro{\xQbR}{\xQbL + \xw}
                    \pgfmathsetmacro{\xTupleL}{\xQbR + \xd}
                    \pgfmathsetmacro{\xTupleR}{\xTupleL + 2.2 * \xw}
                    \pgfmathsetmacro{\xEnd}{\xTupleR + \xd}
                    % V
                    \draw (\xStart, 0) node[left] {$\itp{\Delta'}$} -- (\xVL, 0);
                    \draw[purecomp] (\xVL, \ymid + 0.3) rectangle node {$\itp{V}^{\emptyenv}_{\Delta'}$} (\xVR, \ymid - 0.3);
                    % Diagonal
                    \draw (\xVR, 0) -- (\xDiagL, 0);
                    \draw (\xDiagL, 0) -- (\xDiagR, \ytop);
                    \draw (\xDiagL, 0) -- (\xDiagR, \ybot);
                    % st(\pi_j)
                    \draw (\xDiagR, \ytop) -- (\xProjR, \ytop);
                    \draw (\xDiagR, \ybot) -- (\xProjR, \ybot);
                    \draw[purecomp] (\xProjL, \ytop + 0.3) rectangle node {$\pi_j$} (\xProjR, \ytop - 0.3);
                    % Q^fwd
                    \draw (\xProjR, \ytop) -- (\xQfR, \ytop);
                    \draw (\xProjR, \ybot) -- (\xQfR, \ybot);
                    \draw[comp] (\xQfL, \ytop + \yd) rectangle node {$\itp{Q^{\fwd}_{\opr}}$} (\xQfR, \ymid - \yd);
                    % id x sym
                    \draw (\xQfR, \ytop) -- (\xSymlR, \ytop);
                    \draw (\xQfR, \ymid) -- (\xSymlL, \ymid) -- (\xSymlR, \ybot);
                    \draw (\xQfR, \ybot) -- (\xSymlL, \ybot) -- (\xSymlR, \ymid);
                    % st(k)
                    \draw (\xSymlR, \ytop) -- (\xStKR, \ytop);
                    \draw (\xSymlR, \ymid) -- (\xStKR, \ymid);
                    \draw (\xSymlR, \ybot) -- (\xStKR, \ybot);
                    \draw[comp] (\xStKL, \ytop + \yd) rectangle node {$\itp{H}_{\itp{\Delta}}(\itp{\ctxfc[\pret{y}]})$} (\xStKR, \ymid - \yd);
                    % id x sym
                    \draw (\xStKR, \ytop) -- (\xSymrR, \ytop);
                    \draw (\xStKR, \ymid) -- (\xSymrL, \ymid) -- (\xSymrR, \ybot);
                    \draw (\xStKR, \ybot) -- (\xSymrL, \ybot) -- (\xSymrR, \ymid);
                    % Q^bwd
                    \draw (\xSymrR, \ytop) -- (\xQbR, \ytop);
                    \draw (\xSymrR, \ymid) -- (\xQbR, \ymid);
                    \draw (\xSymrR, \ybot) -- (\xQbR, \ybot);
                    \draw[comp] (\xQbL, \ytop + \yd) rectangle node {$\itp{Q^{\bwd}_{\opr}}$} (\xQbR, \ymid - \yd);
                    % tuple
                    \draw (\xQbR, \ytop) -- (\xTupleL, \ytop);
                    \draw (\xQbR, \ybot) -- (\xTupleL, \ybot);
                    \draw[purecomp] (\xTupleL, \ytop + \yd) rectangle node {$\begin{aligned} & \lambda(y''_j, (y'_{j'})_{j'=1}^m). \\ &(z_{j'})_{j'=1}^m \end{aligned}$} (\xTupleR, \ybot - \yd);
                    \draw (\xTupleR, \ymid) -- (\xEnd, \ymid) node[right] {$\itp{\Delta}$};
                \end{tikzpicture}
                \\
                & \qquad \text{where $z_j = y'_j + y''_j$ and $z_{j'} = y'_{j'}$ for $j' \ne j$}
                \\
                & =
                \begin{tikzpicture}[baseline=(current bounding box.center)]
                    \pgfmathsetmacro{\ytop}{0.4}
                    \pgfmathsetmacro{\ymid}{0}
                    \pgfmathsetmacro{\ybot}{-0.4}
                    \pgfmathsetmacro{\yd}{0.2}
                    \pgfmathsetmacro{\xd}{0.15}
                    \pgfmathsetmacro{\xw}{1.1}
                    \pgfmathsetmacro{\xwDiag}{0.2}
                    \pgfmathsetmacro{\xwSym}{0.2}
                    \pgfmathsetmacro{\xStart}{0}
                    \pgfmathsetmacro{\xDiagL}{\xStart + \xd}
                    \pgfmathsetmacro{\xDiagR}{\xDiagL + \xwDiag}
                    \pgfmathsetmacro{\xVjL}{\xDiagR + \xd}
                    \pgfmathsetmacro{\xVjR}{\xVjL + 1}
                    \pgfmathsetmacro{\xQfL}{\xVjR + \xd}
                    \pgfmathsetmacro{\xQfR}{\xQfL + \xw}
                    \pgfmathsetmacro{\xSymlL}{\xQfR + \xd}
                    \pgfmathsetmacro{\xSymlR}{\xSymlL + \xwSym}
                    \pgfmathsetmacro{\xVL}{\xSymlR + \xd}
                    \pgfmathsetmacro{\xVR}{\xVL + \xw}
                    \pgfmathsetmacro{\xStKL}{\xVR + \xd}
                    \pgfmathsetmacro{\xStKR}{\xStKL + 2.7 * \xw}
                    \pgfmathsetmacro{\xSymrL}{\xStKR + \xd}
                    \pgfmathsetmacro{\xSymrR}{\xSymrL + \xwSym}
                    \pgfmathsetmacro{\xQbL}{\xSymrR + \xd}
                    \pgfmathsetmacro{\xQbR}{\xQbL + \xw}
                    \pgfmathsetmacro{\xTupleL}{\xQbR + \xd}
                    \pgfmathsetmacro{\xTupleR}{\xTupleL + 2.2 * \xw}
                    \pgfmathsetmacro{\xEnd}{\xTupleR + \xd}
                    % Diagonal
                    \draw (\xStart, 0) node[left] {$\itp{\Delta'}$} -- (\xDiagL, 0);
                    \draw (\xDiagL, 0) -- (\xDiagR, \ytop);
                    \draw (\xDiagL, 0) -- (\xDiagR, \ybot);
                    % st(Vj)
                    \draw (\xDiagR, \ytop) -- (\xVjR, \ytop);
                    \draw (\xDiagR, \ybot) -- (\xVjR, \ybot);
                    \draw[purecomp] (\xVjL, \ytop + 0.3) rectangle node {$\itp{V_j}^{\emptyenv}_{\Delta'}$} (\xVjR, \ytop - 0.3);
                    % Q^fwd
                    \draw (\xVjR, \ytop) -- (\xQfR, \ytop);
                    \draw (\xVjR, \ybot) -- (\xQfR, \ybot);
                    \draw[comp] (\xQfL, \ytop + \yd) rectangle node {$\itp{Q^{\fwd}_{\opr}}$} (\xQfR, \ymid - \yd);
                    % id x sym
                    \draw (\xQfR, \ytop) -- (\xSymlR, \ytop);
                    \draw (\xQfR, \ymid) -- (\xSymlL, \ymid) -- (\xSymlR, \ybot);
                    \draw (\xQfR, \ybot) -- (\xSymlL, \ybot) -- (\xSymlR, \ymid);
                    % V
                    \draw (\xSymlR, \ytop) -- (\xVR, \ytop);
                    \draw (\xSymlR, \ymid) -- (\xVR, \ymid);
                    \draw (\xSymlR, \ybot) -- (\xVR, \ybot);
                    \draw[purecomp] (\xVL, \ymid + 1.2 * \yd) rectangle node {$\itp{V}^{\emptyenv}_{\Delta'}$} (\xVR, \ymid - 1.2 * \yd);
                    % st(k)
                    \draw (\xVR, \ytop) -- (\xStKR, \ytop);
                    \draw (\xVR, \ymid) -- (\xStKR, \ymid);
                    \draw (\xVR, \ybot) -- (\xStKR, \ybot);
                    \draw[comp] (\xStKL, \ytop + \yd) rectangle node {$\itp{H}_{\itp{\Delta}}(\itp{\ctxfc[\pret{y}]})$} (\xStKR, \ymid - \yd);
                    % id x sym
                    \draw (\xStKR, \ytop) -- (\xSymrR, \ytop);
                    \draw (\xStKR, \ymid) -- (\xSymrL, \ymid) -- (\xSymrR, \ybot);
                    \draw (\xStKR, \ybot) -- (\xSymrL, \ybot) -- (\xSymrR, \ymid);
                    % Q^bwd
                    \draw (\xSymrR, \ytop) -- (\xQbR, \ytop);
                    \draw (\xSymrR, \ymid) -- (\xQbR, \ymid);
                    \draw (\xSymrR, \ybot) -- (\xQbR, \ybot);
                    \draw[comp] (\xQbL, \ytop + \yd) rectangle node {$\itp{Q^{\bwd}_{\opr}}$} (\xQbR, \ymid - \yd);
                    % tuple
                    \draw (\xQbR, \ytop) -- (\xTupleL, \ytop);
                    \draw (\xQbR, \ybot) -- (\xTupleL, \ybot);
                    \draw[purecomp] (\xTupleL, \ytop + \yd) rectangle node {$\begin{aligned} & \lambda(y''_j, (y'_{j'})_{j'=1}^m). \\ &(z_{j'})_{j'=1}^m \end{aligned}$} (\xTupleR, \ybot - \yd);
                    \draw (\xTupleR, \ymid) -- (\xEnd, \ymid) node[right] {$\itp{\Delta}$};
                \end{tikzpicture}
                \\
                & =
                \bigitp{
                    \begin{aligned}
                        & \plet \; \tuple{y,z'} \pbind Q^{\fwd}_{\opr}[V_i/x] \; \pin \\
                        & \plet \; \tuple{y', y'_1, \dots, y'_m} \pbind \rewriterevh{\tuple{y, V_1, \dots, V_m}}{H}(\tuple{y, y_1, \dots, y_m}.\ctxfc[\pret{y}]) \; \pin\\
                        & \plet \; y''_i \pbind Q^{\bwd}_{\opr}[y'/y, z'/z] \; \pin \\
                        & \pret{\langle y'_1, \dots, y'_{i-1}, y'_i + y''_i, y'_{i+1}, \dots, x'_m \rangle}
                    \end{aligned}
                }
                \\
                & =
                \itp{\rewriterevh{V}{H}(\tuple{y_j}_{j=1}^m. \ctxfc[\opr(y_j)])}.
            \end{aligned}
            \]
        \item $R = \ctxfc[\prevhandle(W) \tuple{z_k}_{k=1}^l. R' \; \pwith \; H']$.
            By the induction hypothesis, we have
            \[
            \begin{aligned}
                \itp{\prevhandle(W) \tuple{z_j}_{j=1}^l. R' \; \pwith \; H'}
                & =
                \eta(\itp{W}^{\emptyenv}_{\Delta}) \musemicolon \itp{H'}_{\Delta''}(\itp{R'})
                \\
                & =
                \itp{\rewriterevh{W}{H'}(\tuple{z_k}_{k=1}^l. R')}.
            \end{aligned}
            \]
            Hence, by the induction hypothesis, we have
            \[
            \begin{aligned}
                & \eta(\itp{V}^{\emptyenv}_{\Delta'}) \musemicolon \itp{H}_{\Delta}(\itp{\ctxfc[\prevhandle(W) \tuple{z_k}_{k=1}^l. R' \; \pwith \; H']})
                \\
                & =
                \eta(\itp{V}^{\emptyenv}_{\Delta'}) \musemicolon \itp{H}_{\Delta}(\itp{\rewriterevh{W}{H'}(\tuple{z_k}_{k=1}^l. R')})
                \\
                & =
                \rewriterevh{V}{H}(\tuple{y_k}_{k=1}^m. \ctxfc[\rewriterevh{W}{H'}(\tuple{z_j}_{j=1}^l. R'])).
            \end{aligned}
            \]
        \item $R = \ctxft[M]$ where $M$ is not a variable.
            By the induction hypothesis, we have
            \[
            \begin{aligned}
                & \eta(\itp{V}^{\emptyenv}_{\Delta'}) \musemicolon \itp{H}_{\Delta}(\itp{\ctxft[M]})
                \\
                & = 
                \eta(\itp{V}^{\emptyenv}_{\Delta'}) \musemicolon \itp{H}_{\Delta}\left(
                    \begin{tikzpicture}[baseline=(current bounding box.center)]
                        \pgfmathsetmacro{\yd}{0.3}
                        \pgfmathsetmacro{\xd}{0.2}
                        \pgfmathsetmacro{\xStart}{0}
                        \pgfmathsetmacro{\xDiagL}{\xStart + \xd}
                        \pgfmathsetmacro{\xDiagR}{\xDiagL + 0.2}
                        \pgfmathsetmacro{\xML}{\xDiagR + \xd}
                        \pgfmathsetmacro{\xMR}{\xML + 0.9}
                        \pgfmathsetmacro{\xRL}{\xMR + \xd}
                        \pgfmathsetmacro{\xRR}{\xRL + 2}
                        \pgfmathsetmacro{\xEnd}{\xRR + \xd}
                        % Diagonal
                        \draw (\xStart, 0) node[left] {$\itp{\Delta}$} -- (\xDiagL, 0);
                        \draw (\xDiagL, 0) -- (\xDiagR,  \yd);
                        \draw (\xDiagL, 0) -- (\xDiagR, -\yd);
                        % st(M)
                        \draw (\xDiagR,  \yd) -- (\xMR,  \yd);
                        \draw (\xDiagR, -\yd) -- (\xMR, -\yd);
                        \draw[purecomp] (\xML, \yd + 0.3) rectangle node {$\itp{M}^{\emptyenv}_{\Delta}$} (\xMR, \yd - 0.3);
                        % F[y]
                        \draw (\xMR,  \yd) -- (\xRR,  \yd);
                        \draw (\xMR, -\yd) -- (\xRR, -\yd);
                        \draw (\xRR, 0) -- (\xEnd, 0) node[right] {$\itp{C}$};
                        \draw[comp] (\xRL, \yd + 0.2) rectangle node {$\itp{\ctxfc[\pret{y}]}$} (\xRR, -\yd - 0.2);
                    \end{tikzpicture}
                \right)
                \\
                & =
                \begin{tikzpicture}[baseline=(current bounding box.center)]
                    \pgfmathsetmacro{\yd}{0.3}
                    \pgfmathsetmacro{\xd}{0.15}
                    \pgfmathsetmacro{\xStart}{0}
                    \pgfmathsetmacro{\xVL}{\xStart + \xd}
                    \pgfmathsetmacro{\xVR}{\xVL + 1}
                    \pgfmathsetmacro{\xDiagAL}{\xVR + \xd}
                    \pgfmathsetmacro{\xDiagAR}{\xDiagAL + 0.2}
                    \pgfmathsetmacro{\xDiagBL}{\xDiagAR + \xd}
                    \pgfmathsetmacro{\xDiagBR}{\xDiagBL + 0.2}
                    \pgfmathsetmacro{\xML}{\xDiagBR + \xd}
                    \pgfmathsetmacro{\xMR}{\xML + 0.9}
                    \pgfmathsetmacro{\xsymlL}{\xMR + \xd}
                    \pgfmathsetmacro{\xsymlR}{\xsymlL + 0.2}
                    \pgfmathsetmacro{\xRL}{\xsymlR + \xd}
                    \pgfmathsetmacro{\xRR}{\xRL + 2.8}
                    \pgfmathsetmacro{\xsymrL}{\xRR + \xd}
                    \pgfmathsetmacro{\xsymrR}{\xsymrL + 0.2}
                    \pgfmathsetmacro{\xRevML}{\xsymrR + \xd}
                    \pgfmathsetmacro{\xRevMR}{\xRevML + 1.7}
                    \pgfmathsetmacro{\xPlusL}{\xRevMR + \xd}
                    \pgfmathsetmacro{\xPlusR}{\xPlusL + 0.2}
                    \pgfmathsetmacro{\xEnd}{\xPlusR + 2 * \xd}
                    % V
                    \draw (\xStart, 0) node[left] {$\itp{\Delta'}$} -- (\xVR, 0);
                    \draw[purecomp] (\xVL, 0.3) rectangle node {$\itp{V}^{\emptyenv}_{\Delta'}$} (\xVR, -0.3);
                    % Diagonal A
                    \draw (\xVR, 0) -- (\xDiagAL, 0);
                    \draw (\xDiagAL, 0) -- (\xDiagAR,  \yd);
                    \draw (\xDiagAL, 0) -- (\xDiagAR, -\yd);
                    % Diagonal B
                    \draw (\xDiagAR, \yd) -- (\xDiagBL, \yd);
                    \draw (\xDiagBL, \yd) -- (\xDiagBR, 2 * \yd);
                    \draw (\xDiagBL, \yd) -- (\xDiagBR, 0);
                    \draw (\xDiagAR, -\yd) -- (\xDiagBR, -\yd);
                    % st(M)
                    \draw (\xDiagBR, 2 * \yd) -- (\xMR, 2 * \yd);
                    \draw (\xDiagBR,       0) -- (\xMR,       0);
                    \draw (\xDiagBR,    -\yd) -- (\xMR,    -\yd);
                    \draw[purecomp] (\xML, 2 * \yd + 0.3) rectangle node {$\itp{M}^{\emptyenv}_{\Delta}$} (\xMR, 2 * \yd - 0.3);
                    % sym l
                    \draw (\xMR, 2 * \yd) -- (\xsymlR, 2 * \yd);
                    \draw (\xMR,       0) -- (\xsymlL,    0) -- (\xsymlR, -\yd);
                    \draw (\xMR,    -\yd) -- (\xsymlL, -\yd) -- (\xsymlR,    0);
                    % F[y]
                    \draw (\xsymlR, 2 * \yd) -- (\xRR, 2 * \yd);
                    \draw (\xsymlR,       0) -- (\xRR,       0);
                    \draw (\xsymlR,    -\yd) -- (\xRR,    -\yd);
                    \draw[comp] (\xRL, 2 * \yd + 0.1) rectangle node {$\itp{H}_{\itp{\Delta}}(\ctxfc[\pret{y}])$} (\xRR, 0 - 0.1);
                    % sym r
                    \draw (\xRR, 2 * \yd) -- (\xsymrR, 2 * \yd);
                    \draw (\xRR,       0) -- (\xsymrL,    0) -- (\xsymrR, -\yd);
                    \draw (\xRR,    -\yd) -- (\xsymrL, -\yd) -- (\xsymrR,    0);
                    % R'[M]
                    \draw (\xsymrR, 2 * \yd) -- (\xRevMR, 2 * \yd);
                    \draw (\xsymrR,       0) -- (\xRevMR,       0);
                    \draw (\xsymrR,    -\yd) -- (\xRevMR,    -\yd);
                    \draw[purecomp] (\xRevML, 2 * \yd + 0.1) rectangle node{$\RD'[\itp{M}^{\emptyenv}_{\Delta}]$} (\xRevMR, 0 - 0.1);
                    % Plus
                    \draw (\xRevMR,      \yd) -- (\xPlusL,      \yd) -- (\xPlusR, 0);
                    \draw (\xRevMR, -1 * \yd) -- (\xPlusL, -1 * \yd) -- (\xPlusR, 0);
                    \draw (\xPlusR,        0) -- (\xEnd,          0) node[right] {$\itp{\Delta}$};
                    \draw[purecomp] (\xPlusR, 0) circle [x radius=0.1cm, y radius=0.1cm] node {$+$};
                \end{tikzpicture}
                \\
                & =
                \begin{tikzpicture}[baseline=(current bounding box.center)]
                    \pgfmathsetmacro{\yd}{0.3}
                    \pgfmathsetmacro{\xd}{0.15}
                    \pgfmathsetmacro{\xStart}{0}
                    \pgfmathsetmacro{\xDiagAL}{\xStart + \xd}
                    \pgfmathsetmacro{\xDiagAR}{\xDiagAL + 0.2}
                    \pgfmathsetmacro{\xDiagBL}{\xDiagAR + \xd}
                    \pgfmathsetmacro{\xDiagBR}{\xDiagBL + 0.2}
                    \pgfmathsetmacro{\xVaL}{\xDiagBR + \xd}
                    \pgfmathsetmacro{\xVaR}{\xVaL + 1}
                    \pgfmathsetmacro{\xML}{\xVaR + \xd}
                    \pgfmathsetmacro{\xMR}{\xML + 0.9}
                    \pgfmathsetmacro{\xVbL}{\xMR + \xd}
                    \pgfmathsetmacro{\xVbR}{\xVbL + 1}
                    \pgfmathsetmacro{\xRL}{\xVbR + \xd}
                    \pgfmathsetmacro{\xRR}{\xRL + 2.8}
                    \pgfmathsetmacro{\xsymrL}{\xRR + \xd}
                    \pgfmathsetmacro{\xsymrR}{\xsymrL + 0.2}
                    \pgfmathsetmacro{\xVcL}{\xsymrR + \xd}
                    \pgfmathsetmacro{\xVcR}{\xVcL + 1}
                    \pgfmathsetmacro{\xRevML}{\xVcR + \xd}
                    \pgfmathsetmacro{\xRevMR}{\xRevML + 1.7}
                    \pgfmathsetmacro{\xPlusL}{\xRevMR + \xd}
                    \pgfmathsetmacro{\xPlusR}{\xPlusL + 0.2}
                    \pgfmathsetmacro{\xEnd}{\xPlusR + 2 * \xd}
                    % Diagonal A
                    \draw (\xStart,  0) node[left] {$\itp{\Delta'}$} -- (\xDiagAL, 0);
                    \draw (\xDiagAL, 0) -- (\xDiagAR,  \yd);
                    \draw (\xDiagAL, 0) -- (\xDiagAR, -\yd);
                    % Diagonal B
                    \draw (\xDiagAR, \yd) -- (\xDiagBL, \yd);
                    \draw (\xDiagBL, \yd) -- (\xDiagBR, 2 * \yd);
                    \draw (\xDiagBL, \yd) -- (\xDiagBR, 0);
                    \draw (\xDiagAR, -\yd) -- (\xDiagBR, -\yd);
                    % Va
                    \draw (\xDiagBR, 2 * \yd) -- (\xVaR, 2 * \yd);
                    \draw (\xDiagBR,       0) -- (\xVaR,       0);
                    \draw (\xDiagBR,    -\yd) -- (\xVaR,    -\yd);
                    \draw[purecomp] (\xVaL, 2 * \yd + 0.3) rectangle node {$\itp{V}^{\emptyenv}_{\Delta'}$} (\xVaR, 2 * \yd - 0.3);
                    % st(M)
                    \draw (\xVaR, 2 * \yd) -- (\xMR, 2 * \yd);
                    \draw (\xVaR,       0) -- (\xMR,       0);
                    \draw (\xVaR,    -\yd) -- (\xMR,    -\yd);
                    \draw[purecomp] (\xML, 2 * \yd + 0.3) rectangle node {$\itp{M}^{\emptyenv}_{\Delta}$} (\xMR, 2 * \yd - 0.3);
                    % Vb
                    \draw (\xMR, 2 * \yd) -- (\xVbR, 2 * \yd);
                    \draw (\xMR,       0) -- (\xVbR,       0);
                    \draw (\xMR,    -\yd) -- (\xVbR,    -\yd);
                    \draw[purecomp] (\xVbL, 0 + 0.4) rectangle node {$\itp{V}^{\emptyenv}_{\Delta'}$} (\xVbR, 0 - 0.2);
                    % F[y]
                    \draw (\xVbR, 2 * \yd) -- (\xRR, 2 * \yd);
                    \draw (\xVbR,       0) -- (\xRR,       0);
                    \draw (\xVbR,    -\yd) -- (\xRR,    -\yd);
                    \draw[comp] (\xRL, 2 * \yd + 0.1) rectangle node {$\itp{H}_{\itp{\Delta}}(\ctxfc[\pret{y}])$} (\xRR, 0 - 0.1);
                    % sym r
                    \draw (\xRR, 2 * \yd) -- (\xsymrR, 2 * \yd);
                    \draw (\xRR,       0) -- (\xsymrL,    0) -- (\xsymrR, -\yd);
                    \draw (\xRR,    -\yd) -- (\xsymrL, -\yd) -- (\xsymrR,    0);
                    % Vb
                    \draw (\xsymrR, 2 * \yd) -- (\xVcR, 2 * \yd);
                    \draw (\xsymrR,       0) -- (\xVcR,       0);
                    \draw (\xsymrR,    -\yd) -- (\xVcR,    -\yd);
                    \draw[purecomp] (\xVcL, 0 + 0.4) rectangle node {$\itp{V}^{\emptyenv}_{\Delta'}$} (\xVcR, 0 - 0.2);
                    % R'[M]
                    \draw (\xVcR, 2 * \yd) -- (\xRevMR, 2 * \yd);
                    \draw (\xVcR,       0) -- (\xRevMR,       0);
                    \draw (\xVcR,    -\yd) -- (\xRevMR,    -\yd);
                    \draw[purecomp] (\xRevML, 2 * \yd + 0.1) rectangle node{$\RD'[\itp{M}^{\emptyenv}_{\Delta}]$} (\xRevMR, 0 - 0.1);
                    % Plus
                    \draw (\xRevMR,      \yd) -- (\xPlusL,      \yd) -- (\xPlusR, 0);
                    \draw (\xRevMR, -1 * \yd) -- (\xPlusL, -1 * \yd) -- (\xPlusR, 0);
                    \draw (\xPlusR,        0) -- (\xEnd,          0) node[right] {$\itp{\Delta}$};
                    \draw[purecomp] (\xPlusR, 0) circle [x radius=0.1cm, y radius=0.1cm] node {$+$};
                \end{tikzpicture}
                \\
                & =
                \bigitp{
                \begin{aligned}
                    & \plet \; y \pbind \pret{M[V_1 / y_1, \dots, V_n / y_n]} \; \pin \\
                    & \plet \; (z, y'_1, \dots, y'_n) \pbind \rewriterevh{\tuple{y, V_1 \dots, V_n}}{H}(\tuple{y, y_1, \dots, y_n}.\ctxft[y]) \; \pin \\
                    & \pret{\langle y'_i + z.\trd(y_i.M[V_1/y_1, \dots, V_{i-1}/y_{i-1}, V_{i+1}/y_{i+1}, \dots, V_m/y_m])(V_i) \rangle_{i = 1}^m}
                \end{aligned}
                }
                \\
                & =
                \itp{
                    \rewriterevh{V}{H}(\tuple{y_j}_{j=1}^m. \ctxft[M])
                }.
            \end{aligned}
            \]
    \end{itemize}

    Next, we show $\itp{P} = \itp{P'}$ if $P \to P'$.
    It suffices to show the following claims:
    \begin{itemize}
        \item $\itp{M} = \itp{M'}$ for $M$ and $M'$ satisfying $M \to M'$,
        \item $\itp{\plet \; x \pbind \pret{V} \; \pin \; R} = \itp{R[V/x]}$, and
        \item $\itp{\prevhandle(V)\tuple{x_i}_{i=1}^n.R \; \pwith \; H} = \itp{\rewriterevh{V}{H}(\tuple{x_i}_{i=1}^n.R)}$.
    \end{itemize}
    The first claim is nothing other than the soundness \itemref{item:soundness-term} for terms.
    The second claim is immediately proven by induction on the structure of $R$.
    The third claim is proven above.
\end{appendixproof}
%\subsection{Adequacy}
To prove the adequacy (Corollary~\ref{cor:adequacy}), we define logical relations between semantic values and syntactic terms/commands.
The definition is similar to that of \cite[Definition~5.14]{Sanada2024}.
\begin{definition}[logical relations, {cf.\ \cite[Definition~5.14]{Sanada2024}}]\label{def:logical-relations}
    We define relations
    $(\logrelt_A) \subseteq \catx(1,\itp{A}) \times \{ M \mid \emptyenv \vdash M : A \}$
    and
    $(\logrelc_A) \subseteq \promonad_{\signatopr}(1, \itp{A}) \times \{ P \mid \emptyenv \asemicolon \emptyenv \vdash P : A \}$
    for each type $A$ as follows:
    \begin{align*}
        v \logrelt_{\beta} M
        & \iff
        M \redto^* \const \text{ and } \itp{\const} = v
        \\
        (v_1, \dots ,v_n) \logrelt_{A_1 \times \dots \times A_n} M
        & \iff
        M \redto^* \tuple{V_1, \dots, V_n}
        \text{ and }
        v_i \logrelt_{A_i} V_i
        \text{ for every $i = 1, \dots, n$}
    \end{align*}
    and $a \logrelc_C P$ if
    \begin{itemize}
        \item $P \redto^* \pret{V}$, $a = [\apure(v)]$ and $v \logrelt_{C} V$ for some $v$ and $V$, or
        \item $P \redto^* \ctxfc[\opr(V)]$,
            $a = [\apure(v) \acomp \opr \acomp b]$
            for a value $V$, operation $\opr : A \opto B \in \signatopr$, $v \in \itp{A}$ and $b$
            satisfying
            $v \logrelt_{A} V$ and
            $[\apure(w) \acomp b] \logrelc_{C} \ctxfc[\pret{W}]$ for any $w \logrelt_{B} W$.
    \end{itemize}
    %Furthermore, we define a relation $(\logrelh_{C}) \subseteq (\lambda D. \promonad_{\signatopr}(\itp{D}, \itp{C}) \expto \ralg{\promonad_{\signatopr}}(\itp{D})) \times \{ H \mid \revhandlerjudgment H : C\}$.
    %It holds that $h \logrelh_{C} H$ if for any $v_i \logrelt V_i$ ($i = 1, \dots, n$),
    %$\emptyenv \asemicolon x_1 : A_1, \dots, x_n : A_n \vdash R : C$ and $a \in \promonad_{\signatopr}(\prod_{i=1}^n \itp{A_i}, \itp{C})$ satisfying $(\apure((v_i)_{i=1}^n) \acomp a) \logrelc_{C} R[V_i / x_i]$,
    %it holds that
    %\[ \apure((v_i)_{i=1}^n) \acomp h_{\prod_{i=1}^n A_i}(a) \logrelc_{C} \prevhandle(\tuple{V_i}_{i=1}^n).R \; \pwith \; H. \]
\end{definition}
Intuitively, in the above definition of $a \logrelc_C P$, the element $b$ is the semantic counterpart of the continuation $\ctxfc[\blank]$.

\begin{toappendix}
\begin{lemma}\label{lem:logrel-backward}
    The following hold for $\emptyenv \vdash M : A$ and $\emptyenv \asemicolon \emptyenv \vdash P : A$:
    \begin{enumerate}
        \item\label{item:logrel-backward:term} If $M \redto^* M'$ and $v \logrelt_A M'$, then $v \logrelt_A M$.
        \item\label{item:logrel-backward:command} If $P \redto^* P'$ and $a \logrelc_A P'$, then $a \logrelc_A P$.
    \end{enumerate}
\end{lemma}
\begin{proof}[Proof of Lemma~\ref{lem:logrel-backward}]
    \itemref{item:logrel-backward:term}
    We prove by case analysis on the structure of the type $A$.

    \textbf{Case} $A = \beta$.
    If $M \redto^* M'$ and $v \logrelt_{\beta} M'$,
    then $M' \redto^* \const$ holds for some $\const$ satisfying $\itp{\const} = v$ by the definition of $\logrelt_A$.
    Hence, we have $M \redto^* M' \redto^* V$.
    Therefore, $v \logrelt_\beta M$ holds.

    \textbf{Case} $A = A_1 \times \cdots \times A_n$.
    If $M \redto^* M'$ and $v \logrelt_{A_1 \times \cdots \times A_n} M'$,
    then $M' \redto^* \tuple{V_1, \dots, V_n}$ holds for some $V_i$ satisfying $v = (v_1, \dots, v_n)$ and $v_i \logrelt V_i$ ($1 \le i \le n$).
    Hence, we have $M \redto^* M' \redto^* \tuple{V_1, \dots, V_n}$.
    Therefore, we have $v \logrelt_A M$.
    
    \itemref{item:logrel-backward:command}
    If $P \redto^* P'$ and $a \logrelc_A P'$,
    then there are the two cases.
    \begin{enumerate}
        \item $P' \redto^* \pret{V}$, $a = \apure(v)$ and $v \logrelt_A V$ for some $v$ and $V$.
            We have $P \redto^* P' \redto^* \pret{V}$.
            Therefore, we have $a \logrelc_A P$ by the definition of $\logrelc_A$.
        \item $P' \redto^* \ctxfc[\opr(V)]$,
            $a = [\apure(v) \acomp \opr \acomp b]$
            for a value $V$, operation $\opr : A \opto B \in \signatopr$, $v \in \itp{A}$ and $b$ satisfying
            $v \logrelt_{A} V$ and
            $[\apure(w) \acomp b] \logrelc_{C} \ctxfc[\pret{W}]$ for any $w \logrelt_{B} W$.
            In this case, we have $P \redto^* P' \redto^* \ctxfc[\opr(V)]$.
            Therefore, we have $a \logrelc_A P$.
    \end{enumerate}
\end{proof}

\begin{lemma}\label{lem:logrel-let-in}
    Let $\Delta = x_1 : A_1, \dots, x_n : A_n$ and let
    $\emptyenv \asemicolon \Delta \vdash P : A$ and $\emptyenv \asemicolon x : A, \Delta \vdash Q : C$ be well-typed commands,
    $v_i \logrelt_{A_i} V_i$ for $1 \leq i \leq n$ and $w \logrelt_A W$.
    If $\apure(v_1, \dots, v_n) \acomp a \logrelc_A P[V_1/x_1, \dots, V_n/x_n]$ and $\apure(w, v_1, \dots, v_n) \acomp b \logrelc_C Q[W/x, V_1/x_1, \dots, V_n/x_n]$ hold, then it holds that
    \[
        \apure((v_i)_{i=1}^n) \acomp \apure(\diag_{\prod_{i=1}^n \itp{A_i}}) \acomp \astr_{\prod_{i=1}^n \itp{A_i}}(a) \acomp b
        %=
        %\begin{tikzpicture}[baseline=(current bounding box.center)]
        %    \pgfmathsetmacro{\dx}{0.3}
        %    \pgfmathsetmacro{\dy}{0.2}
        %    \pgfmathsetmacro{\wy}{0.7}
        %    \pgfmathsetmacro{\wx}{0.8}
        %    \pgfmathsetmacro{\xstart}{0}
        %    \pgfmathsetmacro{\xvl}{\xstart}
        %    \pgfmathsetmacro{\xvr}{\xvl + 1.3 * \wx}
        %    \pgfmathsetmacro{\xdiag}{\xvr + \dx}
        %    \pgfmathsetmacro{\xar}{\xdiag + \dx}
        %    \pgfmathsetmacro{\xal}{\xar + \wx}
        %    \pgfmathsetmacro{\xbl}{\xal + \dx}
        %    \pgfmathsetmacro{\xbr}{\xbl + \wx}
        %    \pgfmathsetmacro{\xend}{\xbr + \dx}
        %    \draw (0,0) -- (\xdiag-0.1, 0) -- (\xdiag+0.1, \wy/2) -- (\xbl, \wy/2);
        %    \draw (\xdiag-0.1, 0) -- (\xdiag+0.1, -\wy/2) -- (\xbl, -\wy/2);
        %    \draw (\xbr, 0) -- (\xend, 0);
        %    \draw[purecomp] (\xvl, \wy/2) rectangle node {$(v_i)_{i=1}^n$} (\xvr, -\wy/2);
        %    \draw[fill=white] (\xal, \wy) rectangle node {$a$} (\xar, 0);
        %    \draw[fill=white] (\xbl, \wy) rectangle node {$b$} (\xbr, -\wy);
        %\end{tikzpicture}
        \logrelc_C
        (\plet \; x \pbind P \; \pin \; Q)[V_1/x_1, \dots, V_n/x_n].
    \]
\end{lemma}
\begin{proof}[Proof of Lemma~\ref{lem:logrel-let-in}]
    Let $\sigma = V_1/x_1, \dots, V_n/x_n$.
    We prove by induction on the number $\opnum{a}$ of operations contained in $a$.

    \textbf{Base case}: $\opnum{a} = 0$.
    By the definition of $\apure((v_i)_{i=1}^n) \acomp a \logrelc_A P[\sigma]$ and $\opnum{a} = 0$, we have
    $P[\sigma] \redto^* \pret{V}$, $\apure(v_1, \dots, v_n) \acomp a = \apure(v)$ and $v \logrelt_{A} V$ for some $v$ and $V$.
    We have
    \[
    \begin{aligned}
        (\plet \; x \pbind P \; \pin \; Q)[\sigma]
        & = \plet \; x \pbind P[\sigma] \; \pin \; Q[\sigma] \\
        & \redto^* \plet \; x \pbind \pret{V} \; \pin \; Q[\sigma] \\
        & \redto Q[V/x, \sigma]
    \end{aligned}
    \]
    and $\apure(v_1, \dots, v_n) \acomp \apure(\diag_{\prod_{i=1}^n \itp{A_i}}) \acomp \astr_{\prod_{i=1}^n \itp{A_i}}(a) \acomp b = \apure(v, v_1, \dots, v_n) \acomp b$.
    We have $\apure(v, v_1, \dots, v_n) \acomp b \logrelc_{C} Q[V/x, \sigma]$ by the assumption.
    Therefore, we have
    \[ \apure(v_1, \dots, v_n) \acomp \apure(\diag_{\prod_{i=1}^n \itp{A_i}}) \acomp \astr_{\prod_{i=1}^n \itp{A_i}}(a) \acomp b
    \logrelc_{C} (\plet \; x \pbind P \; \pin \; Q)[\sigma]\]
    by Lemma~\ref{lem:logrel-backward}.

    \textbf{Induction step}: $\opnum{a} > 0$.
    By the definition of $\apure((v_i)_{i=1}^n) \acomp a \logrelc_A P[\sigma]$ and $\opnum{a} > 0$, we have
    $P[\sigma] \redto^* \ctxfc[\opr(V)]$ and
    $\apure((v_i)_{i=1}^n) \acomp a = \apure(v) \acomp \opr \acomp a'$
    for a value $V$, operation $\opr : B \opto D  \in \signatopr$, $v \in \itp{B}$ and $a'$
    satisfying
    $v \logrelt_{B} V$ and
    $\apure(w) \acomp a' \logrelc_{A} \ctxfc[\pret{W}]$ for any $w \logrelt_{D} W$.
    We have
    \[
    \begin{aligned}
        (\plet \; x \pbind P \; \pin \; Q)[\sigma]
        & =
        \plet \; x \pbind P[\sigma] \; \pin \; Q[\sigma]
        \\
        & \redto^*
        \plet \; x \pbind \ctxfc[\opr(V)] \; \pin \; Q[\sigma]
    \end{aligned}
    \]
    and
    \[
    \begin{aligned}
        & \apure((v_i)_{i=1}^n) \acomp \apure(\diag_{\prod_{i=1}^n \itp{A_i}}) \acomp \astr_{\prod_{i=1}^n \itp{A_i}}(a) \acomp b
        \\
        & =
        \apure(v, v_1, \dots, v_n) \acomp \astr_{\prod_{i=1}^n \itp{A_i}}(\opr \acomp a') \acomp b
        \\
        & =
        \apure(v) \acomp \opr \acomp a' \acomp \apure(\lambda x . (x, v_1, \dots, v_n)) \acomp b
    \end{aligned}
    \]
    Given any $w \logrelt_{\delta} W$,
    we have $\opnum{\apure(w) \acomp a'} < \opnum{\apure(v) \acomp \opr \acomp a'} = \opnum{\apure((v_i)_{i=1}^n) \acomp a}$.
    Hence, by the induction hypothesis, we have
    \[
        (\apure(w) \acomp a' \acomp \apure(\lambda x . (x, v_1, \dots, v_n)) \acomp b)
        \logrelc_{C}
        \plet \; x \pbind \ctxfc[\pret{W}] \; \pin \; Q[\sigma].
    \]
    Therefore, we have
    \[
        \apure((v_i)_{i=1}^n) \acomp \apure(\diag_{\prod_{i=1}^n \itp{A_i}}) \acomp \astr_{\prod_{i=1}^n \itp{A_i}}(a) \acomp b
        \logrelc_{C}
        (\plet \; x \pbind P \; \pin \; Q)[\sigma]
    \]
    by Lemma~\ref{lem:logrel-backward} and the definition of $\logrelc_{C}$.
\end{proof}
\end{toappendix}

The logical relations are used to prove the following adequacy theorem, yielding Corollary~\ref{cor:adequacy}.
\begin{theoremapxproof}\label{thm:adequacy}
    For $\Gamma = x_1 : A_1, \dots, x_n : A_n$ and $\Delta = y_1 : B_1, \dots, y_m : B_m$, the following hold:
    \begin{enumerate}
        \item\label{thm:adequacy:item:term} For any well-typed term $\Gamma, \Delta \vdash M : A$, $v_i \logrelt_{A_i} V_i$ ($1 \le i \le n$) and $w_j \logrelt_{B_j} W_j$ ($1 \le j \le m$), it holds that
        \[ d \semicolon \itp{M}^{\Gamma}_{\Delta}(c) \logrelt_A M[\sigma, \tau] \]
        where
        $c : 1 \to \prod_{i=1}^n \itp{A_i}$ is defined by $c = \langle v_1, \dots, v_n \rangle$,
        $d : 1 \to \prod_{j=1}^m \itp{B_j}$ is defined by $d = \langle w_1, \dots, w_m \rangle$,
        $\sigma = V_1/x_1,
        \dots, V_n/x_n$ and $\tau = W_1/y_1, \dots, W_m/y_m$.
        \item\label{thm:adequacy:item:command} For any well-typed command $\Gamma \asemicolon \Delta \vdash P : C$ and $v_i \logrelt_{A_i} V_i$ ($1 \le i \le n$) and $w_j \logrelt_{B_j} W_j$ ($1 \le j \le m$), it holds that
        \[ \apure(d) \acomp \itp{P}(c) \logrelc_C P[\sigma, \tau]. \]
        $c : 1 \to \prod_{i=1}^n \itp{A_i}$ is defined by $c = \langle v_1, \dots, v_n \rangle$,
        $d : 1 \to \prod_{j=1}^m \itp{B_j}$ is defined by $d = \langle w_1, \dots, w_m \rangle$,
        $\sigma = V_1/x_1,
        \dots, V_n/x_n$ and $\tau = W_1/y_1, \dots, W_m/y_m$.
    \end{enumerate}
\end{theoremapxproof}
\begin{appendixproof}[Proof of Theorem~\ref{thm:adequacy}]
    Let $\sigma = V_1/x_1, \dots, V_n/x_n$ and $\tau = W_1/y_1, \dots, W_m/y_m$.
    
    \itemref{thm:adequacy:item:term}
    The proof is done by induction on the derivation of $\Gamma \vdash M : A$.
    
    \itemref{thm:adequacy:item:command}
    The proof is done by induction on $(\rhdepth{P}, \dsize{\Gamma \asemicolon \Delta \vdash P : C}) \in (\Nat^2, \le_{\Nat^2})$.

    \textbf{Base case}: $(\rhdepth{P}, \dsize{\Gamma \asemicolon \Delta \vdash P : C}) = (0, 1)$.
    There are three cases: $P = \pret{x_i}$, $P = \pret{y_j}$ and $P = \pret{\const}$ for some $\const : \beta \in \signatfunc$.
    
    \textbf{Case} $P = \pret{x_i}$.
    We have $P[\sigma, \tau] = \pret{V_i}$ and
    \[
        \apure(d) \acomp \itp{\pret{x_i}}(c)
        = \apure(d) \acomp \apure(\terminal_{\itp{\Delta}} \semicolon c \semicolon \pi_i)
        = \apure(v_i).
    \]
    Hence by definition of $\logrelc_C$, we have $\apure(v_i) \logrelc_C \pret{V_i}$ as desired.

    \textbf{Case} $P = \pret{y_j}$.
    We have $P[\sigma, \tau] = \pret{W_j}$ and
    \[
        \apure(d) \acomp \itp{\pret{y_j}}(c)
        = \apure(d) \acomp \pi_j
        = \apure(w_j).
    \]
    Hence by definition of $\logrelc_C$, we have $\apure(w_j) \logrelc_C \pret{W_j}$ as desired.

    \textbf{Case} $P = \pret{\const}$ for some $\const : \beta \in \signatfunc$.
    We have $P[\sigma, \tau] = \pret{\const}$ and 
    \[
        \apure(d) \acomp \itp{\pret{\const}}(c)
        = \apure(d) \acomp \apure(\terminal_{\itp{\Delta}} \semicolon \itp{\const})
        = \apure(\itp{\const}).
    \]
    Hence by definition of $\logrelc_C$, we have $\apure(\itp{\const}) \logrelc_C \pret{\const}$ as desired.

    \textbf{Inductive step}: $(\rhdepth{P}, \dsize{\Gamma \asemicolon \Delta \vdash P : C}) > (0, 1)$.
    We proceed by case analysis on the last rule used in the derivation of $\Gamma \asemicolon \Delta \vdash P : C$.

    \textbf{Case} \tycpure.
    The root of the derivation is
    \[
        \infer[\tycpure]{
            \Gamma \asemicolon \Delta \vdash \pret{M} : A
        }{
            \Gamma, \Delta \vdash M : A
        }
    \]
    By \itemref{thm:adequacy:item:term}, we have
    $d \semicolon \itp{M}^{\Gamma}_{\Delta}(c) \logrelt_A M[\sigma, \tau]$.
    By definition of $\logrelt_A$, we have $M[\sigma, \tau] \redto^* U$ for some value $U$ and $\itp{U} = d ; \itp{M}^{\Gamma}_{\Delta}(c)$.
    Hence, we have
    $\pret{M}[\sigma, \tau]
    = \pret{M[\sigma, \tau]}
    \redto^* \pret{U}$ and $\apure(U) \logrelc_A \pret{U}$.
    We have
    \[
        \apure(d) \acomp \itp{\pret{M}}(c)
        = \apure(d \semicolon \itp{M}^{\Gamma}_{\Delta}(c))
        = \apure(\itp{U}).
    \]
    By Lemma~\ref{lem:logrel-backward}, we have
    $\apure(d \semicolon \itp{\pret{M}}(c))
    \logrelc_A
    \pret{M}[\sigma, \tau]$.

    \textbf{Case} \tycop.
    The root of the derivation is
    \[
        \infer[\tycop]{
            \Gamma \asemicolon \Delta \vdash \opr(M) : C
        }{
            \Gamma, \Delta \vdash M : A
            &
            \opr : A \opto C \in \signatopr
        }
    \]
    By \itemref{thm:adequacy:item:term}, we have
    $d \semicolon \itp{M}^{\Gamma}_{\Delta}(c) \logrelt_A M[\sigma, \tau]$.
    By definition of $\logrelt_A$, we have $M[\sigma, \tau] \redto^* U$ for some value $U$ and $\itp{U} = d \semicolon \itp{M}^{\Gamma}_{\Delta}(c)$.
    Hence, we have
    $\opr(M)[\sigma, \tau]
    = \opr(M[\sigma, \tau])
    \redto^* \opr(U)$
    and
    \[
        \apure(\itp{U}) \acomp \opr
        =
        \apure(d \semicolon \itp{M}^{\Gamma}_{\Delta}(c)) \acomp \opr
        =
        \apure(d) \acomp \itp{\opr(M)}(c).
    \]
    Therefore, by the definition of $\logrelc_C$ and Lemma~\ref{lem:logrel-backward}, we have
    $\apure(d) \acomp \itp{\opr(M)}(c)
    \logrelc_C
    \opr(M)[\sigma, \tau]$.

    \textbf{Case} \tyclet.
    The root of the derivation is
    \[
        \infer[\tyclet]{
            \Gamma \asemicolon \Delta \vdash \plet \; y \pbind Q \; \pin \; R : C
        }{
            \Gamma \asemicolon \Delta \vdash Q : B
            &
            \Gamma \asemicolon y : B, \Delta \vdash R : C
        }
    \]
    By the induction hypothesis and Lemma~\ref{lem:logrel-let-in}, we have
    \[
        \apure(d) \acomp \itp{\plet \; y \pbind Q \; \pin \; R}(c)
        \logrelc_C
        (\plet \; y \pbind Q \; \pin \; R)[\sigma, \tau].
    \]

    \textbf{Case} \tychandle.
    The root of the derivation of $\Gamma \asemicolon \Delta \vdash P : \typprod_{k=1}^l B'_k$ is
    \[
        \infer[\tychandle]{
            \Gamma \asemicolon \Delta \vdash \prevhandle(M) \tuple{y'_1, \dots, y'_l}.R \; \pwith \; H : \typprod_{k=1}^l B'_k
        }{
            \Gamma, \Delta \vdash M : \typprod_{k=1}^l B'_k
            &
            \Gamma \asemicolon \Delta' \vdash R : C
            &
            \revhandlerjudgment H : C
        }
    \]
    where
    $\Delta' = y'_1 : B'_1, \dots, y'_l : B'_{l}$ and
    $H = \{ x \mapsto P_H \} \cup \{ (x \mapsto Q_{\opr}^{\fwd} \mid y, z \mapsto Q_{\opr}^{\bwd}) \}_{\opr \in \signatopr}$.
    By \itemref{thm:adequacy:item:term}, we have
    $d \semicolon \itp{M}^{\Gamma}_{\Delta}(c) \logrelt_{\typprod_{k=1}^l B'_k} M[\sigma, \tau]$. Hence, by definition of $\logrelt$, we have
    $M[\sigma, \tau] \redto^* \tuple{W'_1, \dots, W'_l}$
    for some values $W'_1, \dots, W'_l$
    with $\langle \itp{W'_1}, \dots, \itp{W'_l} \rangle = d \semicolon \itp{M}^{\Gamma}_{\Delta}(c)$.
    We have
    \[
    \begin{aligned}
        P[\sigma, \tau]
        & =  (\prevhandle(M) \tuple{y'_1, \dots, y'_l}.R \; \pwith \; H)[\sigma, \tau]
        \\
        & =
        \prevhandle(M[\sigma, \tau]) \tuple{y'_1, \dots, y'_l}.R[\sigma] \; \pwith \; H
        \\
        & \redto^*
        \prevhandle(\tuple{W'_1, \dots, W'_l}) \tuple{y'_1, \dots, y'_l}.R[\sigma] \; \pwith \; H
        \\
        & \redto
        \rewriterevh{\tuple{W_j}_{j=1}^l}{H}(R[\sigma]).
    \end{aligned}
    \]
    We have 
    $\rhdepth{\rewriterevh{\tuple{W_j}_{j=1}^l}{H}(R[\sigma])} < \rhdepth{P}$.
    Hence, by the induction hypothesis and Theorem~\ref{thm:soundness}, we have
    \[
        \apure(d) \acomp \itp{P}(c)
        = \itp{\rewriterevh{\tuple{W_j}_{j=1}^l}{H}(R[\sigma])}
        \logrelc_C
        \rewriterevh{\tuple{W_j}_{j=1}^l}{H}(R[\sigma]).
    \]
    By Lemma~\ref{lem:logrel-backward}, we have
    $\apure(d) \acomp \itp{P}(c)
    \logrelc_C
    P[\sigma, \tau]$.
\end{appendixproof}
\begin{corollary}[Adequacy]\label{cor:adequacy}
    For any well-typed closed command $\emptyenv \asemicolon \emptyenv \vdash P : A$,
    If $\itp{P} = \apure(v)$ for some total point $v: 1\to\itp{A}$, then $P \redto^* \pret{V}$ for some value $V$ satisfying $v \logrelt_A V$.
\end{corollary}
\begin{appendixproof}[Proof of Corollary~\ref{cor:adequacy}]
    By Theorem~\ref{thm:adequacy}, we have
    $\apure(v) = \itp{P} \logrelc_A P$.
    By the definition of $\logrelc_A$, we have $P \redto^* \pret{V}$ for some value $V$ satisfying $v \logrelt_A V$.
\end{appendixproof}

\section{Related Work}
Arrows \cite{Hughes2000} provide a wider class of models of computational effects than monads.
While monads capture the type of the result of computations with effects \cite{Moggi1991}, arrows capture both the input and output types of computations with effects.
The theoretical foundation of arrows has been studied in terms of strong promonads \cite{HeunenJacobs06,JacobsHeunenHasuo2009,Asada2010}.
The arrow calculus \cite{LindleyWadlerYallop2010} is an arrow version of the monadic metalanguage \cite{Moggi1991}.
Arrows are suitable to describe sequential computations \cite{Lindley2014,Sanada2024}.

Algebraic effects \cite{PlotkinPower2001} and effect handlers \cite{PlotkinPretnar2013} provide powerful abstractions and implementation techniques for computational effects.
There are various extensions and variations of algebraic effects and handlers.
The arrow calculus with operations and handlers \cite{Sanada2024} is an extension of the arrow calculus and provides an arrow version of algebraic effects and handlers.

The most important and interesting functionality for training of neural networks is backpropagation \cite{RumelhartHintonWilliams1986,Linnainmaa1976}, which is a special case of reverse-mode automatic differentiation (AD), which is a family of algorithms to compute derivatives of functions (e.g.\ \cite{Baydin+2017,PearlmutterSiskind2008}).
CRDCs \cite{Cockett+2020} and RDRCs \cite{Cruttwell+2021} provide a categorical foundation of reverse-mode automatic differentiation.

There are some languages, libraries and implementations for AD. 
%\begin{itemize}
    %\item
    PyTorch \cite{Paszke+2019} and TensorFlow \cite{Abadi+2016} are popular libraries in Python for deep learning that support AD via computational graphs.
    Programmers write neural networks imperatively using these libraries.
    These libraries do not have safe type systems nor formal semantics.
    In contrast, our language has a safe type system and rigorous categorical semantics.

    %\item
    A simple differentiable programming language by Abadi and Plotkin \cite{AbadiPlotkin2019} is a pure functional language with reverse-mode AD.
    The language has formal operational semantics and denotational semantics based on RDRCs \cite{Cruttwell+2021}.
    We added effectful computation and handlers to a simplified version of their language.

    %\item
    Sigal \cite{Sigal2021,Sigal2024} described implementation techniques for forward- and reverse-mode automatic differentiation using effect handlers.
    His reverse-mode AD technique is based on monads and mutable state.
    In his approach, handlers compute derivatives of pure functions, and mutable state is essential for the computation.
    In contrast, our approach automatically computes derivatives of pure functions while allowing programmers to customize derivatives of algebraic operations (e.g.\ STE).
    In our MLP example, mutable state is used only to store parameter matrices.

    %\item
    CHAD \cite{VakarSmeding2022} is a categorical framework for forward- and reverse-mode AD.
    The framework provides a transformation from terms in a source language to terms in a target language that computes derivatives.
    In our language, derivatives are computed within one language.
%\end{itemize}

Our framework provides a practical language for design and implementation of neural networks with formal semantics.
Compared with the existing approaches for AD via functional programming languages (e.g.\ \cite{AbadiPlotkin2019,Sigal2021,VakarSmeding2022}),
the main advantage of our framework is separation between design and behavior of neural networks.
This separation enables us to associate layers with non-standard derivatives such as STE.

There are several approaches to model neural networks using category theory:
for example, lens-based approaches \cite{FongJohnson2019,FongSpivakTuyeras2021,Cruttwell+2022}, string diagrammatic representations of neural networks \cite{XuMaruyama2021}, and CRDCs and RDRCs \cite{Cockett+2020,Cruttwell+2021,Cruttwell+2022}.
Our approach is closely related to these: if $\promonad$ is the identity promonad on $\catx$, the presheaf $\ralg{\promonad}$ of Definition~\ref{def:reverse-algebra} is exactly the one on $\mathrm{Lens}(\catx)$ represented by the terminal object, restricted along the lens construction $\catx \to \mathrm{Lens}(\catx)$ \cite[Prop.~2.7]{Cruttwell+2022}.
We expect that a suitable generalization of the lens construction to strong promonads on RDRCs will explain Definition~\ref{def:reverse-algebra} and Lemma~\ref{lem:ralg-A-is-functor} for general $\promonad$.

What distinguishes our framework is not merely a separation between abstract architecture and learning behavior, which already appears to some extent in existing categorical approaches, but the fact that the implementation mechanism itself is described by promonad algebras and reverse handlers.
In \cite{FongSpivakTuyeras2021}, compositional architectures are described in a category such as $\mathbf{NNet}$, while backpropagation is supplied by external functorial constructions.
In \cite{Cruttwell+2022}, this implementation story is brought further inside the categorical framework through CRDC structure and parametric lenses, but it is still organized as semantic machinery assembled from given components.
By contrast, our framework makes the implementation side itself programmable within the language.
More broadly, modelling neural networks by promonads on RDRCs connects a categorical account of neural-network design and implementation with a programming-language account, based on the effect-handler paradigm, together with formal operational and denotational semantics.

\section{Conclusion and Future Work}
We introduced a differentiable arrow calculus with operations and reverse handlers.
We defined its operational and denotational semantics, and proved the soundness (Theorem~\ref{thm:soundness}) and adequacy (Corollary~\ref{cor:adequacy}) of the denotational semantics with respect to the operational semantics.
There are many examples of neural network architectures that can be expressed naturally in our language (Table~\ref{tab:examples-summary}).
%including
%MLPs (Example~\ref{ex:mlp-backpropagation}, \ref{ex:mlp-reduction}, and \ref{ex:mlp-denotation}), 
%residual neural networks (Example~\ref{ex:residual-neural-network} and \ref{ex:resnet-denotation}),
%autoencoders (Example~\ref{ex:autoencoder}),
%CNNs (Example~\ref{ex:convolutional-neural-network} and \ref{ex:convolutional-neural-network-H-CNN}),
%and U-nets (Example~\ref{ex:U-net} and \ref{ex:U-net-denotation}).
%QAT via STE can also be implemented using reverse handlers (Example~\ref{ex:ste}).
The string diagrams of their denotations clearly illustrate their architectures.
Hence, our framework provides not only a programming language for neural networks, but also a mathematically rigorous graphical language for describing neural network architectures.

We used reverse handlers to implement pairs of an encoder and a decoder in Example~\ref{ex:autoencoder} and \ref{ex:U-net}.
This is not because reverse handlers have unnecessarily high expressive power, but because there is a fundamental relation between CRDCs and lenses \cite{Cruttwell+2022}, which are a model of bidirectional computation.
The precise description of this relation in our setting is left as future work.

As a possible extension toward greater practicality, we plan to enrich our language with additional features.
One challenge is to incorporate delayed traces \cite{SprungerKatsumata2021} to support recurrent neural networks.

\begin{acks}
    The first author was supported by JSPS KAKENHI Grant Number JP24K23867.
    The second author was supported by JSPS KAKENHI Grant Number JP23KJ1365.
    The third author is supported by JST SPRING, Grant Number JPMJSP2110.
\end{acks}

\begin{toappendix}
\section{Details of Examples}
\label{sec:details-examples}

In this section, we provide details of the examples that appeared in \S\ref{sec:examples-neural-network} and
\ref{subsec:examples-operational-semantics}.
%In particular, in the latter subsection we give a more detailed explanation of
%the examples involving non-smooth functions, where we need either to extend our
%framework or to rely on smooth approximation techniques.

\begin{figure}
    \centering
    \scalebox{0.7}{$
    \begin{aligned}
        & ([\ell_0 \mapsto \fromvec{m_0}, \ell_1 \mapsto \fromvec{m_1}], \prevhandle(\fromvec{v})\tuple{x_{\inp}}. R_{\MLP} \; \pwith \; H)
        \\
        & \redto
        ([\ell_0 \mapsto \fromvec{m_0}, \ell_1 \mapsto \fromvec{m_1}], \rewriterevh{\fromvec{v}}{H} (x_{\inp}.R_{\MLP}))
        \\
        & =
        \left(
            \left[
                \begin{aligned}
                    \ell_0 & \mapsto \fromvec{m_0},
                    \\
                    \ell_1 & \mapsto \fromvec{m_1}
                \end{aligned}
            \right],
            \left(
                \begin{aligned}
                    & \plet \; \tuple{y, z'} \pbind Q^{\fwd}_{\opFC{\ell_0}}[\fromvec{v}/x] \; \pin
                    \\
                    & \plet \; \tuple{y', x'} \pbind \rewriterevh{\tuple{y, \fromvec{v}}}{H} \left(
                        \tuple{y, x_{\inp}}.
                        \begin{aligned}
                            & \plet \; z_{\hid} \pbind \pret{y} \; \pin \;
                            \plet \; x_{\hid} \pbind \pret{\fnSwish_{n_{\hid}}(z_{\hid})} \; \pin \\
                            & \plet \; z_{\out} \pbind \opFC{\ell_1}(x_{\hid}) \; \pin \; \pret{z_{\out}}
                        \end{aligned}
                    \right) \; \pin
                    \\
                    & \plet \; x'' \pbind Q^{\bwd}_{\opFC{\ell_0}}[y'/y, z'/z] \; \pin \; \pret{x' + x''}
                \end{aligned}
            \right)
        \right)
        \\
        & \redto^*
        \left(
            \left[
                \begin{aligned}
                    \ell_0 & \mapsto \fromvec{m_0},
                    \\
                    \ell_1 & \mapsto \fromvec{m_1}
                \end{aligned}
            \right],
            \left(
                \begin{aligned}
                    & \plet \; \tuple{y, z'} \pbind \tuple{\fromvec{m_0 v}, \tuple{\fromvec{m_0}, \fromvec{v}}} \; \pin
                    \\
                    & \plet \; \tuple{y', x'} \pbind \rewriterevh{\tuple{y, \fromvec{v}}}{H} \left(
                        \tuple{y, x_{\inp}}.
                        \begin{aligned}
                            & \plet \; z_{\hid} \pbind \pret{y} \; \pin \;
                            \plet \; x_{\hid} \pbind \pret{\fnSwish_{n_{\hid}}(z_{\hid})} \; \pin \\
                            & \plet \; z_{\out} \pbind \opFC{\ell_1}(x_{\hid}) \; \pin \; \pret{z_{\out}}
                        \end{aligned}
                    \right) \; \pin
                    \\
                    & \plet \; x'' \pbind Q^{\bwd}_{\opFC{\ell_0}}[y'/y, z'/z] \; \pin \; \pret{x' + x''}
                \end{aligned}
            \right)
        \right)
        \\
        & \redto^*
        \left(
            \left[
                \begin{aligned}
                    \ell_0 & \mapsto \fromvec{m_0},
                    \\
                    \ell_1 & \mapsto \fromvec{m_1}
                \end{aligned}
            \right],
            \left(
                \begin{aligned}
                    & \plet \; \tuple{y', x'} \pbind \rewriterevh{\tuple{\fromvec{m_0 v}, \fromvec{v}}}{H} \left(
                        \tuple{y, x_{\inp}}.
                        \plet \; x_{\hid} \pbind \pret{\fnSwish_{n_{\hid}}(y)} \; \pin \;
                        \plet \; z_{\out} \pbind \opFC{\ell_1}(x_{\hid}) \; \pin \; \pret{z_{\out}}
                    \right) \; \pin
                    \\
                    & \plet \; x'' \pbind Q^{\bwd}_{\opFC{\ell_0}}[y'/y, \tuple{\fromvec{m_0}, \fromvec{v}}/z] \; \pin \; \pret{x' + x''}
            \end{aligned}
            \right)
        \right)
        \\
        & \redto
        \left(
            \left[
                \begin{aligned}
                    \ell_0 & \mapsto \fromvec{m_0},
                    \\
                    \ell_1 & \mapsto \fromvec{m_1}
                \end{aligned}
            \right],
            \left(
                \begin{aligned}
                    & \plet \; \tuple{y', x'} \pbind \\
                    & \quad
                        \begin{aligned}
                            & \plet \; z \pbind \pret{\fnSwish(\fromvec{m_0 v})} \; \pin \\
                            & \plet \; (z', y', x') \pbind \rewriterevh{\tuple{z, \fromvec{m_0 v}, \fromvec{v}}}{H}\left(\tuple{z, y, x_{\inp}}. \plet \; x_{\hid} \pbind \pret{z} \; \pin \; \plet \; z_{\out} \pbind \opFC{\ell_1}(x_{\hid}) \; \pin \; \pret{z_{\out}}\right) \; \pin \\
                            & \pret{\langle y' + z'.\trd(y.\fnSwish(y))(\fromvec{m_0 v}), x' + z'.\trd(x_{\inp}. \fnSwish(\fromvec{m_0 v}))(\fromvec{v}) \rangle}
                        \end{aligned}
                    \\
                    & \pin \; \plet \; x'' \pbind Q^{\bwd}_{\opFC{\ell_0}}[y'/y, \tuple{\fromvec{m_0}, \fromvec{v}}/z] \; \pin \; \pret{x' + x''}
            \end{aligned}
            \right)
        \right)
        \\
        & \redto^*
        \left(
            \left[
                \begin{aligned}
                    \ell_0 & \mapsto \fromvec{m_0},
                    \\
                    \ell_1 & \mapsto \fromvec{m_1}
                \end{aligned}
            \right],
            \left(
                \begin{aligned}
                    & \plet \; \tuple{y'_0, x'_0} \pbind \\
                    & \quad
                        \begin{aligned}
                            & \plet \; \tuple{z', y', x'} \pbind
                            \rewriterevh{\tuple{\fromvec{v_1}, \fromvec{m_0 v}, \fromvec{v}}}{H}\left(\tuple{z, y, x_{\inp}}.\plet \; z_{\out} \pbind \opFC{\ell_1}(z) \; \pin \; \pret{z_{\out}}
                            \right) \\
                            & \pin \; \pret{\langle y' + z'.\trd(y.\fnSwish(y))(\fromvec{m_0 v}), x' + z'.\trd(x_{\inp}. \fnSwish(\fromvec{m_0 v}))(\fromvec{v}) \rangle}
                        \end{aligned}
                    \\
                    & \pin \; \plet \; x'' \pbind Q^{\bwd}_{\opFC{\ell_0}}[y'_0/y, \tuple{\fromvec{m_0}, \fromvec{v}}/z] \; \pin \;
                    \pret{x'_0 + x''}
            \end{aligned}
            \right)
        \right)
        \\
        & \quad \text{where $\fromvec{v_1} = \Eval(\fnSwish, \fromvec{m_0 v})$}
        \\
        & \redto
        \left(
            \left[
                \begin{aligned}
                    \ell_0 & \mapsto \fromvec{m_0},
                    \\
                    \ell_1 & \mapsto \fromvec{m_1}
                \end{aligned}
            \right],
            \left(
                \begin{aligned}
                    & \plet \; \tuple{y'_0, x'_0} \pbind \\
                    & \quad
                        \begin{aligned}
                            & \plet \; \tuple{z', y', x'} \pbind \\
                            & \quad \begin{aligned}
                                & \plet \; \tuple{y', x'} \pbind Q^{\fwd}_{\opFC{\ell_1}}[\fromvec{v_1}/x] \; \pin
                                \\
                                & \plet \; \tuple{y'_1, z', y'_0, x'_0} \pbind \rewriterevh{\tuple{y', \fromvec{v_1}, \fromvec{m_0 v}, \fromvec{v}}}{H}(\tuple{y'_1, z', y'_0, x'_0}. \plet \; z_{\out} \pbind \pret{y'_1} \; \pin \; \pret{z_{\out}}) \; \pin
                                \\
                                & \plet \; z'' \pbind Q^{\bwd}_{\opFC{\ell_1}}[y'_1/y', x'/z] \; \pin \;
                                \pret{\langle z' + z'', y'_0, x'_0 \rangle}
                            \end{aligned}\\
                            & \pin \; \pret{\langle y' + z'.\trd(y.\fnSwish(y))(\fromvec{m_0 v}), x' + z'.\trd(x_{\inp}. \fnSwish(\fromvec{m_0 v}))(\fromvec{v}) \rangle}
                        \end{aligned}
                    \\
                    & \pin \; \plet \; x'' \pbind Q^{\bwd}_{\opFC{\ell_0}}[y'_0/y, \tuple{\fromvec{m_0}, \fromvec{v}}/z] \; \pin \;
                    \pret{x'_0 + x''}
            \end{aligned}
            \right)
        \right)
        \\
        & \redto^*
        \left(
            \left[
                \begin{aligned}
                    \ell_0 & \mapsto \fromvec{m_0},
                    \\
                    \ell_1 & \mapsto \fromvec{m_1}
                \end{aligned}
            \right],
            \left(
                \begin{aligned}
                    & \plet \; \tuple{y'_0, x'_0} \pbind \\
                    & \quad
                        \begin{aligned}
                            & \plet \; \tuple{z', y', x'} \pbind \\
                            & \quad \begin{aligned}
                                & \plet \; \tuple{y'_1, z', y'_0, x'_0} \pbind \rewriterevh{\tuple{\fromvec{m_1 v_1}, \fromvec{v_1}, \fromvec{m_0 v}, \fromvec{v}}}{H}(\tuple{y'_1, z', y'_0, x'_0}.\pret{y'_1}) \; \pin
                                \\
                                & \plet \; z'' \pbind Q^{\bwd}_{\opFC{\ell_1}}[y'_1/y', \tuple{\fromvec{m_1}, \fromvec{v_1}}/z] \; \pin \;
                                \pret{\langle z' + z'', y'_0, x'_0 \rangle}
                            \end{aligned}\\
                            & \pin \; \pret{\langle y' + z'.\trd(y.\fnSwish(y))(\fromvec{m_0 v}), x' + z'.\trd(x_{\inp}. \fnSwish(\fromvec{m_0 v}))(\fromvec{v}) \rangle}
                        \end{aligned}
                    \\
                    & \pin \; \plet \; x'' \pbind Q^{\bwd}_{\opFC{\ell_0}}[y'_0/y, \tuple{\fromvec{m_0}, \fromvec{v}}/z] \; \pin \;
                    \pret{x'_0 + x''}
            \end{aligned}
            \right)
        \right)
        \\
        & \redto^*
        \left(
            \left[
                \begin{aligned}
                    \ell_0 & \mapsto \fromvec{m_0},
                    \\
                    \ell_1 & \mapsto \fromvec{m_1}
                \end{aligned}
            \right],
            \left(
                \begin{aligned}
                    & \plet \; \tuple{y'_0, x'_0} \pbind \\
                    & \quad
                        \begin{aligned}
                            & \plet \; \tuple{z', y', x'} \pbind
                            \plet \; z'' \pbind Q^{\bwd}_{\opFC{\ell_1}}[\fromvec{\learningrate(m_1 v_1 - t)}/y', \tuple{\fromvec{m_1}, \fromvec{v_1}}/z] \; \pin \;
                            \pret{\langle \fromvec{0} + z'', \fromvec{0}, \fromvec{0} \rangle}
                            \\
                            & \pin \; \pret{\langle y' + z'.\trd(y.\fnSwish(y))(\fromvec{m_0 v}), x' + z'.\trd(x_{\inp}. \fnSwish(\fromvec{m_0 v}))(\fromvec{v}) \rangle}
                        \end{aligned}
                    \\
                    & \pin \; \plet \; x'' \pbind Q^{\bwd}_{\opFC{\ell_0}}[y'_0/y, \tuple{\fromvec{m_0}, \fromvec{v}}/z] \; \pin \;
                    \pret{x'_0 + x''}
            \end{aligned}
            \right)
        \right)
        \\
        & \redto^*
        \left(
            \left[
                \begin{aligned}
                    \ell_0 & \mapsto \fromvec{m_0},
                    \\
                    \ell_1 & \mapsto \fromvec{m_1 - \learningrate(m_1 v_1 - t) (v_1^{\transpose}) }
                \end{aligned}
            \right],
            \left(
                \begin{aligned}
                    & \plet \; \tuple{y'_0, x'_0} \pbind \pret{\langle \fromvec{0} + \fromvec{u_1}.\trd(y.\fnSwish(y))(\fromvec{m_0 v}), \fromvec{0} + \fromvec{u_1}.\trd(x_{\inp}. \fnSwish(\fromvec{m_0 v}))(\fromvec{v}) \rangle} \; \pin \\
                    & \plet \; x'' \pbind Q^{\bwd}_{\opFC{\ell_0}}[y'_0/y, \tuple{\fromvec{m_0}, \fromvec{v}}/z] \; \pin \;
                    \pret{x'_0 + x''}
            \end{aligned}
            \right)
        \right)
        \\
        & \quad \text{where $u_1 = \learningrate m^{\transpose}_1 (m_1 v_1 - t)$}
        \\
        & \redto^*
        \left(
            \left[
                \begin{aligned}
                    \ell_0 & \mapsto \fromvec{m_0},
                    \\
                    \ell_1 & \mapsto \fromvec{m_1 - \learningrate(m_1 v_1 - t) (v_1^{\transpose}) }
                \end{aligned}
            \right],
            \plet \; x'' \pbind Q^{\bwd}_{\opFC{\ell_0}}[\fromvec{u'_1}/y, \tuple{\fromvec{m_0}, \fromvec{v}}/z] \; \pin \; \pret{\fromvec{0} + x''}
        \right)
        \\
        & \quad \text{where $\fromvec{u'_1} = \Eval(\rd[\fnSwish], \tuple{\fromvec{u_1}, \fromvec{m_0 v}})$}
        \\
        & \redto^*
        \left(
            \left[
                \ell_0 \mapsto \fromvec{m_0 - u'_1 (v^{\transpose}) },
                \ell_1 \mapsto \fromvec{m_1 - \learningrate(m_1 v_1 - t) (v_1^{\transpose}) }
            \right],
            \pret{\fromvec{m_0^{\transpose} u'_1 }}
        \right)
    \end{aligned}
    $}
    \caption{Reduction sequence of $\prevhandle(\fromvec{v})\tuple{x_{\inp}}. R_{\MLP} \; \pwith \; H_{\MLP}$}\label{fig:mlp-reduction}
    \Description{Reduction sequence of backpropagation of the multilayer perceptron $R_{\MLP}$.}
\end{figure}

\begin{example}[gradient in MLP, continued from Example~\ref{ex:mlp-reduction}]
\label{ex:mlp-gradient-calc}
    For $m_0 \in \Real^{n_{\inp} n_{\hid}}$ and $m_1 \in \Real^{n_{\hid} n_{\out}}$, the entire reduction sequence of the following reduction shown in Fig.~\ref{fig:mlp-reduction}.
    \[
    \begin{aligned}
        & ([\ell_0 \mapsto m_0, \ell_1 \mapsto m_1], \prevhandle(\fromvec{v})\tuple{x_{\inp}}. R_{\MLP} \; \pwith \; H_{\MLP}) \\
        & \redto^*
        \left(
            \left[ \ell_0 \mapsto \fromvec{m_0 - u'_1 (v^{\transpose}) }, \ell_1 \mapsto \fromvec{m_1 - \learningrate(m_1 v_1 - t) (v_1^{\transpose}) } \right],
            \pret{V}
        \right)
        \quad\text{where $V$ is a value.}
    \end{aligned}
    \]
    
    We show that the resulting heap
    $[\ell_0 \mapsto \fromvec{m_0 - u'_1(v^{\transpose})}, \ell_1 \mapsto \fromvec{m_1 - \learningrate (m_1 v_1 - t)(v_1^{\transpose})}]$
    after the reduction sequence in Fig.~\ref{fig:mlp-reduction} coincides with the matrices updated by usual gradient descent of the MLP $R_{\MLP}$ with the learning rate $\learningrate$.
    We write $\itp{\fnSwish}$ for the Swish activation function defined by $\itp{\fnSwish}(x) = x \cdot \sigmoid(x)$ and assume that the following equations hold:
    \[
        \Eval(\fnSwish_n, \fromvec{x}) = \fromvec{\itp{\fnSwish}(x_k)_{k=1}^{n}},
        \quad
        \Eval(\rd[\fnSwish_n], \tuple{\fromvec{x}, \fromvec{y}}) = \fromvec{\left(x_k \frac{\partial}{\partial y_k}\itp{\fnSwish}(y_k)\right)_{k=1}^{n}}.
    \]
    Recall that we defined $u'_1$ and $v_1$ by
    \[
    \fromvec{u'_1} = \Eval(\rd[\fnSwish], \tuple{\fromvec{\learningrate m_1^{\transpose}(m_1 v_1 - t)}, \fromvec{m_0 v}}),
    \quad
    \fromvec{v_1} = \Eval(\fnSwish, \fromvec{m_0 v}) = \fromvec{( \itp{\fnSwish}((m_0 v)_k) )_{k=1}^{n_{\hid}}}.
    \]
    The loss $\Loss(v ; m_0, m_1)$ is the squared error defined by
    \[ \Loss(v ; m_0, m_1) = \frac{1}{2} \sum_{k=1}^{n_{\out}} ((m_1 \itp{\fnSwish}(m_0 v))_k - t_k)^2. \]

    First, we calculate the gradient of the loss with respect to $m_{1ij}$:
    \[
    \begin{aligned}
        \frac{\partial \Loss}{\partial m_{1ij}} (v ; m_0, m_1)
        & =
        \sum_{k=1}^{n_{\out}} ((m_1 \itp{\fnSwish}(m_0 v))_k - t_k) \frac{\partial}{\partial m_{1ij}} (m_1 \itp{\fnSwish}(m_0 v))_k
        \\
        & =
        \sum_{k=1}^{n_{\out}} ((m_1 v_1)_k - t_k) \frac{\partial}{\partial m_{1ij}} \sum_{j'} m_{1 k j'} v_{1j'}
        \\
        & =
        ((m_1 v_1)_i - t_i) v_{1j}.
    \end{aligned}
    \]
    Hence, we have $\frac{\partial \Loss}{\partial m_1} (v ; m_0, m_1) = ((m_1 v_1) - t) (v_1^{\transpose})$ and the update of $m_1$ by gradient descent with the learning rate $\learningrate$ is given by
    \[
        m_1 - \learningrate \frac{\partial \Loss}{\partial m_1} (v ; m_0, m_1)
        =
        m_1 - \learningrate ( (m_1 v_1) - t ) (v_1^{\transpose}).
    \]

    Next, we calculate the gradient of the loss with respect to $m_{0ij}$:
    \[
    \begin{aligned}
        \frac{\partial \Loss}{\partial m_{0ij}} (v ; m_0, m_1)
        & =
        \sum_{l}
        \frac{\partial \Loss}{\partial v_{1l}} (v ; m_0, m_1)
        \frac{\partial v_{1l}}{\partial m_{0ij}} (v ; m_0)
        \\
        & =
        \sum_{l}
        \left(
            \sum_{k}
            ((m_1 v_1)_k - t_k) \frac{\partial}{\partial v_{1l}} (m_1 v_1)_k
        \right)
        \frac{\partial v_{1l}}{\partial m_{0ij}} (v ; m_0)
        \\
        & =
        \sum_{l}
        \left(
            \sum_{k}
            ((m_1 v_1)_k - t_k) \frac{\partial}{\partial v_{1l}} \sum_{j'} m_{1 k j'} v_{1 j'}
        \right)
        \frac{\partial v_{1l}}{\partial m_{0ij}} (v ; m_0)
        \\
        & =
        \sum_{l}
        \left(
            \sum_{k}
            m_{1kl}
            ((m_1 v_1)_k - t_k)
        \right)
        \frac{\partial v_{1l}}{\partial m_{0ij}} (v ; m_0)
        \\
        & =
        \sum_{l}
        (m_{1}^{\transpose} (m_1 v_1 - t))_{l}
        \frac{\partial v_{1l}}{\partial m_{0ij}} (v ; m_0)
        \\
        & =
        \sum_{l}
        (m_{1}^{\transpose} (m_1 v_1 - t))_{l}
        \frac{\partial}{\partial m_{0ij}} \itp{\fnSwish}((m_0 v)_l)
        \\
        & =
        \sum_{l}
        (m_{1}^{\transpose} (m_1 v_1 - t))_{l}
        \frac{\dif \itp{\fnSwish}}{\dif y}((m_0 v)_l) \frac{\partial}{\partial m_{0ij}} (m_0 v)_l
        \\
        & =
        \sum_{l}
        (m_{1}^{\transpose} (m_1 v_1 - t))_{l}
        \frac{\dif \itp{\fnSwish}}{\dif y}((m_0 v)_l)
        \left(\sum_{j'} \frac{\partial}{\partial m_{0ij}} m_{0lj'} v_{j'}\right)
        \\
        & =
        \sum_{l}
        (m_{1}^{\transpose} (m_1 v_1 - t))_{l}
        \frac{\dif \itp{\fnSwish}}{\dif y}((m_0 v)_l)
        \left(\sum_{j'} \frac{\partial}{\partial m_{0ij}} m_{0lj'} v_{j'}\right)
        \\
        & =
        (m_{1}^{\transpose} (m_1 v_1 - t))_{i}
        \frac{\dif \itp{\fnSwish}}{\dif y}((m_0 v)_i) v_{j}.
    \end{aligned}
    \]
    Hence, we have
    \[
        \frac{\partial \Loss}{\partial m_0} (v ; m_0, m_1)
        =
        \left( (m_1^{\transpose}(m_1 v_1 - t))_i \frac{\dif \itp{\fnSwish}}{\dif y}((m_0 v)_i)\right)_{i=1}^{n_\hid} (v^{\transpose})
    \]
    and the update of $m_0$ by gradient descent with the learning rate $\learningrate$ is given by
    \[
    \begin{aligned}
        m_0 - \learningrate \frac{\partial \Loss}{\partial m_0} (v ; m_0, m_1)
        & =
        m_0 - \left( \learningrate (m_1^{\transpose}(m_1 v_1 - t))_i \frac{\dif \itp{\fnSwish}}{\dif y}((m_0 v)_i)\right)_{i=1}^{n_\hid} (v^{\transpose})
        \\
        & =
        m_0 - u'_1 (v^{\transpose}).
    \end{aligned}
    \]
\end{example}

\begin{definition}[interpretations of function symbols]
    Let $\catx = \Smooth$, $\BType = \{ \typreal(n) \mid n \in \Nat \}$.
    Let $\itp{\blank}_{\BType} \colon \BType \to \obj(\catx)$ be an interpretation defined by $\itp{\typreal(n)}_{\BType} = \Real^n$.
    We define interpretations of the function symbols used in the examples as follows.
    \[
    \begin{aligned}
        & \itp{\fnSwish_n} \colon \Real^n \to \Real^n
        &
        & \itp{\fnSwish_n} (x) = (x_i \sigmoid(x_i))_{i=1}^n
        \\
        & \itp{\fnscalarmul_n} \colon \Real \times \Real^n \to \Real^n
        &
        & \itp{\fnscalarmul_n}(a, x) = ax = (a x_i)_{i=1}^n
        \\
        & \itp{\fnminus_n} \colon \Real^n \times \Real^n \to \Real^n
        &
        & \itp{\fnminus_n}(x,y) = x - y = (x_i - y_i)_{i=1}^n
        \\
        & \itp{\fnmatmul_{n,m}} \colon \Real^{nm} \times \Real^n \to \Real^m
        &
        & \itp{\fnmatmul_{n,m}}(m, x) = mx = \left(\sum_{j=1}^n m_{i,j}x_j\right)_{i=1}^m
        \\
        & \itp{(\blank)^{\fntranspose}_{n,m}} \colon \Real^{nm} \to \Real^{mn}
        &
        & \itp{(\blank)^{\fntranspose}_{n,m}}(m) = m^{\transpose} = (m_{i,j})_{j,i}
        \\
        & \itp{\fromvec(v)} \colon \Real^0 \to \Real^n
        &
        & \itp{\fromvec{v}}(0) = v \quad \text{where $v \in \Real^n$}
    \end{aligned}
    \]
\end{definition}

%\subsection{Interpretation of CNNs and STE: problems with non-smooth functions}
%\label{subsec:RDRC}
%
%In this subsection, we discuss the interpretation of the CNN examples (Example~\ref{ex:convolutional-neural-network} and Example~\ref{ex:U-net}) and the STE example (Example~\ref{ex:ste}).
%The latter uses the symbol $\fnround$; on the other hand, as we explain in the next example, the handler $H_{\CNN}$ involves a function symbol $\fnpool$.
%These function symbols are (at least \textit{morally}) interpreted as non-smooth functions, and hence cannot be accommodated in our current framework using a CRDC such as $\Smooth$.
%
%The aim of this subsection is to explain why this difficulty occurs and to outline some possible ways to address it.
%To make the discussion more precise, we begin by giving the definition of the handler $H_{\CNN}$.
In the remainder of this section, we discuss the interpretation of the CNN examples (Example~\ref{ex:convolutional-neural-network} and Example~\ref{ex:U-net}).
We begin by giving the precise definition of the handler $H_{\CNN}$.
\begin{example}[a handler for backpropagation of CNN, continued from Example~\ref{ex:convolutional-neural-network}]
    \label{ex:convolutional-neural-network-H-CNN}
    In the situation of Example~\ref{ex:convolutional-neural-network},
    we add the following function symbols to the signature $\signatfunc$:
    \begin{align*}
        & \fnconv_{n,m,c,c'} : \typreal(cn) \times \typreal(cc'm) \to \typreal(c'(n - m + 1)),
        \\
        & \rd[\fnconv_{n,m,c,c'}] : \typreal(c'(n - m + 1)) \times (\typreal(cn) \times \typreal(cc'm)) \to \typreal(cn) \times \typreal(cc'm),
        \\
        & \fnpool_{n,m,c} : \typreal(cn) \to \typreal(c \lceil n/m \rceil),
        \\
        & \rd[\fnpool_{n,m,c}] : \typreal(c \lceil n/m \rceil) \times \typreal(cn) \to \typreal(cn).
    \end{align*}
    Intuitively, $\fnconv_{n,m,c,c'}(\tuple{\fromvec{x}, \fromvec{w}})$ is the one-dimensional convolution of a one-dimensional input $x$ with $c$ channels and a set of $c'$ filters $w$ each of which has $c$ channels and size $m$:
    \[
    \itp{\fnconv_{n,m,c,c'}}(x, w)
    =
    \left(\sum_{k=1}^{c} \sum_{j=1}^{m}
    x_{k, i + j} w_{k, l, j}\right)_{l \in [c'], i \in [n - m + 1]}.
    \]
    The reverse differentiation $\rd[\fnconv_{n,m,c,c'}]$ computes the gradients of the loss with respect to the input and the filters from the gradient with respect to the output:
    \[
    \itp{\rd[\fnconv_{n,m,c,c'}]}(u, (x, w))
    =
    \left(
        \left(
            \sum_{l=1}^{c'} \sum_{j=1}^{m} u_{l, i - j} w_{k, l, j}
        \right)_{k \in [c], i \in [n]},
        \left(
            \sum_{i=1}^{n - m + 1} u_{l, i} x_{k, i + j}
        \right)_{k \in [c], l \in [c'], j \in [m]}
    \right).
    \]
    The intuition behind the pooling function $\fnpool$ is also the usual max pooling (see Definition~\ref{def:interpretation-cnn-fnconv-fnpool}).

    We define the following commands for convolution and pooling operations:
    \[
    \begin{aligned}
        Q^{\fwd}_{\opConv{n, \loctyp{\ell}{(m, c, c')}}}
        & =
        \plet \; w \pbind \opGet{\ell}() \; \pin \; \pret{\tuple{\fnconv_{n,m,c,c'}(x, w), \tuple{w, x}}}
        \\
        Q^{\bwd}_{\opConv{n, \loctyp{\ell}{(m, c, c')}}}
        & =
        \begin{aligned}
            & \plet \; \tuple{w, x} \pbind z \; \pin \;
            \plet \; \tuple{x', w'} \pbind \rd[\fnconv_{n,m,c,c'}](y, \tuple{x, w}) \; \pin
            \\
            & \plet \; \_ \pbind \opPut{\ell}(w - w') \; \pin \; \pret{x'}
        \end{aligned}
    \end{aligned}
    \]
    \[
    \begin{aligned}
        Q^{\fwd}_{\opPool{n, m, c}}
        & =
        \pret{\tuple{\fnpool_{n,m,c}(x), x}}
        \\
        Q^{\bwd}_{\opPool{n, m, c}}
        & =
        \pret{\rd[\fnpool](y, z)}
    \end{aligned}
    \]
    The handler $H_{\CNN}$ for backpropagation of CNNs is defined by adding these commands to the handler $H_{\MLP}$ defined in Example~\ref{ex:mlp-backpropagation}.
\end{example}

%To give a denotational semantics for $H_{\CNN}$, we need to interpret the function symbols $\fnconv$ and $\fnpool$.
%In particular, we would like (morally) to interpret $\fnpool$ as the usual max-pooling operation, which is not smooth.
%The same issue arises in Example~\ref{ex:ste}: we would like to interpret the function symbol $\fnround$ as the usual rounding function, which is also not smooth.
%
%One way to address this is to move to a more general category $\PartSmooth$ of partial smooth functions. This is no longer a CRDC, but a \emph{reverse differential restriction category} (RDRC) \cite{Cruttwell+2021}.
%This strategy is in line with the work of Abadi and Plotkin \cite{AbadiPlotkin2019}.
%Like CRDCs, RDRCs come equipped with a reverse differential combinator $\RD$ with the same typing as in the CRDC setting.
%Therefore, once we fix an RDRC together with interpretations of base types and function symbols, we can interpret the syntax introduced in \S\ref{sec:syntax}.
%In the RDRC $\PartSmooth$, we have (for instance) piecewise smooth functions such as the usual max-pooling and rounding functions, which enables us to give the intended interpretation of $H_{\CNN}$ and STE in a precise way:
Now we are ready to give the interpretation of the CNN examples, by giving the interpretations of the function symbols $\fnconv$ and $\fnpool$.
\begin{definition}
    \label{def:interpretation-cnn-fnconv-fnpool}
    Let $\catx = \PartSmooth$, $\BType = \{ \typreal(n) \mid n \in \Nat \}$.
    Let $\itp{\blank}_{\BType} \colon \BType \to \obj(\catx)$ be an interpretation defined by $\itp{\typreal(n)}_{\BType} = \Real^n$.
    We define interpretations of the function symbols used in the CNN examples as follows.
    \[
    \begin{aligned}
        & \itp{\fnconv_{n,m,c,c'}} \colon \Real^{cn} \times \Real^{cc'm} \partto \Real^{c'(n - m + 1)}
        \\
        & \itp{\fnconv_{n,m,c,c'}}(x, w) = \left(
            \sum_{k=1}^{c} \sum_{j=1}^{m}
            x_{k, i + j} w_{k, k', j}
        \right)_{k' \in \{ 1, \dots, c' \}, i \in \{ 1, \dots, n - m + 1\} }
        \\
        & \itp{\fnpool_{n,m,c}} \colon \Real^{cn} \partto \Real^{c \lceil n/m \rceil}
        \\
        & \itp{\fnpool_{n,m,c}}(x) = \begin{cases}
            \text{undefined} & \exists k, i, j, j'.\ x_{k,m(i-1)+j} = x_{k,m(i-1)+j'}
            \\
            \displaystyle
            \left(
                \max_{j \in \{1, \dots, m\}} x_{k, m(i-1) + j}
            \right)_{\substack{k \in \{ 1, \dots, c \}\\ i \in \{ 1, \dots, \lceil n/m \rceil\}} }
            & \text{otherwise}
        \end{cases}
%        \label{eq:interpretation-fnpool}
    \end{aligned}
    \]
\end{definition}
%However, the product in an RDRC is not a categorical product but a \emph{restriction product}, which does not satisfy the usual universal property of products.
%This may cause difficulties in establishing soundness and adequacy of the semantics: we would need to compare this interpretation with our definitions and proofs of soundness and adequacy, which might in turn require modifying the operational semantics in \S\ref{sec:operational-semantics} to make
%it compatible.
%This requires more careful analysis, and we leave it for future work.
%
%Another possible approach is to interpret $\fnpool$ using a smooth approximation of the max-pooling function, based on smooth approximations of the $\max$ function (often called \emph{smooth maximum} functions) like log-sum-exp \cite{BoydVandenberghe2004}.
%This fits entirely within our current CRDC setting, yielding a ``soft'' interpretation of $H_{\CNN}$ and STE.
%However, this may depart from the intended meaning of $H_{\CNN}$ and STE, and it is not clear which smooth approximation is most appropriate for the intended applications;
%addressing this would require another sort of careful analysis.
%You may think of this approach faces the same difficulty as the ordinary approach using RDRCs, but this framework yields more flexibility: we can control the \emph{depth} of the approximation by adding appropriately adding formal symbols in $\signatopr$ and setting handlers for them.
\end{toappendix}

\bibliographystyle{ACM-Reference-Format}
\bibliography{reverse-handle}
\end{document}